\documentclass[12pt]{article}
\usepackage{geometry}
\usepackage{authblk}
\usepackage{graphicx}
\usepackage{enumerate}
\usepackage{natbib}
\usepackage{bm}
\usepackage[shortlabels]{enumitem}
\usepackage{setspace}
\usepackage{amssymb}
\usepackage{mathtools}

\usepackage{algorithm}
\usepackage{algpseudocode}
\usepackage{amsmath}

\usepackage{mathrsfs}
\usepackage{amsthm}

\usepackage{adjustbox}
\usepackage{here}
\usepackage[american]{babel}
\usepackage{threeparttable}
\usepackage{verbatim}
\usepackage{siunitx}
\usepackage{capt-of}

\usepackage{url} 
\usepackage{multirow}

\usepackage{booktabs}
\usepackage{epstopdf}
\newtheorem{definition}{Definition}
\newtheorem{remark}{Remark}
\newtheorem{theorem}{Theorem}
\newtheorem{proposition}{Proposition}
\newtheorem{corollary}{Corollary}
\newtheorem{lemma}{Lemma}

\newtheorem{example}{Example}

\newcommand{\TTD}{\Tilde{\Tilde{\boldsymbol{D}}}}
\newcommand{\TTE}{\Tilde{\Tilde{\boldsymbol{\epsilon}}}}

\usepackage[colorlinks=true,citecolor=blue,bookmarks=false,linkcolor=blue]{hyperref}

\usepackage{array}

\usepackage{etoolbox}

\makeatletter
\patchcmd{\@makecaption}
  {\parbox}
  {\advance\@tempdima-\fontdimen2} 
  {}{}
\makeatother

\begin{document}

\title{On the instrumental variable estimation with many weak and invalid instruments}
\author[a]{Yiqi Lin}
\author[b]{Frank Windmeijer}
\author[a]{Xinyuan Song}
\author[c]{Qingliang Fan\thanks{\noindent Correspondence: Qingliang Fan, Department of Economics, The Chinese University of Hong Kong, Shatin, N.T., Hong Kong. Email: \nolinkurl{michaelqfan@gmail.com}. }}

\affil[a]{Department of Statistics, The Chinese University of Hong Kong}
\affil[b]{Department of Statistics, University of Oxford}
\affil[c]{Department of Economics, The Chinese University of Hong Kong}

\maketitle

\begin{abstract}  

We discuss the fundamental issue of identification in linear instrumental variable (IV) models with unknown IV validity. With the assumption of the ``sparsest rule'', which is equivalent to the plurality rule but becomes operational in computation algorithms, we investigate and prove the advantages of non-convex penalized approaches over other IV estimators based on two-step selections, in terms of selection consistency and accommodation for individually weak IVs. Furthermore, we propose a surrogate sparsest penalty that aligns with the identification condition and provides oracle sparse structure simultaneously. Desirable theoretical properties are derived for the proposed estimator with weaker IV strength conditions compared to the previous literature. Finite sample properties are demonstrated using simulations and the selection and estimation method is applied to an empirical study concerning the effect of {  BMI on diastolic blood pressure}.
\end{abstract}

\section{Introduction}\label{sec:intro}

Recently, estimation of causal effects with high-dimensional observational data has drawn much attention in many research fields such as economics, epidemiology and genomics. The instrumental variable (IV) method is widely used when the treatment variable of interest is endogenous.
As shown in Figure \ref{fig:figure1label}, the ideal  IV needs to be correlated with the endogenous treatment variable (C1), it should not have a direct effect on the outcome (C2) and should not be related to unobserved confounders that affect both outcome and treatment (C3).

\begin{figure}[htp]
\centering
\includegraphics[width = 6.5cm]{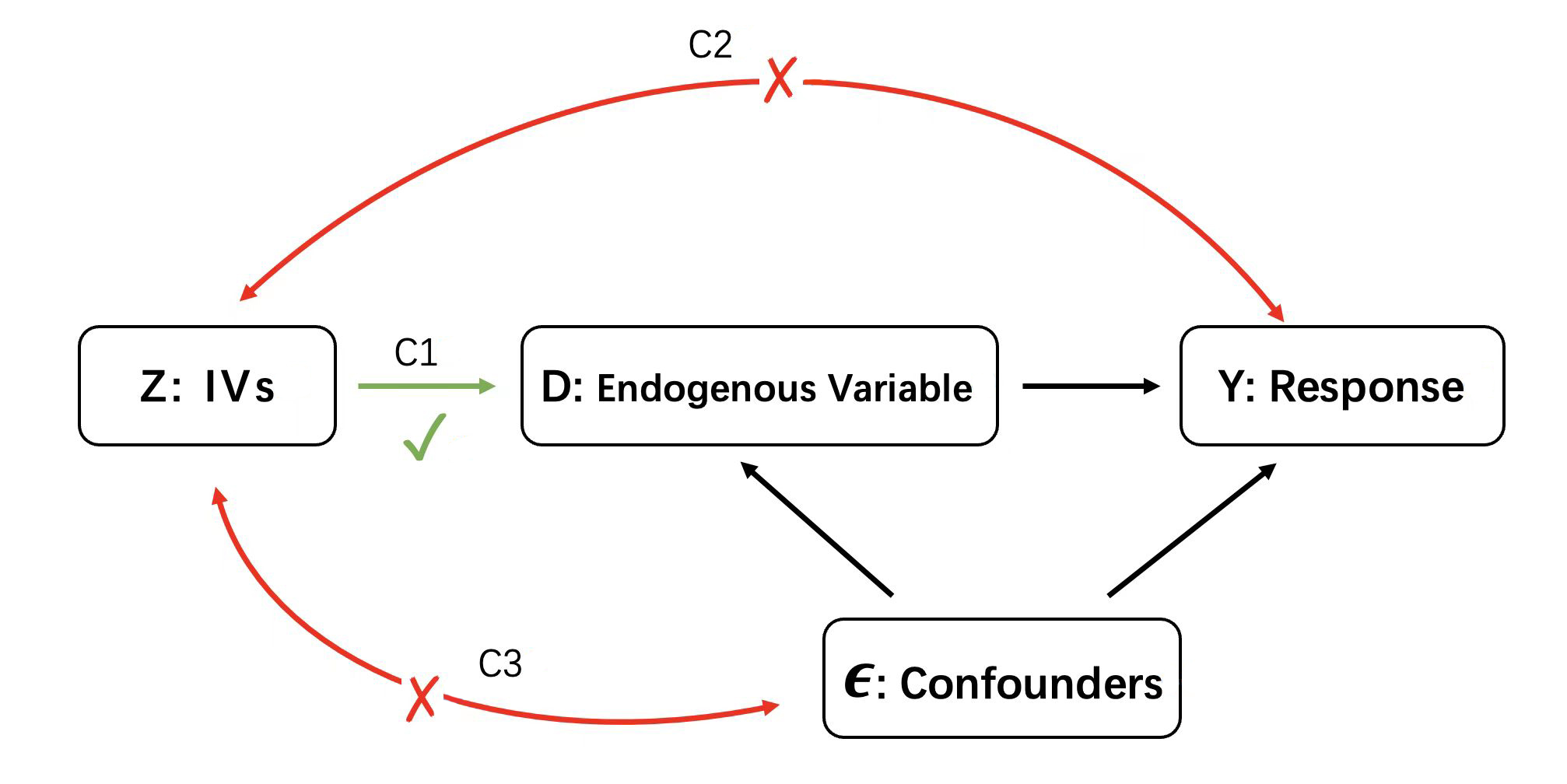}
\caption{Relevance and Validity of IVs}
\label{fig:figure1label}
\end{figure}

Our research is motivated by the difficulty of finding IVs that satisfy all the above conditions. In applications, invalid IVs (violation of C2 or C3) \citep{davey2003mendelian,kang2016instrumental,windmeijer2019use} and weak IVs {  (concerning the weak correlation in C1)} \citep{Bound1995,staiger1997instrumental} are prevalent. A strand of literature studies the ``many weak IVs" problem \citep{Stock2002,chao2005consistent}. With the increasing availability of large datasets, IV models are often high-dimensional \citep{Belloni2012sparse,lin2015,fan2018nonparametric}, and have potentially weak IVs \citep{Andrews2019}, and invalid IVs \citep{guo2018confidence,windmeijer2019confidence}. Among those problems, we mainly focus on the invalid IV problem, while allowing for potential high-dimensionality and weak signals.

\subsection{Related Works}
Most related works fall into two main categories: robust estimation with invalid IVs and estimation which can select valid IVs without any prior knowledge of validity. The first strand of literature allows all IVs to be invalid. For example, \cite{kolesar2015identification} restricted the direct effects of IVs on treatment and response being random effects and independent. In practice, this assumption might be difficult to justify. \cite{lewbel2012using,tchetgen2021genius,guo2022causal} utilized conditional heteroskedasticity or heterogeneous curvatures to achieve robustness with potentially all IVs invalid. However, their performances are not satisfactory once the identification condition is not evident.

The second strand focused on unknown invalid IVs, while imposing certain identification conditions on the number of valid IVs. \cite{kang2016instrumental} proposed a Lasso type estimator (sisVIVE).  \cite{windmeijer2019use} pointed out the inconsistent variable selection of sisVIVE under a relative IV strength condition and proposed an adaptive Lasso estimator, which has asymptotic oracle properties under the assumption that more than half of the IVs are valid, also called the majority rule.  \cite{guo2018confidence,windmeijer2019confidence} further developed two-step (the first one for relevance, the second one for validity) selection approaches, Two-Stage Hard Thresholding (TSHT) and Confidence Intervals IV (CIIV), respectively, under the plurality rule conditional on the set of relevant IVs. The plurality rule states that the valid IVs form the largest group. However, the approaches mentioned above are not robust to many weak IVs due to the restriction of the majority/plurality rule amongst the strong IVs instead of all IVs. Our method closely follows this strand of literature. Instead of a two-step selection, we require the plurality rule for valid IVs for a one-step selection procedure, thus considerably relaxing the requirement of valid IVs in theory and most practical scenarios.

The study of many (weak) IVs originated from empirical motivations but often assumed known validity. For example, \cite{staiger1997instrumental,hansen2008estimation,newey2009generalized,hansen2014instrumental} considered different estimators for situations with many (weak) valid IVs but fixed the number of known covariates. \cite{kolesar2015identification,kolesar2018minimum} allowed the number of covariates to grow with the sample size. { \cite{seng2022structural} introduced an alternative model averaging method to handle weak IVs in low- and high-dimensional settings.} We consider the weak IV issues that are prevalent in empirical studies.
\subsection{Main Results and Contributions}
We propose a {\bf{W}}eak and {\bf{I}}nvalid IV robust {\bf{T}}reatment effect (WIT) estimator. The sparsest rule is sufficient for identification and is operational in numerical optimization. The proposed procedure has a selection stage (regarding IV validity) and a post-selection estimation stage. The selection stage is a penalized IV-regression via minimax concave penalty (MCP, \citealp{zhang2010nearly}), a proper surrogate penalty aligned with the identification condition to achieve model selection consistency of valid IVs under much weaker technical conditions than existing methods \citep{guo2018confidence, windmeijer2019confidence}. In the estimation stage, we utilize the limited information maximum likelihood (LIML) estimator to handle the weak IVs  \citep{staiger1997instrumental}. An efficient computational algorithm for the optimal solution is provided. {{The computer codes for implementing the WIT estimator are available at \url{https://github.com/QoifoQ/WIT}.}}

The key contributions of this paper are summarized as follows.
\begin{enumerate}[itemsep=2pt,topsep=0pt,parsep=0pt]
\item We provide a self-contained framework to investigate the fundamental problem in model identification for linear IV models with unknown validity of instruments. Specifically, we study the identification condition from the general data generating process (DGP) framework. {Furthermore, we discuss the alignment of model identification and variable selection regarding IV validity, which requires a non-convex penalty function. } 

\item This study extends the IV estimation with unknown invalid IVs (namely, \citealp{kang2016instrumental,guo2018confidence,windmeijer2019use,windmeijer2019confidence}) to allow for many potentially weak IVs. We show that the sparsest rule, equivalent to the plurality rule on \emph{the whole IV set}, could accommodate weak IVs in empirically relevant scenarios. Furthermore, we revisit the penalized approaches using the sparsest rule and propose a concept of proper surrogate sparsest penalty that targets identification conditions and provides sparse structure. We propose to deploy MCP as a surrogate sparsest penalty and ensure the targeted solution is the global minimizer.  On the contrary, the existing methods \citep{kang2016instrumental,windmeijer2019use} do not fit the surrogate sparsest penalty and hence are mistargeting the model identification.

\item  Our method is a one-step valid IV selection instead of the previous sequential two-step selections \citep{guo2018confidence,windmeijer2019confidence}. This allows us to utilize individually weak IVs instead of discarding them. We provide theoretical foundations to ensure the compatibility of weak IVs under a mild minimal signal condition. Formally, we establish the selection consistency of the valid IV set, the consistency, and asymptotic normality of the proposed treatment effect estimator under many potentially invalid and weak IVs, where both the number of valid and invalid IVs are increasing with the sample size $n$. We also provide the theoretical results for the case of a fixed and finite number of IVs. {  Our model accomodates different rates for IV validity violations which are illustrated through representative low- and high-dimensional cases. }\end{enumerate}

The article is organized as follows. In Section \ref{sec:model}, we describe the model with some invalid IVs and analyze identification conditions in a general way.  In Section \ref{sec:Theorem}, we present the methodology and the novel WIT estimator. We establish the theorems to identify the valid IVs, estimation consistency, and asymptotic normality.  Section \ref{sec:sim} shows the finite sample performance of our proposed estimator using comprehensive numerical experiments. Section \ref{sec:real} applies our methods to  {  study the effect of BMI on diastolic blood pressure using Mendelian Randomization}. Section \ref{sec:con} concludes. All the technical details and proofs are provided in the appendix.

\section{Model and Identification Strategy}\label{sec:model}
\subsection{Potential Outcome Model with Some Invalid IVs}\label{sec:model_1}

{ 
For $i =1,2,\ldots,n$,  we have the random sample $(Y_i,D_i,\boldsymbol{Z}_{i.})$, where $Y_i\in \mathbb{R}^1$ is the outcome variable, $D_i\in\mathbb{R}^1$ is the (endogenous) treatment variable and $\boldsymbol{Z}_i\in\mathbb{R}^p$ are the potential IVs. Following the same model setting as in \cite{small2007sensitivity,kang2016instrumental,guo2018confidence,windmeijer2019confidence,windmeijer2019use}, we consider a linear functional form between treatments $D_i$ and instruments $\boldsymbol{Z}_{i.}$ as the first-stage specification; meanwhile, a linear exposure of $Y_i$ and $D_i$ and $\boldsymbol{Z}_i$ is assumed as follows:
\begin{equation}
\begin{aligned}
Y_{i} &=D_{i} \beta^*+\boldsymbol{Z}_{i .}^{\top} \boldsymbol{\alpha}^*+\epsilon_{i}, \\
D_i &= \boldsymbol{Z}_{i.}^{\top}\boldsymbol{\gamma}^*+\eta_i.\label{Structure}
\end{aligned}
\end{equation}
where $\epsilon_{i},\eta_{i}$ are random errors. 
\begin{remark}
    Assuming a homogeneous treatment effect (denoted as $\beta^*$) among subjects simplifies the identification problem in instrumental variable analysis. 
\end{remark}
Following  \cite{kang2016instrumental}, we define the valid instruments as follows,
\begin{definition}
For $j = 1,\ldots,p$, the $j$-th instrument is valid if $\alpha^*_j = 0$.
\end{definition}
The validity of the $j$-th IV is quantified by $\alpha^*_j$, which captures the direct effect of the potential IV $\boldsymbol{Z}_j$ on the outcome (C1) as well as its influence on unmeasured confounders (C2). More details can be found in  \cite{kang2016instrumental}.}
 Further, we define the valid IV set $\mathcal{V}^* = \{j: \alpha^*_j = 0\}$ and invalid IV set  $\mathcal{V}^{c*} = \{j: \alpha^*_j \neq 0\}$. Let  ${p_{\mathcal{V}^*}} = |\mathcal{V}^*|$, ${p_{\mathcal{V}^{c*}}} = |\mathcal{V}^{c*}|$ and $p = {p_{\mathcal{V}^{*}}}+{p_{\mathcal{V}^{c*}}}$. Notably, ${p_{\mathcal{V}^*}} \geq 1$ refers to the existence of an excluded IV, thus satisfying the order condition \citep{wooldridge2010econometric}.  Let the $n \times p$ matrix of observations on the instruments be denoted by $\boldsymbol{Z}$, and the $n$-vectors of outcomes and treatments by $\boldsymbol{Y}$ and $\boldsymbol{D}$,  respectively.
  We consider the cases of many and weak IVs in \eqref{Structure} and make the following model assumptions:\\
{\bf{Assumption}} 1 (Many valid and invalid IVs): $p<n$, ${p_{\mathcal{V}^{c*}}}/n \rightarrow \upsilon_{p_{\mathcal{V}^{c*}}}+o(n^{-1/2})$ and ${p_{\mathcal{V}^*}}/n\rightarrow \upsilon_{p_{\mathcal{V}^*}}+o(n^{-1/2})$ for some non-negative constants $\upsilon_{p_{\mathcal{V}^{c*}}}$ and $\upsilon_{p_{\mathcal{V}^*}}$ such that $0\leq\upsilon_{p_{\mathcal{V}^*}}+\upsilon_{p_{\mathcal{V}^{c*}}}<1$.\\
{\bf{Assumption}} 2: Assume $\boldsymbol{Z}$ is standardized. It then has full column rank and $\|\boldsymbol{Z}_j\|^2_2\leq n$ for $j = 1,2,\ldots,p$.\\
{\bf{Assumption}} 3: Let $\boldsymbol{u}_{i}=\left(\epsilon_{i}, \eta_{i} \right)^{\top}$. $\boldsymbol{u}_i\mid \boldsymbol{Z}_i$ are i.i.d. and follow a multivariate normal distribution with mean zero and positive definite covariance matrix $ \boldsymbol{\Sigma} = \left( \displaystyle{\begin{smallmatrix}
              \sigma_\epsilon^2 & \sigma_{\epsilon,\eta} \\ \sigma_{\epsilon,\eta} & \sigma_\eta^2
           \end{smallmatrix}}  \right)  $. The elements of $\boldsymbol{\Sigma}$ are finite and $\sigma_{\epsilon,\eta}\neq 0$.\\
{\bf{Assumption}} 4 (Strength of valid IVs): The concentration parameter $\mu_n$ grows at the same rate as $n$, i.e., $\mu_n \coloneqq \boldsymbol{\gamma}^{*\top}_{\boldsymbol{Z}_{\mathcal{V}^*}}\boldsymbol{Z}_{\mathcal{V}^*}^{\top}M_{\boldsymbol{Z}_{\mathcal{V}^{c*}}}\boldsymbol{Z}_{\mathcal{V}^*}\boldsymbol{\gamma}^{*}_{\boldsymbol{Z}_{\mathcal{V}^*}}/\sigma_\eta^2 {\rightarrow} \mu_0 n$, for some $\mu_0>0$.

Assumption 1 is identical to the assumption of many instruments in \cite{kolesar2015identification,kolesar2018minimum}. It relaxes the conventional many IVs assumptions \citep{bekker1994alternative,chao2005consistent} that only allow the dimension of valid IVs ${p_{\mathcal{V}^*}}$ to grow with $n$. Also, it has not been considered in the literature on selecting valid IVs  \citep{kang2016instrumental,guo2018confidence,windmeijer2019confidence,windmeijer2019use}. Assumption 2 is standard for data preprocessing and scaling $\boldsymbol{Z}_j$.
 Assumption 3 follows \cite{guo2018confidence,windmeijer2019confidence} to impose the homoskedasticity assumption and endogeneity of treatment $D_i$. 
{  \begin{remark}
     Under the setting of many valid and invalid IVs, i.e., $\upsilon_{p_{\mathcal{V}^{c*}}}\neq 0$ and $\upsilon_{p_{\mathcal{V}^{c}}}\neq 0$, the homoskedasticity assumption is necessary for the LIML estimator in the estimation stage (see details in Section \ref{sec:Theorem}) to be consistent. In the low-dimensional case, i.e., $\upsilon_{p_{\mathcal{V}^{c*}}} = \upsilon_{p_{\mathcal{V}^{c}}} =0$, we can relax Assumption 3 and use a heteroskedasticity robust version of the TSLS estimator. The normality condition in Assumption 3 is not required for the selection stage; it is required only for the estimation stage (asymptotic results of the embedded LIML estimator). Relaxing the normality assumption could impact the formula of standard errors and its consistent variance estimator \citep{kolesar2018minimum}.
 \end{remark}}
 {Assumption 4 implies a strong identification condition in terms of the concentration parameter \citep{bekker1994alternative,newey2009generalized}.  In the fixed $p$ case, it indicates the presence of a constant coefficient $\gamma_j$, $\exists j$, and the rest of IVs could be weak \citep{staiger1997instrumental}. Specifically, we model weakly correlated IVs as $\gamma = Cn^{-\tau}$, $0<\tau\le 1/2$, which is the ``local to zero'' case \citep{staiger1997instrumental}. Essentially this is a mixture of constant $\gamma$-type and asymptotically diminishing $\gamma$-type instruments for fixed $p$. For $p\to \infty$ in the same order as $n$, we allow all the IVs to be weak in the ``local to zero'' case with specified rates. This IV strength assumption can be further weakened along the lines of \cite{hansen2008estimation} to have weak identification asymptotics. In this paper, we focus on the individually weak (diminishing to zero as in \citealp{staiger1997instrumental}) signals model in high-dimensionality. Notice our model allows for much weaker individually weak IVs regardless of their validity (as long as the concentration parameter satisfies Assumption 4), unlike that of \cite{guo2018confidence}. Nevertheless, the constant $\mu_0$ can be a small number to accommodate empirically relevant finite samples with many individually weak IVs.}

\subsection{Identifiability of Model \eqref{Structure} \label{subsec:plur}}
The following moment conditions can be derived from Model \eqref{Structure}:
\begin{eqnarray}
   E\Big(\boldsymbol{Z}^{\top}(\boldsymbol{D}-\boldsymbol{Z}\boldsymbol{\gamma}^*)\Big) = \boldsymbol{0}, \quad E\Big(\boldsymbol{Z}^\top(\boldsymbol{Y}-\boldsymbol{D}\beta^*-\boldsymbol{Z}\boldsymbol{\alpha}^*)\Big)=\boldsymbol{0} \quad\Rightarrow\quad \boldsymbol{\Gamma}^* = \boldsymbol{\alpha}^*+ \beta^*\boldsymbol{\gamma}^*, \label{moment}
\end{eqnarray}
where $\boldsymbol{\Gamma}^* = E(\boldsymbol{Z}^\top\boldsymbol{Z})^{-1}E(\boldsymbol{Z}^\top\boldsymbol{Y})$ and $\boldsymbol{\gamma}^* = E(\boldsymbol{Z}^\top\boldsymbol{Z})^{-1}E(\boldsymbol{Z}^\top\boldsymbol{D})$, both are identified by the reduced form models.
Without the exact knowledge about which IVs are valid, \cite{kang2016instrumental} considered the identification of $(\boldsymbol{\alpha}^*,\beta^*)$ via the unique mapping of
\begin{equation}
  \beta_j^* = {\boldsymbol{\Gamma}_j^*}/{\boldsymbol{\gamma}_j^*} = \beta^* + {{\alpha}_j^*}/{{\gamma}_j^*}.\label{betaj}
\end{equation}
Notice that the moment conditions (\ref{moment}) consist of $p$ equations, but  $(\boldsymbol{\alpha}^*,\beta^*) \in \mathbb{R}^{p+1}$ need to be estimated and is hence under-identified without further restrictions. \cite{kang2016instrumental} proposed a sufficient condition, called majority rule (first proposed by \citealp{han2008}), such that ${p_{\mathcal{V}^*}} \geq \lceil{p}/{2} \rceil$, to identify the model parameters without any prior knowledge of the validity of individual IVs. However, the majority rule could be restrictive in practice. \cite{guo2018confidence} further relaxed it to the plurality rule as follows:
\begin{equation}
\text{Plurality Rule:}\quad  \Big|\mathcal{V}^* = \left\{j : {\alpha_{j}^{*}}/{\gamma_{j}^{*}}=0\right\}\Big|>\max_{c \neq 0}\Big|\left\{j : {\alpha_{j}^{*}}/{\gamma_{j}^{*}}=c\right\}\Big|,\label{plurality}
\end{equation}
which was stated as an ``if and only if'' condition of identification of $(\boldsymbol{\alpha}^*,\beta^*)$. 

We re-examine the identifiability problem from the model DGP perspective. Given first-stage information:  $\{\boldsymbol{D},\boldsymbol{Z},\boldsymbol{\gamma}^*\}$, without loss of generality, we denote the { true} DGP with $\{\beta^*, \boldsymbol{\alpha}^*, \boldsymbol{\epsilon}\}$ in (\ref{Structure}) as DGP $\mathcal{P}_0$ that generates $\boldsymbol{Y}$. Given this $\mathcal{P}_0$, for $ j \in \mathcal{V}^{c*}$, we {  obtain an alternative representation}: $ \boldsymbol{Z}_j{\alpha}^*_j = \frac{\alpha^*_j}{\gamma_j^*}(\boldsymbol{D}-\sum_{l\neq j}\boldsymbol{Z}_l\gamma^*_l-\boldsymbol{\eta})$. Denote $\mathcal{I}_c=\{j\in \mathcal{V}^{c*}:c={\alpha^*_j}/{\gamma^*_j}\}$, where $c\neq 0$. For compatibility, we denote $\mathcal{I}_0 = \mathcal{V}^*$. Thus, we can reformulate $\boldsymbol{Y} = \boldsymbol{D}\beta^*+\boldsymbol{Z}\boldsymbol{\alpha}^*+\boldsymbol{\epsilon}$ in (\ref{Structure}) to:
\begin{equation}
 \boldsymbol{Y} = \boldsymbol{D}\tilde{\beta}^{c}+\boldsymbol{Z}\tilde{\boldsymbol{\alpha}}^{c}+\tilde{\boldsymbol{\epsilon}}^{c},\label{procedure}
\end{equation}
where $\{\tilde{\beta}^{c},\tilde{\boldsymbol{\alpha}}^{c},\tilde{\boldsymbol{\epsilon}}^{c}\} = \{\beta^*+c,\boldsymbol{\alpha}^*-c\boldsymbol{\gamma}^*,\boldsymbol{\epsilon}-c\boldsymbol{\eta}\}$, for some $ j \in \mathcal{V}^{c*}$.  Evidently, for different $c\neq 0$, it forms different DGPs $\mathcal{P}_c = \{\tilde{\beta}^{c},\tilde{\boldsymbol{\alpha}}^{c},\tilde{\boldsymbol{\epsilon}}^{c}\}$ that can generate \emph{the same} $\boldsymbol{Y}$ (given $\boldsymbol{\epsilon}$), which also satisfies the moment condition (\ref{moment}) as $\mathcal{P}_0$ since  $E(\boldsymbol{Z}^\top\tilde{\boldsymbol{\epsilon}}^c) = \boldsymbol{0}$. Building on the argument of \cite{guo2018confidence}, Theorem 1, the additional number of potential DGPs satisfying the moment condition (\ref{moment}) is the number of distinguished $c \ne 0$ for $ j \in \mathcal{V}^{c*}$.
  We formally state this result in the following theorem.
\begin{theorem}\label{Theorem 1}
Suppose Assumptions 1-3 hold, given $\mathcal{P}_0$ and $\{\boldsymbol{D},\boldsymbol{Z},\boldsymbol{\gamma}^*,\boldsymbol{\eta}\}$, it can only produce additional $G = |\{c\neq 0: \alpha^*_j/\gamma_j^* = c, ~j \in \mathcal{V}^{c*}\}|$ groups of different $\mathcal{P}_c$  such that $\mathcal{V}^* \cup \{\cup_{c\neq 0}\mathcal{I}_c\} = \{1,2\ldots,p\}$, $\mathcal{V}^* \cap \mathcal{I}_c = \varnothing$ for any $c\neq 0$ and $\mathcal{I}_c \cap \mathcal{I}_{\tilde c} = \varnothing$ for $c \ne \tilde c$, and $E(\boldsymbol{Z}^\top \tilde{\boldsymbol{\epsilon}}^c) = \boldsymbol{0}$. The sparsity structure regarding $\boldsymbol{\alpha}$ is non-overlapping for different solutions.
\end{theorem}

Theorem \ref{Theorem 1} shows there is a collection of model DGPs { (different parametrizations)}
\begin{equation}
    \mathcal{Q} = \{\mathcal{P}=\{\beta,\boldsymbol{\alpha},\epsilon\}: \boldsymbol{\alpha}\,\, \text{is sparse},E(\boldsymbol{Z}^\top\boldsymbol{\epsilon}) = \boldsymbol{0}\} \label{Q:collection of DGPs}
\end{equation}  corresponding to the same observation $\boldsymbol{Y}$ conditional on first-stage information. Given some $\mathcal{P}_0$, there are additional $1 \le G\le {p_{\mathcal{V}^{c*}}}$ equivalent DGPs. All members in $\mathcal{Q}$ are related through the transformation procedure (\ref{procedure}) and $1 <|\mathcal{Q}| = G+1\leq p$. Notably, the non-overlapping sparse structure among all possible DGPs leads to the sparsest model regarding $\boldsymbol{\alpha}^*$ being equivalent to plurality rule $|\mathcal{V}^*|> {\operatorname{max}}_{c \neq 0} |\mathcal{I}_{c}|$ in the whole set of IVs. 

\subsection{The Sparsest ($\boldsymbol{\alpha}$) Rule}
The sparsest rule is conceptually equivalent to the plurality rule on the whole IV set, considering the non-overlapping sparse solutions given by Theorem \ref{Theorem 1}.
To relax the majority rule,
\cite{guo2018confidence} proposed to use the plurality rule based on the relevant IV set:
\begin{equation}
 \Big|\mathcal{V}^*_{\mathcal{S}^*} = \left\{j \in \mathcal{S}^{*}: {\alpha_{j}^{*}}/{\gamma_{j}^{*}}=0\right\}\Big|>\max_{c \neq 0}\Big|\left\{j \in \mathcal{S}^{*}: {\alpha_{j}^{*}}/{\gamma_{j}^{*}}=c\right\}\Big|,
\end{equation}
where $\mathcal{S}^*$ is the set of strong IVs estimated by $\hat{\mathcal{S}}$ via first-step hard thresholding. 
Thus, TSHT  and  CIIV  explicitly leverage the $\hat{\mathcal{S}}$-based plurality rule to estimate  $\mathcal{V}^*_{\hat{\mathcal{S}}}$ and $\beta^*$.

In contrast to earlier studies on invalid IVs, our approach leverages the information from weak IVs. Firstly, weak IVs can be employed to estimate $\beta^*$. In situations where strong IVs are not available, weak IV robust estimators such as LIML prove useful \citep{Andrews2019}. Secondly, weak IVs can aid in the identification of the valid IVs set, as demonstrated in Theorem \ref{Theorem 3}. When weak IVs are present, the plurality rule applied after the initial selection of strong instruments may yield unstable estimates of $\mathcal{V}^*$, as exemplified in the following scenario.

\begin{example}[weak and invalid IVs]\label{Example 1}
 Let $\boldsymbol{\gamma}^* = (\boldsymbol{0.04}_3,\boldsymbol{0.5}_2,0.2,\boldsymbol{0.1}_4)^\top$ and $\boldsymbol{\alpha}^* = (\boldsymbol{0}_5,1,\boldsymbol{0.7}_4)^\top$.  There are therefore three groups: $\mathcal{I}_0 = \mathcal{V}^* = \{1,2,3,4,5\}$, $\mathcal{I}_{5} = \{6\}$, $\mathcal{I}_{7} = \{7,8,9,10\}$ and plurality rule $|\mathcal{I}_0|>{\operatorname{max}}_{c= 5,7} |\mathcal{I}_{c}|$ holds in the whole IVs set. This setup satisfies the individually weak IVs in fixed $p$ (Discuss later in Corollary \ref{Corollary 2}).
 \begin{figure}[tbp]
\centering
\begin{minipage}[t]{0.45\textwidth}
\centering
\includegraphics[width=7cm]{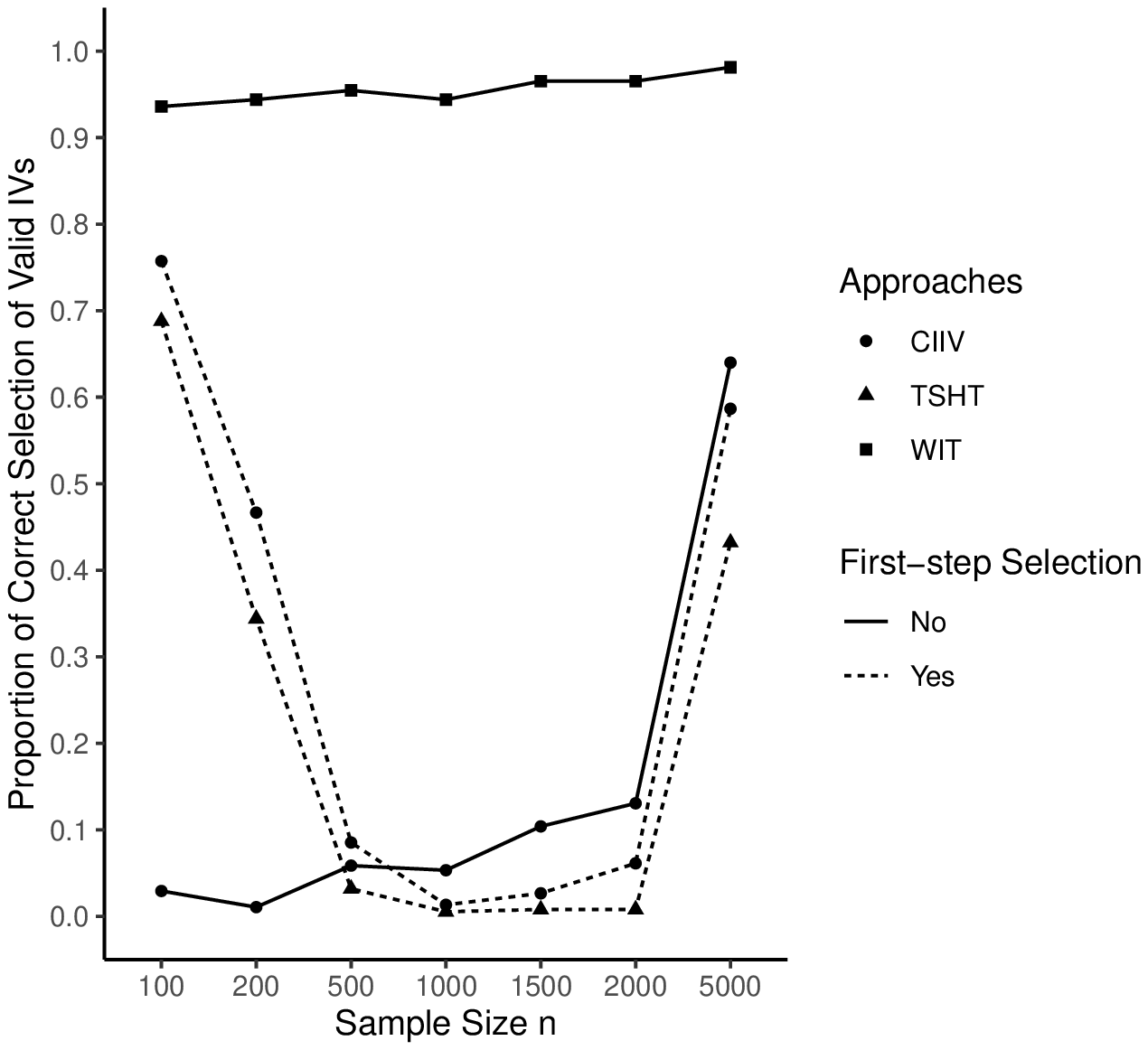}
\caption{The proportion of correct selection of (subset) valid IVs based on 500 replications on each sample size.}
\label{fig 2}
\end{minipage}
~~~
\begin{minipage}[t]{0.45\textwidth}
\centering
\includegraphics[width=7cm]{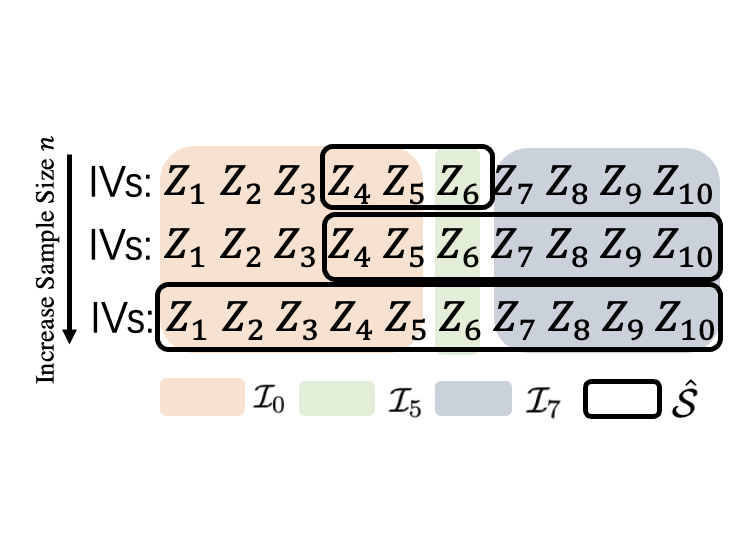}
\caption{Illustration of Plurality rule based on first-step selection.}
\label{fig 3}
\end{minipage}
\end{figure}
Fig. \ref{fig 2} shows that the selection of valid IVs by CIIV and TSHT breaks down in finite samples, e.g., for $n \in [500,2000]$. This is because the solution of the plurality rule after first-stage selection in the finite sample (which may not hold in practice even though it is in theory) is quite sensitive to IV strength and sample sizes. On the other hand, weak IVs also deteriorate the performance of CIIV without first-step selection. {Notably, the proposed WIT estimator significantly outperforms others.} Fig. \ref{fig 3} demonstrates the relevant set $\mathcal{S}^*$ selected by plurality rule-based TSHT and CIIV. It clearly shows $\hat{ \mathcal{S}}$ is unstable and changes with sample size, even though the plurality rule holds in the whole IV set.
\end{example}

The mixture of weak and invalid IVs is ubiquitous in practice, especially in many IVs cases. For the sake of using all instruments' information for estimating $\beta^*$ and identification of $\mathcal{V}^*$, we allow some individual IV strengths to be local to zero  \citep{chao2005consistent}, say $\gamma^*_j\rightarrow 0$, or a small constant that cannot pass the first-stage threshold \citep{guo2018confidence} unless with a very large sample size. However, we can see that in (\ref{betaj}), plurality rule-based methods that rely on first-stage selection are problematic, since $\mathcal{I}_0 = \{j :{\alpha^*_j}/{\gamma_j^*} = 0\}$ is ill-defined asymptotically due to the problem of ``${0}/{0}$'' if $\gamma_j^*$ is local to zero.

For using weak IVs information and improving finite sample performance, we turn to the sparsest rule that is also operational in computation algorithms. From the multiple DGPs $\mathcal{Q}$, recall $\mathcal{P}_c = \{\tilde{\beta}^{c},\tilde{\boldsymbol{\alpha}}^{c},\tilde{\boldsymbol{\epsilon}}^{c}\} = \{\beta^*+c,\boldsymbol{\alpha}^*-c\boldsymbol{\gamma}^*,\boldsymbol{\epsilon}-c\boldsymbol{\eta}\}$, where $\tilde{\boldsymbol{\alpha}}^c_{\mathcal{I}_c} = \boldsymbol{0}$. For other elements in $\tilde{\boldsymbol{\alpha}}^c$ (corresponding to a different DGP in $\mathcal{Q}$) and  $j \in \{j:\alpha^*_j/\gamma^*_j = \tilde{c}\neq c\}$, we obtain
\begin{equation}
    |\tilde{{\alpha}}^c_j| = |\alpha^*_j-c\gamma_j^*| = |\alpha^*_j/\gamma^*_j-c| \cdot |\gamma^*_j| = |\tilde{c}-c| \cdot |\gamma^*_j|.\label{GAP}
\end{equation}

The above $|\tilde{\alpha}^c_j|$ needs to be distinguished from $0$ on the ground of the non-overlapping structure stated in Theorem \ref{Theorem 1}. To facilitate the discovery of all solutions in $\mathcal{Q}$, we assume: \\
{\bf{Assumption}} 5 (Separation Condition): $|\boldsymbol{\alpha}_{\mathcal{V}^{c*}}^*|_{\text{min}}>\kappa(n)$
 and $|\tilde{\alpha}_j^c|>\kappa^{c}(n)$ for $j \in \{j:\alpha^*_j/\gamma^*_j = \tilde{c}\neq c\}$, where $\kappa(n)$ and $\kappa^c(n)$ are a generally vanishing rate specified by a particular estimator.

The conditions described above are comparable to the ``beta-min" condition \citep{van2009conditions,loh2015regularized}. {  The rates $\kappa(n)$ and $\kappa^c(n)$ will be further specified in Theorem \ref{Theorem 3}. We aim to facilitate the understanding of these technical conditions by providing practical examples. In Section 3.3, we introduce two commonly encountered DGPs in both low and high dimensions. These examples demonstrate the validity of Assumption 5 without delving into the intricate specifics of the separation (beta-min) condition. 

\begin{remark}\label{Remark 3}
    Notably, as shown in (\ref{GAP}), $|\tilde{\alpha}^c_j| = |\tilde{c}-c|\cdot|\gamma^*_j|>\kappa^c(n)$ depends on the product of $|\tilde{c}-c|$ and $|\gamma_j^*|$. As discussed in \cite{guo2018confidence}, $|\tilde{c}-c|$ cannot be too small to separate different solutions in $\mathcal{Q}$, and a larger gap $|\tilde{c}-c|$ is helpful to mitigate the problem of small or local to zero $|\gamma_j^*|$ in favor of our model.
\end{remark}
}

Hence, the identification condition known as the sparsest rule is formally defined as\\
{\bf{Assumption}} 6: {(The Sparsest Rule)}: $\boldsymbol{\alpha}^* = {\operatorname{arg min}}_{\mathcal{P}=\{\beta,\boldsymbol{\boldsymbol{\alpha}},\boldsymbol{\epsilon}\}\in \mathcal{Q}} \,\,\|\boldsymbol{\alpha}\|_0.$
\addtocounter{example}{-1}
\begin{example}[continued]\label{Example 2}
Following the procedure (\ref{procedure}), we are able to reformulate two additional solutions of \eqref{moment} given the DGP of Example 1, $\boldsymbol{\alpha}^* = (\boldsymbol{0}_5,1,\boldsymbol{0.7}_4)^\top$: $\tilde{\boldsymbol{\alpha}}^5 = (\boldsymbol{-0.2}_3,\boldsymbol{-2.5}_2,0,\boldsymbol{0.2}_4)^\top$ and $\tilde{\boldsymbol{\alpha}}^7 = (\boldsymbol{-0.28}_3,\boldsymbol{-3.5}_2,-0.4,\boldsymbol{0}_4)^\top$. Thus, the sparsest rule ${\operatorname{arg min}}_{\boldsymbol{\alpha} \in \{\boldsymbol{\alpha}^*,\tilde{\boldsymbol{\alpha}}^5, \tilde{\boldsymbol{\alpha}}^7\}} \|\boldsymbol{\alpha}\|_0$ picks $\boldsymbol{\alpha}^*$ up, and Assumption 5 is easy to satisfy since fixed minimum absolute values except $0$ are $0.7,0.2,0.28$ in $\boldsymbol{\alpha}^*,\tilde{\boldsymbol{\alpha}}^5,\tilde{\boldsymbol{\alpha}}^7$, respectively. This example shows the first-stage signal should not interfere with the valid IV selection in the structural form equation in \eqref{Structure}, as long as the first-stage has  sufficient information (concentration parameter requirement in Assumption 4). Therefore, the most sparse rule using the whole IVs set is desirable. It is also shown to be stable in numerical studies. The detailed performance of the proposed method under this example refers to Case 1(II) in Section \ref{sec:simu1}.
\end{example}

In the next subsection, we reexamine the penalized approaches by \cite{kang2016instrumental} and \cite{windmeijer2019use}, and discuss a class of existing estimators concerning penalization, identification, and computation. We also explore the general penalization approach that aligns model identification with its objective function.

 \subsection{Penalization Approaches with Embedded Surrogate Sparsest Rule}

A Lasso penalization approach was first used in the unknown IV validity context by \cite{kang2016instrumental}. We extend this to a general formulation and discuss the properties of different classes of penalties.

 Consider a general penalized estimator based on moment conditions \eqref{moment},
\begin{equation}
     (\widehat{\boldsymbol{\alpha}}^{\text{pen}},\hat{\beta}^{\text{pen}}) = \underset{\boldsymbol{\alpha},\beta}{\operatorname{arg min}}\,\,\,  \underbrace{\frac{1}{2n}\|P_{\boldsymbol{Z}}(\boldsymbol{Y}-\boldsymbol{Z}\boldsymbol{\alpha}-\boldsymbol{D}\beta)\|_2^2}_{(I)}+ \underbrace{p_\lambda^{\text{pen}}(\boldsymbol{\alpha})}_{(II)}.\label{problem1}
\end{equation}
where $p_\lambda^{\text{pen}}(\boldsymbol{\alpha}) = \sum_{j=1}^{p}p_\lambda^{\text{pen}}(\alpha_j)$ and  $p_\lambda^{\text{pen}}(\cdot)$ %
 is a general penalty function with tuning parameter $\lambda > 0$ and ${p_\lambda^{\text{pen}}}^{\prime}(\cdot)$ is its derivative that satisfy:  $\underset{x \rightarrow 0^+}{\operatorname{lim}}\, {p_\lambda^{\text{pen}}}^{\prime}(x) = \lambda$,  $p_\lambda^{\text{pen}}(0) = 0$, $p_\lambda^{\text{pen}}(x) =p_\lambda^{\text{pen}}(-x) $,  $(x-y)(p_\lambda^{\text{pen}}(x)-p_\lambda^{\text{pen}}(y))\geq0$, and ${p_\lambda^{\text{pen}}}^{\prime}(\cdot)$ is continuous on $(0,\infty)$.

 In the RHS of (\ref{problem1}), (I) and (II) correspond to two requirements for the collection of valid DGPs in $\mathcal{Q}$ defined in (\ref{Q:collection of DGPs}). $(I)$ is a scaled finite sample version of $E\Big((\boldsymbol{Z}^\top  \boldsymbol{\epsilon})^\top (\boldsymbol{Z}^\top\boldsymbol{Z})^{-1}(\boldsymbol{Z}^\top  \boldsymbol{\epsilon})\Big)$, which is a $(\boldsymbol{Z}^\top\boldsymbol{Z})^{-1}$ weighted quadratic term of condition $E(\boldsymbol{Z}^\top\boldsymbol{\epsilon}) = \boldsymbol{0}$, and $(II)$ is imposed to ensure sparsity structure in $\widehat{\boldsymbol{\alpha}}$.

Further, regarding $(I)$, one can reformulate (\ref{problem1}) with respect to $\widehat{\boldsymbol{\alpha}}^{\text{pen}}$ as
\begin{equation}
     \widehat{\boldsymbol{\alpha}}^{\text{pen}} = \underset{\boldsymbol{\alpha}}{\operatorname{arg min}}\,\,\, \frac{1}{2n}\|\boldsymbol{Y}-{\tilde{\boldsymbol{Z}}}\boldsymbol{\alpha}\|_2^2+p_\lambda^{\text{pen}}(\boldsymbol{\alpha}),
\label{seqid}
\end{equation}
where $\tilde{\boldsymbol{Z}} = M_{\widehat{\boldsymbol{D}}}\boldsymbol{Z}$ and $\widehat{\boldsymbol{D}} = P_{\boldsymbol{Z}}\boldsymbol{D} = \boldsymbol{Z}\widehat{\boldsymbol{\gamma}}$, where $\widehat{\boldsymbol{\gamma}}$ is the least squares estimator of $\boldsymbol{\gamma}$, see \cite{kang2016instrumental}. The design matrix $\tilde{\boldsymbol{Z}}$ is rank-deficient with rank $p-1$ since $\tilde{\boldsymbol{Z}}\widehat{\boldsymbol{\gamma}} = 0$. However, we show that it does not affect the $\boldsymbol{\alpha}$ support recovery using a proper penalty function.  {{On the other hand, $\tilde{\boldsymbol{Z}}$ is a function of $\boldsymbol{\eta}, \boldsymbol{\gamma}^*$ and $\boldsymbol{Z}$, hence is correlated with $\boldsymbol{\epsilon}$. This inherited endogeneity initially stems from $\widehat{\boldsymbol{D}}$, in which $E(\widehat{\boldsymbol{D}}^\top \boldsymbol{\epsilon}) = \sigma_{\epsilon,\eta}^2 p/n$ does not vanish in the many IVs model (Assumption 1).
The following lemma implies that the level of endogeneity of each $\tilde{\boldsymbol{Z}_j}$ is limited.
}}
\begin{lemma}\label{Lemma 1}
 Suppose Assumptions 1-4 hold and denote average gram matrix $\boldsymbol{Q_n} = \boldsymbol{Z}^\top \boldsymbol{Z}/n$. The endogeneity level of the $j$-th transformed IV $\tilde{\boldsymbol{Z}_j}$  follows
\begin{equation}
     \tilde{\boldsymbol{Z}}_j^\top \boldsymbol{\epsilon}/n
     = \underbrace{\sigma_{\epsilon,\eta}^2 p/n}_{E(\widehat{\boldsymbol{D}}^\top \boldsymbol{\epsilon}/n)}\cdot \underbrace{\frac{\boldsymbol{Q}_{nj}^\top \boldsymbol{\gamma}^*}{\boldsymbol{\gamma}^{*\top}\boldsymbol{Q}_{n}\boldsymbol{\gamma}^*+\sigma_{\eta}^2 p/n}}_{\text{dilution weight}} +O_p(n^{-1/2}).
\end{equation}
\end{lemma}

\begin{remark}
Under Assumption 1, $p/n\rightarrow\upsilon_{p_{\mathcal{V}^*}}+\upsilon_{p_{\mathcal{V}^{c*}}} <1$ does not vanish as $n \to \infty$. The dilution weight is related to $\boldsymbol{Q}_n$ and first-stage signal $\boldsymbol{\gamma}^*$. {{In general the dilution weight is $o(1)$ and hence negligible except for the existence of dominated $\gamma^*_j$.}} However, in the fixed $p$ case, since $p/n\rightarrow0$, the endogeneity of $\tilde{\boldsymbol{Z}}$ disappears asymptotically.
\end{remark}

Concerning $(II)$ in (\ref{problem1}), Theorem \ref{Theorem 1} shows that model \eqref{Structure} can be identified by different strategies with non-overlapping results. On the ground of the sparsest rule assumption, the role of the penalty on $\boldsymbol{\alpha}$, i.e.\ $p_\lambda^{\text{pen}}(\boldsymbol{\alpha})$, should not only impose a sparsity structure but also serve as an objective function corresponding to the identification condition we choose. For example, the penalty $\lambda \|\boldsymbol{\alpha}\|_0$ matches the sparsest rule.

To see the roles of a proper penalty function clearly, we rewrite ($\ref{problem1}$) into an equivalent constrained objective function with the optimal penalty $\|\boldsymbol{\alpha}\|_0$ regarding the sparsest rule:
\begin{equation}
     (\widehat{\boldsymbol{\alpha}}^{\text{opt}},\hat{\beta}^{\text{opt}}) = \underset{\boldsymbol{\alpha},\beta}{\operatorname{arg min}}\,\,\, \|\boldsymbol{\alpha}\|_0 \quad \text{s.t.} \,\, \|P_{\boldsymbol{Z}}(\boldsymbol{Y}-\boldsymbol{D}\beta-\boldsymbol{Z}\boldsymbol{\alpha})\|_2^2<\delta, \label{view of objective function}
\end{equation}
where $\delta$ is the tolerance level which we specify in Section \ref{sec:sim}. The constraint above narrows the feasible solutions into $\mathcal{Q}$ because it aligns with the Sargan test statistic, \citep{sargan1958estimation}, $\|P_{\boldsymbol{Z}}(\boldsymbol{Y}-\boldsymbol{D}\beta-\boldsymbol{Z}\boldsymbol{\alpha})\|_2^2/\|(\boldsymbol{Y}-\boldsymbol{D}\beta-\boldsymbol{Z}\boldsymbol{\alpha})/\sqrt{n}\|_2^2  = O_p(1)$ under the null hypothesis $E(\boldsymbol{Z}^\top \boldsymbol{\epsilon}) = \boldsymbol{0}$ as required in $\mathcal{Q}$; otherwise, the constraint becomes $O_p(n)$ that cannot be bounded by $\delta$.
Thus, a properly chosen  $\delta$ in \eqref{view of objective function} leads to an equivalent optimization problem 
\begin{equation*}
      (\widehat{\boldsymbol{\alpha}}^{\text{opt}},\hat{\beta}^{\text{opt}}) =\underset{{\mathcal{P}=\{\beta,\boldsymbol{\boldsymbol{\alpha}},\boldsymbol{\epsilon}\}\in \mathcal{Q}}}{\operatorname{arg min}}\,\,\, \|\boldsymbol{\alpha}\|_0
\end{equation*}
that matches the identification condition in Assumption 6.
Therefore, the primary optimization object in (\ref{view of objective function}) should also serve as an identification condition: the sparsest rule. 

{  \begin{remark}
    The constrained optimization perspective provides valuable insights into the issue of invalid IVs in \eqref{Structure}. Specifically, the penalty term in \eqref{problem1} not only imposes a sparse structure but also serves as an identification rule differing from penalized linear regressions without an endogeneity issue.
\end{remark}}

Due to computational NP-hardness for $\|\boldsymbol{\alpha}\|_0$ in (\ref{view of objective function}), a surrogate penalty function is needed.  \cite{kang2016instrumental} proposed to replace the optimal $l_0$-norm with Lasso in (\ref{problem1}), denoting their estimator sisVIVE as $(\widehat{\boldsymbol{\alpha}}^{\text{sis}},\hat{\beta}^{\text{sis}})$. And $\mathcal{V}^*$ is estimated as $\hat{\mathcal{V}}^{\text{sis}} = \{j:\widehat{\alpha}_j^{\text{sis}} = 0\}$. However, the surrogate $\ell_1$ penalty brings the following issues: {  (a) Failure in consistent variable selection under some deterministic conditions, namely the sign-aware invalid IV strength (SAIS) condition \citep[Proposition 2]{windmeijer2019use}: 
\begin{equation}
\left|\widehat{\boldsymbol{\gamma}}_{\mathcal{V}^{c*}}^{\top} {\operatorname{sgn}}\left(\boldsymbol{\alpha}^*_{\mathcal{V}^{c*}}\right)\right|>\left\|\widehat{\boldsymbol{\gamma}}_{\mathcal{V}^*}\right\|_{1}  \label{SAIS} 
\end{equation}
The sisVIVE cannot achieve $\mathcal{P}_0$ under the  SAIS condition, which is likely to hold when the invalid IVs are relatively stronger than valid ones. (b) \cite{kang2016instrumental}, Theorem 2, proposed
a non-asymptotic error bound $|\hat{\beta}^{\text{sis}}-\beta^*|$ for sisVIVE. The dependence of the restricted isometry property (RIP) constant of $P_{\widehat{\boldsymbol{D}}}\boldsymbol{Z}$ and the error bound is not clear due to the random nature of $\widehat{\boldsymbol{D}}$. (c) The objective function deviates from the original sparsest rule: $g_1(\mathcal{P}) = \|\boldsymbol{\alpha}\|_0$ and $g_2(\mathcal{P}) = \|\boldsymbol{\alpha}\|_1$ correspond to incompatible identification conditions unless satisfying an additional strong requirement
    \begin{equation}
         \boldsymbol{\alpha}^* = {\operatorname{arg min}}_{\mathcal{P}=\{\beta,\boldsymbol{\boldsymbol{\alpha}},\boldsymbol{\epsilon}\}\in \mathcal{Q}} \,\,g_j(\mathcal{P}), \forall j = 1,2
         \iff \|\boldsymbol{\alpha}^*-c\boldsymbol{\gamma}^*\|_1>\|\boldsymbol{\alpha}^*\|_1, \forall c\neq 0,\label{additional con}
    \end{equation}
    which further impedes sisVIVE in estimating $\beta^* \in \mathcal{P}_0$.
    }

Problems (a) and (b) relate to the Lasso problem within invalid IVs, while (c) exposes a deeper issue beyond Lasso: a proper surrogate penalty in (\ref{seqid}) must align with the identification condition.

\citet{windmeijer2019use} suggested the Adaptive Lasso \citep{zou2006adaptive}, using as an initial estimator the median estimator of \cite{han2008} to tackle problem (a)'s SAIS issue. This solution also addresses (c) concurrently. However, it necessitates a more demanding majority rule and requires all IVs to be strong in the fixed $p$ case. Like TSHT and CIIV, it is sensitive to weak IVs.

The following proposition outlines the appropriate surrogate sparsest penalty.
\begin{proposition}\label{Proposition 1}
{\bf (The proper Surrogate Sparsest penalty)} Suppose Assumptions 1-6 are satisfied. If $p_\lambda^{\text{pen}}(\boldsymbol{\alpha})$ is the surrogate sparsest rule in the sense that it gives sparse solutions and
\begin{equation}
     \boldsymbol{\alpha}^* = \underset{{\mathcal{P}=\{\beta,\boldsymbol{\boldsymbol{\alpha}},\boldsymbol{\epsilon}\}\in \mathcal{Q}}}{\operatorname{arg min}} \|\boldsymbol{\alpha}\|_0 = \underset{{\mathcal{P}=\{\beta,\boldsymbol{\boldsymbol{\alpha}},\boldsymbol{\epsilon}\}\in \mathcal{Q}}}{\operatorname{arg min}} p_\lambda^{\text{pen}}(\boldsymbol{\alpha}),\label{surrogate penalty}
\end{equation}
then $p_\lambda^{\text{pen}}(\cdot)$ must be concave and ${p_\lambda^{\text{pen}}}^{\prime}(t) = O(\lambda \kappa(n))$ for any $t>\kappa(n)$.
\end{proposition}

The surrogate sparsest penalty requirement aligns with the folded-concave penalized method \citep{fan2001variable,zhang2010nearly}. We adopt MCP penalty for our proposed approach in the following. In standard sparse linear regression, MCP necessitates $\lambda = \lambda(n) = O(\sqrt{\log p/n})$ and ${p_\lambda^{\text{MCP}}}^\prime (t)= 0$ for $t>C\lambda$, with $C$ a constant, fulfilling Proposition \ref{Proposition 1}. We further elaborate that this property is applicable in invalid IVs scenarios in the following section, and demonstrate its ability to bypass sisVIVE's shortcomings shown in (a) and (b).

\begin{remark}
Proposition \ref{Proposition 1} mandates concavity for the surrogate penalty, precluding the use of Adaptive Lasso. However, by imposing additional conditions like the majority rule, which identifies $\boldsymbol{\alpha}^*$, the Adaptive Lasso penalty can satisfy (\ref{surrogate penalty}). It is crucial to differentiate the motivation for using concave penalties in the surrogate sparsest penalty from the one for debiasing techniques such as the debiased Lasso \citep{javanmard2018debiasing}. The latter may not fulfill the identification condition, thus potentially altering the objective function (\ref{view of objective function}).
\end{remark}

\section{WIT Estimator}\label{sec:Theorem}

\subsection{Estimation Procedure}\label{sect31} 
We implement the penalized regression framework (\ref{problem1}) and specifically employ a concave penalty in (\ref{seqid}), the MCP, which provides nearly unbiased estimation. The MCP is favored numerically for its superior convexity in penalized loss and its consistent variable selection property without imposing incoherence conditions on the design matrix \citep{loh2017support,feng2019sorted}. This makes MCP more suited to two-stage estimation problems than Lasso. 

The selection stage is formally defined as:
 \begin{equation}
      \widehat{\boldsymbol{\alpha}}^{\text{MCP}} = \underset{\boldsymbol{\alpha}}{\operatorname{arg min}}\,\,\, \frac{1}{2n}\|\boldsymbol{Y}-\boldsymbol{\widetilde{Z}}\boldsymbol{\alpha}\|_2^2+ p_\lambda^{\text{MCP}}(\boldsymbol{\alpha}),\label{Problem3}
 \end{equation}
 where $\widetilde{\boldsymbol{Z}} = M_{\widehat{\boldsymbol{D}}}\boldsymbol{Z}$, and $\widehat{\boldsymbol{D}} = P_{\boldsymbol{Z}}\boldsymbol{D} = \boldsymbol{Z}\widehat{\boldsymbol{\gamma}}$ are the same as in \eqref{seqid}.  $p_\lambda^{\text{MCP}}(\boldsymbol{\alpha}) = \sum_{j = 1}^p p_\lambda^{\text{MCP}}({\alpha}_j) =  \sum_{j = 1}^p\int_0^{|\alpha_j|} \Big(\lambda-{t}/{\rho}\Big)_{+}dt$ is the MCP penalty and $\rho>1$ is the tuning parameter, which also controls convexity level $1/\rho$, and its corresponding derivative is ${p_\lambda^{\text{MCP}}}^{\prime}(t) = (\lambda-{|t|}/{\rho})_+$. Unlike Lasso, MCP imposes no penalty when $|\alpha_j|>\lambda\rho$, resulting in nearly unbiased coefficient estimation and exact support recovery without the SAIS condition (see \ref{sec:A2} for further discussion). Consequently, a consistent estimation of the valid IVs set, i.e., $\hat{\mathcal{V}} = \{ j:\widehat{{\alpha}}^{\text{MCP}}_j = 0 \}$ and $\Pr(\hat{\mathcal{V}} = \mathcal{V}^*) \overset{p}{\rightarrow}1$, is expected to hold under more relaxed conditions. We next illustrate how the WIT blends the benefits of penalized TSLS and LIML estimators at various stages.

The LIML estimator is consistent not only in classic many (weak) IVs \citep{bekker1994alternative,hansen2008estimation}, but also in many IVs and many included covariates \citep{kolesar2015identification,kolesar2018minimum}. However, simultaneous estimation of $\hat{\kappa}_{\text{liml}}$ and $\hat{\mathcal{V}^c}$ in \eqref{k_fuller} below is difficult. In the selection stage \eqref{problem1}, we use the moment-based objective function. If we do not consider the penalty term (II), the moment-based part (I) of \eqref{problem1} coincides with TSLS. Furthermore, the bias in TSLS has a limited effect on consistent variable selections (see Theorem \ref{Theorem 3}). In the estimation step, however, due to LIML's superior finite sample performance and the issue of TSLS in the presence of many (or weak) IVs even when $\mathcal{V}^*$ is known \citep{sawa1969exact,chao2005consistent}, we embed the LIML estimator to estimate $\beta^*$ on the basis of estimated valid IVs set via (\ref{Problem3}). The performance of oracle-TSLS shown in simulations verifies this choice.

Consequently, we proposed the \textbf{W}eak and some \textbf{I}nvalid instruments robust \textbf{T}reatment effect (WIT) estimator in the estimation stage,
\begin{eqnarray}
\left(
\hat{\beta}^{\text{WIT}} ,
\hat{\boldsymbol{\alpha}}_{\boldsymbol{Z}_{\hat{\mathcal{V}}^{c}}}^{\text{WIT}}
\right)^\top=\Big(\left[
\boldsymbol{D},  {\boldsymbol{Z}_{\hat{\mathcal{V}}^{c}}}
\right]^{\top}\left(\boldsymbol{I}-\hat{\kappa}_{\text{liml}} {M}_{{\boldsymbol{Z}}}\right)\left[
\boldsymbol{D},  {\boldsymbol{Z}_{\hat{\mathcal{V}}^{c}}}
\right]\Big)^{-1}\Big(\left[
\boldsymbol{D},  {\boldsymbol{Z}_{\hat{\mathcal{V}}^{c}}}
\right]^{\top}\left(\boldsymbol{I}-\hat{\kappa}_{\text{liml}} {M}_{{\boldsymbol{Z}}}\right)\boldsymbol{Y}\Big)\label{WIT},\\
\hat{\kappa}_{\text{liml}} = \underset{\beta}{\operatorname{min}}\,\big\{G(\beta) = \big((\boldsymbol{Y}-\boldsymbol{D}\beta)^\top M_{{\boldsymbol{Z}}}(\boldsymbol{Y}-\boldsymbol{D}\beta)\Big)^{-1}\Big((\boldsymbol{Y}-\boldsymbol{D}\beta)^\top M_{{\boldsymbol{Z}_{\hat{\mathcal{V}}^{c}}}}(\boldsymbol{Y}-\boldsymbol{D}\beta)\Big)\Big\}
\label{k_fuller},
\end{eqnarray}
Note, (\ref{WIT}) belongs to the general $k$-class estimators \citep{nagar1959bias}, whose properties vary upon the choice of $\hat{\kappa}$, i.e.\ $\hat{\kappa} = 0$ refers to OLS and $\hat{\kappa} = 1$ reduces to the TSLS estimator. (\ref{k_fuller}) has a closed-form solution: $\hat{\kappa}_{\text{liml}} = \lambda_{\text{min}}\Big(\{[\boldsymbol{Y},\boldsymbol{D}]^\top M_{\boldsymbol{Z}}[\boldsymbol{Y},\boldsymbol{D}]\}^{-1}\{[\boldsymbol{Y},\boldsymbol{D}]^\top M_{\boldsymbol{Z}_{\hat{\mathcal{V}}^c}}[\boldsymbol{Y},\boldsymbol{D}]\}\Big)$, where $\lambda_{\text{min}}(\cdot)$ denotes the smallest eigenvalue.
 Focusing on $\hat{\beta}^{\text{WIT}}$ as the primary interest, we reformulate (\ref{WIT}) and (\ref{k_fuller}) based on the residuals of $\boldsymbol{Y},\boldsymbol{D},\boldsymbol{Z}_{\hat{\mathcal{V}}}$ on $\boldsymbol{Z}_{\hat{\mathcal{V}}^{c}}$. Denote $\boldsymbol{Y}_{\perp} = M_{\boldsymbol{Z}_{\hat{\mathcal{V}}^c}}\boldsymbol{Y}$, $\boldsymbol{D}_{\perp} = M_{\boldsymbol{Z}_{\hat{\mathcal{V}}^c}}\boldsymbol{D}$ and $\boldsymbol{Z}_{\hat{\mathcal{V}}\perp} = M_{\boldsymbol{Z}_{\hat{\mathcal{V}}^c}}\boldsymbol{Z}_{\hat{\mathcal{V}}}$  and notice $M_{\boldsymbol{Z}_{\hat{\mathcal{V}}^c}}M_{\boldsymbol{Z}_{\hat{\mathcal{V}}\perp}} = M_{\boldsymbol{Z}}$, thus it is equivalent to derive asymptotic results on the following model \eqref{resmodel}.
 \begin{eqnarray}
\hat{\beta}^{\text{WIT}}&=&\left(\boldsymbol{D}_{\perp}^{\top}\left(\boldsymbol{I}-\hat{\kappa}_{\text{liml}} M_{\boldsymbol{Z}_{\hat{\mathcal{V}}\perp}}\right) \boldsymbol{D}_{\perp}\right)^{-1}\left(\boldsymbol{D}_{\perp}^{\top}\left(\boldsymbol{I}-\hat{\kappa}_{\text{liml}} M_{\boldsymbol{Z}_{\hat{\mathcal{V}}\perp}}\right) \boldsymbol{Y}_{\perp}\right), \label{resmodel}\\
\hat{\kappa}_{\text{liml}} &=& \lambda_{\text{min}}\Big(\{[\boldsymbol{Y}_{\perp},\boldsymbol{D}_{\perp}]^\top M_{\boldsymbol{Z}_{\hat{\mathcal{V}}\perp}}[\boldsymbol{Y}_{\perp},\boldsymbol{D}_{\perp}]\}^{-1}\{[\boldsymbol{Y}_{\perp},\boldsymbol{D}_{\perp}]^\top[\boldsymbol{Y}_{\perp},\boldsymbol{D}_{\perp}]\}\Big)
.\end{eqnarray}

\subsection{Asymptotic Behavior of WIT Estimator}
Throughout this section, we aim to recover the one specific element in $\mathcal{Q}$, denoted as $({\beta^*},\boldsymbol{\alpha}^*,\boldsymbol{\epsilon})$. Though a slight abuse of notation, we use  $\widehat{\boldsymbol{\alpha}}$
to denote a local solution of (\ref{Problem3}) with MCP.

All local solutions of (\ref{Problem3}) we consider are characterized by the Karush–Kuhn–Tucker (KKT) or first-order  condition, i.e.\
 \begin{equation}
     \widetilde{\boldsymbol{Z}}^{\top}(\boldsymbol{Y}-\widetilde{\boldsymbol{Z}}\widehat{\boldsymbol{\alpha}})/n = \frac{\partial}{\partial \boldsymbol{t}}\sum_{j = 1}^{p}p^\prime_\lambda ({t}_j)\Big|_{\boldsymbol{t} = \widehat{\boldsymbol{\alpha}}}.\label{form 1}
 \end{equation}
 Explicitly, to the end of finding valid IVs via comparing with true signal $\boldsymbol{\alpha}^*$, we rewrite (\ref{form 1}) as
\begin{equation}
\begin{cases}\left(\lambda-\frac{1}{\rho}\left|\widehat{\alpha}_{j}\right|\right)_{+} \leqslant \operatorname{sgn}\left(\widehat{\alpha}_{j}\right) \tilde{\boldsymbol{Z}}_{j}^{\top}(\boldsymbol{Y}-\tilde{\boldsymbol{Z}} \widehat{\boldsymbol{\alpha}}) / n \leqslant \lambda, & j \in \hat{\mathcal{V}}^{c} \\ \left|\tilde{\boldsymbol{Z}}_{j}^{\top}(\boldsymbol{Y}-\tilde{\boldsymbol{Z}} \widehat{\boldsymbol{\alpha}}) / n\right| \leqslant \lambda, &
j \in \hat{\mathcal{V}}\end{cases} \label{form 2}
\end{equation}
where the inequalities in the first line stem from the convexity of the MCP penalty and in the last line originate in the sub-derivative of the MCP penalty at the origin.

As discussed in Section \ref{sec:model}, $\tilde{\boldsymbol{Z}}$ is a function of $(\boldsymbol{Z},\boldsymbol{\gamma}^*,\boldsymbol{\eta})$ and thus endogenous with $\boldsymbol{\epsilon}$. The fact that $\tilde{\boldsymbol{Z}}$ inherits the randomness of $\boldsymbol{\eta}$ distinguishes itself from the general assumptions put on the design matrix, and obscures the feasibility of conditions required to achieve exact support recovery in the literature of penalized least squares estimator (PLSE) \citep{feng2019sorted,loh2017support,zhang2012general}.
The sisVIVE method imposed the RIP condition directly on $\tilde{\boldsymbol{Z}}$ to establish an error bound. However, the restricted eigenvalue (RE) condition \citep{bickel2009simultaneous} is the weakest condition \citep{van2009conditions} available to guarantee rate minimax performance in prediction and coefficient estimation, as well as to establish variable selection consistency for Lasso penalty. \cite{feng2019sorted} further adopted the RE condition for non-convex penalty analysis. We then state the conditions on the design matrix $\tilde{\boldsymbol{Z}}$ of \eqref{Problem3} in the following. Define restricted cone $\mathscr{C}(\mathcal{V}^{*} ; \xi)=\left\{\boldsymbol{u}:\left\|\boldsymbol{u}_{\mathcal{V}^{*}}\right\|_{1} \leq \xi\left\|\boldsymbol{u}_{\mathcal{V}^{c*}}\right\|_{1}\right\}$ for some $\xi >0$ that estimation error $\widehat{\boldsymbol{\alpha}}-\boldsymbol{\alpha}^*$ belongs to.  The restricted eigenvalue $K_{\mathscr{C}}$ for $\tilde{\boldsymbol{Z}}$ is defined as $K_{\mathscr{C}}$ = $K_{\mathscr{C}}(\mathcal{V}^*,\xi)\coloneqq \underset{\boldsymbol{u}}{\operatorname{inf}}\{\|\tilde{\boldsymbol{Z}}\boldsymbol{u}\|_2/(\|\boldsymbol{u}\|_2 n^{1/2}): \boldsymbol{u}\in \mathscr{C}(\mathcal{V}^* ; \xi) \}$ and the RE condition refers to the condition that $K_{\mathscr{C}}$ for $\tilde{\boldsymbol{Z}}$ should be bounded away from zero.
\begin{lemma}\label{Lemma 2}
 (RE condition of $\tilde{\boldsymbol{Z}}$) {{Under Assumptions 1-3, there exists a constant  $\xi \in (0, \|\widehat{\boldsymbol{\gamma}}_{\mathcal{V}^{*}}\|_1/\|\widehat{\boldsymbol{\gamma}}_{\mathcal{V}^{c*}}\|_1)$ and a restricted cone  $\mathscr{C}(\mathcal{V}^* ; \xi)$ defined by chosen $\xi$ such that $K_{\mathscr{C}}^2 >0$ holds strictly.}}
 \end{lemma}
 Lemma \ref{Lemma 2} elaborates that the RE condition on $\tilde{\boldsymbol{Z}}$ holds without any additional assumptions on $\tilde{\boldsymbol{Z}}$, unlike the extra RIP condition for sisVIVE. Moreover, this restricted cone is invariant to scaling, thus indicating the accommodation of many weak IVs. These two features suggest the theoretical advantages of penalized methods (\ref{seqid}) over existing methods.

 Next, we discuss the selection of valid IVs by comparing the local solution of (\ref{form 2}) with the oracle (moment-based) counterpart.
 Define  $\widehat{\boldsymbol{\alpha}}^{\text{or}}_{\mathcal{V}^*} = \boldsymbol{0}$ and
 \begin{equation}
     \widehat{\boldsymbol{\alpha}}^{\text{or}}_{\mathcal{V}^{c*}} = (\tilde{\boldsymbol{Z}}_{\mathcal{V}^{c*}}^\top \tilde{\boldsymbol{Z}}_{\mathcal{V}^{c*}})^{-1}\tilde{\boldsymbol{Z}}_{\mathcal{V}^{c*}}^\top\boldsymbol{Y}\quad {\text{or}} \quad \widehat{\boldsymbol{\alpha}}_{\mathcal{V}^{c*}}^{\text{or}} = (\boldsymbol{Z}_{\mathcal{V}^{c*}}^\top \boldsymbol{Z}_{\mathcal{V}^{c*}})^{-1}[\boldsymbol{Z}_{\mathcal{V}^{c*}}^\top(\boldsymbol{Y}-\widehat{\boldsymbol{D}}\hat{\beta}_{\text{or}}^{\text{TSLS}})],
 \end{equation}
 where $\hat{\beta}_{\text{or}}^{\text{TSLS}} = [\boldsymbol{D}^\top(P_{\boldsymbol{Z}}-P_{\boldsymbol{Z}_{\mathcal{V}^{c*}}})\boldsymbol{D}]^{-1}[\boldsymbol{D}^\top(P_{\boldsymbol{Z}}-P_{\boldsymbol{Z}_{\mathcal{V}^{c*}}})\boldsymbol{Y}]$ and the two versions of $\widehat{\boldsymbol{\alpha}}^{\text{or}}_{\mathcal{V}^{c*}}$ are equivalent. Notice this $\hat{\beta}^{\text{TSLS}}_{\text{or}}$ is not for the final treatment effect estimation, but to illustrate the selection stage consistency only.  To this end, we show the supremum norm of  $\boldsymbol{R}^{\text{or}} = \widetilde{\boldsymbol{Z}}^{\top}(\boldsymbol{Y}-\widetilde{\boldsymbol{Z}}{\widehat{\boldsymbol{\alpha}}}^{\text{or}})/n$ is bounded by an inflated order of $O(\sqrt{\log p_{\mathcal{V}^*} /n})$. Denote $\TTD = (  P_{\boldsymbol{Z}}-P_{\boldsymbol{Z}_{\mathcal{V}^{c*}}})\boldsymbol{D}$ and $\Tilde{\Tilde{\boldsymbol{Q}}}_n = \boldsymbol{Z}_{\mathcal{V}^*}^\top (  P_{\boldsymbol{Z}}-P_{\boldsymbol{Z}_{\mathcal{V}^{c*}}})\boldsymbol{Z}_{\mathcal{V}^*}/n$, we derive the following lemma.
 \begin{lemma}\label{Lemma 3}
 Suppose Assumptions 1-4 hold and let
 \begin{equation}
      \zeta \asymp \frac{p_{\mathcal{V}^*}}{n}\cdot\frac{\|\Tilde{\Tilde{\boldsymbol{Q}}}_n\boldsymbol{\gamma}_{\mathcal{V}^*}^*\|_{\infty}}{\boldsymbol{\gamma}_{\mathcal{V}^*}^{*\top}\Tilde{\Tilde{\boldsymbol{Q}}}_n\boldsymbol{\gamma}_{\mathcal{V}^*}^*}+\sqrt{\frac{\log p_{\mathcal{V}^*}}{n}}.\label{margin}
 \end{equation} Then, the supremum norms of residual $\boldsymbol{R}^{\text{or}}$ are bounded by $\zeta$, i.e.,
 \begin{eqnarray}
     \|\boldsymbol{R}^{\text{or}}\|_\infty &\leq&  \Big\|\frac{\boldsymbol{Z}_{\mathcal{V}^*}^\top \TTE}{n}\Big\|_\infty+\Bigg\|\frac{\frac{\boldsymbol{Z}_{\mathcal{V}^*}^\top \TTD}{n}\cdot \frac{\TTD^\top\boldsymbol{\epsilon}}{n}}{\frac{\TTD^\top\TTD}{n}}\Bigg\|_\infty \leq \zeta
 \end{eqnarray}
 holds with probability approaching 1.
 \end{lemma}
 Based on Lemmas \ref{Lemma 2} and \ref{Lemma 3}, we consider the set  $\mathscr{B}(\lambda,\rho) = \{\widehat{\boldsymbol{\alpha}} \text{ in} \,\,(\ref{form 2}):  \lambda \geq \zeta, \rho > K^{-2}_{\mathscr{C}}(\mathcal{V}^*,\xi)\lor 1 \}$, in which $\zeta$ is defined in (\ref{margin}) and $\xi$ is guaranteed by Lemma \ref{Lemma 2}, as a collection of all local solutions $\widehat{\boldsymbol{\alpha}}$ computed in (\ref{form 2}) through a broad class of MCP under a certain penalty level $\lambda$ and convexity $1/\rho$. Given that the computed local solutions in practice are through a discrete path with some starting point (see Section \ref{sec:3.3}), we further consider the computable solution set  $\mathscr{B}_0(\lambda,\rho)$, introduced by \cite{feng2019sorted}, i.e.,
 \begin{equation}
     \mathscr{B}_0(\lambda,\rho) = \{\widehat{\boldsymbol{\alpha}}:\widehat{\boldsymbol{\alpha}} \text{ and starting point }\widehat{\boldsymbol{\alpha}}^{(0)} \text{ are connected in } \mathscr{B}(\lambda,\rho)\}.\label{B_0}
 \end{equation}
 The connection between $\mathscr{B}_0(\lambda,\rho)$ and $\mathscr{B}(\lambda,\rho)$ is that $\exists \,\, \widehat{\boldsymbol{\alpha}}^{(l)}\in \mathscr{B}(\lambda,\rho)$ with penalty level $\lambda^{(l)}$ increasing with the index $l = 1,2,\ldots$, such that $\widehat{\boldsymbol{\alpha}}^{(0)}-\boldsymbol{\alpha}^*\in \mathscr{C}(\mathcal{V}^*,\xi)$, $\widehat{\boldsymbol{\alpha}} = \widehat{\boldsymbol{\alpha}}^{(l)}$ for large enough $l$ and $\|\widehat{\boldsymbol{\alpha}}^{(l)}-\widehat{\boldsymbol{\alpha}}^{(l-1)}\|_1<a_0 \lambda^{(l)}$, where $a_0$ is specified in Lemma \ref{Lemma A3} of \ref{proof of theorem 3}. Thus, $\mathscr{B}_0(\lambda,\rho)$ is a collection of approximations of  $\boldsymbol{\alpha}$ in all DGPs.

 Denote
  \begin{align*}
 \operatorname{Bias}(\hat{\beta}_{\text{or}}^{\text{TSLS}}) &= \frac{\boldsymbol{D}^\top(P_{\boldsymbol{Z}}-P_{\boldsymbol{Z}_{\mathcal{V}^{c*}}})\boldsymbol{\epsilon}}{\boldsymbol{D}^\top(P_{\boldsymbol{Z}}-P_{\boldsymbol{Z}_{\mathcal{V}^{c*}}})\boldsymbol{D}} = \frac{\boldsymbol{D}^\top_{\perp} P_{\boldsymbol{Z}_{\perp}}\boldsymbol{\epsilon}_{\perp} }{\boldsymbol{D}^\top_{\perp} P_{\boldsymbol{Z}_{\perp}}\boldsymbol{D}_{\perp}};\\
 \bar{\boldsymbol{\gamma}}^*_{\mathcal{V}^{c*}} &= \boldsymbol{\gamma}^*_{\mathcal{V}^{c*}}+(\boldsymbol{Z}_{\mathcal{V}^{c*}}^\top \boldsymbol{Z}_{\mathcal{V}^{c*}})^{-1}\boldsymbol{Z}_{\mathcal{V}^{c*}}^\top\boldsymbol{Z}_{\mathcal{V}^*}\boldsymbol{\gamma}^*_{\mathcal{V}^*}.
 \end{align*}
 
 Also, let $\Tilde{\Tilde{\boldsymbol{Q}}}_n^c$ and $\operatorname{Bias}(\hat{\beta}_{\text{or}}^{c~\text{TSLS}})$ be defined as $\mathcal{P}_c$ version of  $\Tilde{\Tilde{\boldsymbol{Q}}}_n$ and $\operatorname{Bias}(\hat{\beta}_{\text{or}}^{\text{TSLS}})$. Then, we provide the asymptotic result of selection consistency of WIT.
 \begin{theorem}\label{Theorem 3}
 {\bf(Selection Consistency)}  Specify separation conditions $\kappa(n)$ and $\kappa^c(n)$ in Assumption 5 as
  \begin{eqnarray}
     \kappa(n) &\asymp& \underbrace{\sqrt{\frac{\log p_{\mathcal{V}^*}}{n}}}_{T_1} + \underbrace{\frac{p_{\mathcal{V}^*}}{n}\cdot\frac{\|\Tilde{\Tilde{\boldsymbol{Q}}}_n\boldsymbol{\gamma}_{\mathcal{V}^*}^*\|_{\infty}}{\boldsymbol{\gamma}_{\mathcal{V}^*}^{*\top}\Tilde{\Tilde{\boldsymbol{Q}}}_n\boldsymbol{\gamma}_{\mathcal{V}^*}^*}}_{T_2} + \underbrace{|\operatorname{Bias}(\hat{\beta}_{\text{or}}^{\text{TSLS}})| \|\bar{\boldsymbol{\gamma}}^*_{\mathcal{V}^{c*}}\|_\infty}_{T_3}, \label{rate of kappa}\\
      \kappa^c(n) &\asymp& (1+c)\Big\{{\sqrt{\frac{\log |\mathcal{I}_c|}{n}}} + {\frac{|\mathcal{I}_c|}{n}\cdot\frac{\|\Tilde{\Tilde{\boldsymbol{Q}}}^c_n\boldsymbol{\gamma}_{\mathcal{I}_c}^*\|_{\infty}}{\boldsymbol{\gamma}_{\mathcal{I}_c}^{*\top}\Tilde{\Tilde{\boldsymbol{Q}}}^c_n\boldsymbol{\gamma}_{\mathcal{I}_c}^*}}\Big\} + {|\operatorname{Bias}(\hat{\beta}_{\text{or}}^{c~\text{TSLS}})|\|\bar{\boldsymbol{\gamma}}^*_{\mathcal{I}^c_c}\|_\infty},\label{rate of kappa^c}
 \end{eqnarray}
 where $T_1\rightarrow0$ as $n\rightarrow\infty$.
 Suppose Assumptions 1-6 hold and consider computable local solutions specified in (\ref{B_0}). Then
\begin{equation}
   \widehat{\boldsymbol{\alpha}}^{\text{MCP}} = \underset{\widehat{\boldsymbol{\alpha}}\in \mathscr{B}_0(\lambda,\rho)}{\operatorname{arg min}} \|\widehat{\boldsymbol{\alpha}}\|_0,~ \Pr(\hat{\mathcal{V}} = \mathcal{V}^*, \widehat{\boldsymbol{\alpha}}^{\text{MCP}}  = \widehat{\boldsymbol{\alpha}}^{\text{or}}) \overset{p}{\rightarrow}1\label{Selection Consistency}.
\end{equation}
 \end{theorem}
 In Theorem \ref{Theorem 3}, $T_1$ is similar to a standard rate $\sqrt{\log p /n}$ in penalized linear regression, while $T_2$ and $T_3$ are the additional terms that only play a role in the many IVs context and vanish fast in the finite strong IVs case (see Corollary \ref{Corollary 2}). This result is new to the literature.
 \begin{proposition}\label{Proposition 2}
 Under the same assumptions of Theorem \ref{Theorem 3}, if there does not exist a dominant scaled $\gamma^*_j$, i.e.\ $\|\Tilde{\Tilde{\boldsymbol{Q}}}_n^{1/2}\boldsymbol{\gamma}_{\mathcal{V}^*}^*\|_\infty/\|\Tilde{\Tilde{\boldsymbol{Q}}}_n^{1/2}\boldsymbol{\gamma}_{\mathcal{V}^*}^*\|_1 = o\Big(\|\Tilde{\Tilde{\boldsymbol{Q}}}_n^{1/2}\boldsymbol{\gamma}_{\mathcal{V}^*}^*\|_1/(p_{\mathcal{V}^*}\|\Tilde{\Tilde{\boldsymbol{Q}}}_n^{1/2}\|_\infty)\Big)$, then $T_2\rightarrow0$.
 \end{proposition}

 Proposition \ref{Proposition 2} shows that $T_2$ is limited in the general case where dominant scaled $\gamma^*_j$ does not exist. For example, if we assume $\boldsymbol{Q}_n = \boldsymbol{I}$ and $\boldsymbol{\gamma}_{\mathcal{V}^*}^* = C\boldsymbol{1}_{p_{\mathcal{V}^*}}$, where $C$ is a constant or diminishing to zero, then  $\|\boldsymbol{\gamma}_{\mathcal{V}^*}^*\|_\infty/\|\boldsymbol{\gamma}_{\mathcal{V}^*}^*\|_1 = o(\|\boldsymbol{\gamma}_{\mathcal{V}^*}^*\|_1/p_{\mathcal{V}^*}) = o(Cp_{\mathcal{V}^*}/p_{\mathcal{V}^*}) = o(C)$ holds and it follows that $T_2\rightarrow0$.

 \begin{proposition}\label{Proposition 3}{\bf (Approximation of $\operatorname{Bias}(\hat{\beta}_{or}^{TSLS})$)} Let $s = \operatorname{max}(\mu_n,{p_{\mathcal{V}^*}})$. Under Assumptions 1-4, we obtain
 \begin{eqnarray}
        E\left[\operatorname{Bias}(\hat{\beta}^{\text{TSLS}}_{\text{or}})\right] &=& \frac{\sigma_{\epsilon\eta}}{\sigma_{\eta}^{2}}\left(\frac{{p_{\mathcal{V}^*}}}{(\mu_n+{p_{\mathcal{V}^*}})}-\frac{2 \mu_n^{2}}{(\mu_n+{p_{\mathcal{V}^*}})^{3}}\right)+o\left(s^{-1}\right).\label{Exp Bias}
 \end{eqnarray}

 \end{proposition}

\begin{remark}
 The rate of concentration parameter $\mu_n$ will affect $T_3$ through  $|\operatorname{Bias}(\hat{\beta}^{\text{TSLS}}_{\text{or}})|$ in the many IVs setting. Suppose Assumption 4 holds, that $\mu_n \overset{p}{\rightarrow}\mu_0 n$, the leading term in (\ref{Exp Bias}) is $\displaystyle{\frac{\sigma_{\epsilon\eta}}{\sigma^2_\eta}\frac{\nu_{p_{\mathcal{V}^*}}}{\mu_0+\nu_{p_{\mathcal{V}^*}}}}	\ll \displaystyle{\frac{\sigma_{\epsilon\eta}}{\sigma^2_\eta}}$ for moderate $\mu_0$ since $0<\nu_{p_{\mathcal{V}^*}}<1$ while $\mu_0$ could be larger than $1$. While in the many weak IVs setting \citep{chao2005consistent,hansen2008estimation,newey2009generalized}, $\mu_n/n\overset{p}{\rightarrow}0$ and the leading term in \eqref{Exp Bias} becomes  $\displaystyle{{\sigma_{\epsilon\eta}}/{\sigma^2_\eta}}$. Thus, the many weak IVs setting imposes some difficulty (a higher $T_3$) for selecting valid IVs in Theorem \ref{Theorem 3}. 
 \end{remark}

The following theorem describes the asymptotic behavior of the WIT estimator for many valid and invalid IVs cases by combining Theorem 1 and invariant likelihood arguments in \cite{kolesar2018minimum}. We further denote two statistics
\begin{equation}
    \boldsymbol{S} = \frac{1}{n-p}(\boldsymbol{Y},\boldsymbol{D})^\top M_{\boldsymbol{Z}}(\boldsymbol{Y},\boldsymbol{D}),\quad \boldsymbol{T} = \frac{1}{n}(\boldsymbol{Y}_{\perp},\boldsymbol{D}_{\perp})^\top M_{\boldsymbol{Z}_{\perp}}(\boldsymbol{Y}_{\perp},\boldsymbol{D}_{\perp}) \label{Two Stat}
\end{equation}
 as the estimates of the covariance matrix of reduced-form error $\boldsymbol{\Omega} = \operatorname{Cov}(\boldsymbol{\epsilon}+\beta^*\boldsymbol{\eta},\boldsymbol{\eta})$ and a variant of concentration parameter, respectively. Also,  let $m_{\text{max}} = \lambda_{\text{max}}( \boldsymbol{S}^{-1}\boldsymbol{T})$, $\hat{\mu}_n = \operatorname{max}(m_{\text{max}}-{p_{\mathcal{V}^*}}/n,0)$ and $\widehat{\boldsymbol{\Omega}} = \frac{n-p}{n-{p_{\mathcal{V}^{c*}}}/n} \boldsymbol{S}+\frac{n}{n-{p_{\mathcal{V}^{c*}}}/n}(\boldsymbol{T}-\frac{\hat{\mu}_n}{\hat{\boldsymbol{a}}^\top \boldsymbol{S}^{-1} \hat{\boldsymbol{a}}} \hat{\boldsymbol{a}} \hat{\boldsymbol{a}}^\top )$, where $\hat{\boldsymbol{a}} = (\hat{\beta}^{\text{WIT}},1)$ and {{ $|\boldsymbol{\Sigma}|$ is the determinant of $\boldsymbol{\Sigma}$}}.

 \begin{theorem}\label{theorem 4}
 Under the same conditions as in Theorem \ref{Theorem 3}, we obtain:
 \begin{enumerate}[itemsep=2pt,topsep=0pt,parsep=0pt,label=(\alph*)]
 \item (Consistency): $\hat{\beta}^{\text{WIT}}\overset{p}{\rightarrow} \beta^*$ with $\hat{\kappa}_{\text{liml}} = \frac{1-\upsilon_{p_{\mathcal{V}^*}}}{1-\upsilon_{p_{\mathcal{V}^{c*}}}-\upsilon_{p_{\mathcal{V}^*}}}+ o_p(1)$.
 \item (Asymptotic normality): $\sqrt{n}(\hat{\beta}^{\text{WIT}}-\beta^*)  \overset{d}{\rightarrow}\mathcal{N}\Big(0,\mu_{0}^{-2}\big[\sigma_\epsilon^2\mu_0+\frac{\upsilon_{p_{\mathcal{V}^{c*}}}(1-\upsilon_{p_{\mathcal{V}^*}})}{1-\upsilon_{p_{\mathcal{V}^{c*}}}-\upsilon_{p_{\mathcal{V}^*}}}|\boldsymbol{\Sigma}|\big]\Big)$.
 \item (Consistent variance estimator):
 \begin{equation*}
     \begin{aligned}
          \widehat{\operatorname{Var}}(\hat{\beta}^{\text{WIT}}) &= \frac{\hat{\boldsymbol{b}}^\top \widehat{\boldsymbol{\Omega}}\hat{\boldsymbol{b}}(\hat{\mu}_n+{p_{\mathcal{V}^*}}/n)}{-\hat{\mu}_n}\Big(\hat{Q}_S\widehat{\boldsymbol{\Omega}}_{22}-\boldsymbol{T}_{22} +\frac{\hat{c}}{1-\hat{c}}\frac{\hat{Q}_S}{\hat{\boldsymbol{a}}^\top \widehat{\boldsymbol{\Omega}}^{-1}\hat{\boldsymbol{a}}}\Big)^{-1}\\
          &\overset{p}{\rightarrow} \mu_{0}^{-2}\big[\sigma_\epsilon^2\mu_0+\frac{\upsilon_{p_{\mathcal{V}^{c*}}}(1-\upsilon_{p_{\mathcal{V}^*}})}{1-\upsilon_{p_{\mathcal{V}^{c*}}}-\upsilon_{p_{\mathcal{V}^*}}}|\boldsymbol{\Sigma}|\big],
     \end{aligned}
 \end{equation*}
 where $\hat{\boldsymbol{b}} = (1,-\hat{\beta}^{\text{WIT}})$ and $\hat{Q}_S = \frac{\hat{\boldsymbol{b}}^\top \boldsymbol{T}\hat{\boldsymbol{b}}}{\hat{\boldsymbol{b}}^\top \widehat{\boldsymbol{\Omega}}\hat{\boldsymbol{b}}}$.
 \end{enumerate}
 \end{theorem}

 Notably, when the number of invalid IVs ${p_{\mathcal{V}^*}}$ is a constant, the variance estimator above is reduced to the one that \cite{bekker1994alternative} derived for the typical many IVs case.  \cite{hansen2008estimation} showed that it is still valid under many weak IVs asymptotics.

\subsection{Special Cases in Low \& High Dimensions}
{  The values of $\kappa(n)$ and $\kappa^c(n)$ outlined in Theorem \ref{Theorem 3} provide the general formulae to confirm Assumption 5. In this subsection, we discuss some  representative cases in low and high dimensions and verify that Assumption 5 holds.

\subsubsection{Finite $p$ Case}\label{Section 331}
First, we show that the WIT estimator is a more powerful tool than the existing methods requiring the majority rule \citep{kang2016instrumental,windmeijer2019use} also under the \emph{finite number of IVs with a mixture of strong and weak IVs} settings. WIT achieves the same asymptotic results as \cite{windmeijer2019use} under more relaxed conditions. Specifically, under finite IVs, Assumption 1 can be reduced to Assumption 1$^\prime$ as follows.

{\bf{Assumption}} 1$^\prime$ (Finite Number of IVs): ${p_{\mathcal{V}^{c*}}}\geq 1$ and ${p_{\mathcal{V}^*}}\geq1$ are fixed constants, and ${p_{\mathcal{V}^{c*}}}+{p_{\mathcal{V}^*}} = p<n$.

{In the finite IVs case, the $T_2$ and $T_3$ terms in Theorem \ref{Theorem 3} for selecting valid IVs go to zero fast and the required separation rate reduces to $\kappa(n)\asymp n^{-1/2}$. We present the asymptotic properties for the WIT estimator here. Consider the following mixed IV strengths case. Let $\gamma_j^* = Cn^{-\tau_j}$ for $j = 1,2,\ldots,p$, $\tau_{\mathcal{V}^*} = {\operatorname{argmax}}_{\tau_j}\{\tau_j:j\in\mathcal{V}^*\}$, $\tau_{\mathcal{I}_c} = {\operatorname{argmax}}_{\tau_j}\{\tau_j:j\in\mathcal{I}_c\}$, and $\tau_{\mathcal{I}_{\tilde{c}}} = {\operatorname{argmax}}_{\tau_j}\{\tau_j:j\in\mathcal{I}_{\tilde{c}}\}$, where $c\neq \tilde{c}$.
\begin{corollary}\label{Corollary 2}
{\bf (Finite $p$ with Mixture of Strong and Weak IVs)} Suppose Assumptions 1$^\prime$, 2-4 and 6 hold. Additionally, we assume each IV is at least a weak one such that $\gamma^*_j = O(n^{-\tau_j})$ and $0\leq\tau_j\leq 1/2$ for $j = 1,2,\ldots,p$. For any fixed ${min}_{\operatorname{j\in \mathcal{V}^{c*}}}\, |{\alpha}^*_j| >0$, if $\tau_{\mathcal{V}^*}+2\tau_{\mathcal{I}_c}<1$, $\tau_{\mathcal{V}^*}+2\tau_{\mathcal{I}_{\tilde{c}}}<1$ and $\tau_{\mathcal{I}_c}+\tau_{\mathcal{I}_{\tilde{c}}} <2/3$, then we have\\
(a) (Selection consistency): $  \widehat{\boldsymbol{\alpha}}^{\text{MCP}} = \underset{\widehat{\boldsymbol{\alpha}}\in \mathscr{B}_0(\lambda,\rho)}{\operatorname{arg min}} \|\widehat{\boldsymbol{\alpha}}\|_0, \Pr(\hat{\mathcal{V}} = \mathcal{V}^*, \widehat{\boldsymbol{\alpha}}^{\text{MCP}}  = \widehat{\boldsymbol{\alpha}}^{\text{or}}) \overset{p}{\rightarrow}1$.\\
 (b) (Consistency \& equivalence of WIT and TSLS): $\hat{\beta}^{\text{WIT}}\overset{p}{\rightarrow} \beta^*$ with $\hat{\kappa}_{\text{liml}} = 1+ o_p(1)$.\\
 (c) (Asymptotic normality): $\sqrt{n}(\hat{\beta}^{\text{WIT}}-\beta^*)  \overset{d}{\rightarrow}\mathcal{N}\Big(0,\mu_{0}^{-1}\sigma_\epsilon^2\Big)$. \\
 (d) (Consistent variance estimator): $$ \widehat{\operatorname{Var}}(\hat{\beta}^{\text{WIT}}) = \frac{\hat{\boldsymbol{b}}^\top \widehat{\boldsymbol{\Omega}}\hat{\boldsymbol{b}}(\hat{\mu}_n+{p_{\mathcal{V}^*}}/n)}{-\hat{\mu}_n}\Big(\hat{Q}_S\widehat{\boldsymbol{\Omega}}_{22}-\boldsymbol{T}_{22} +\frac{\hat{c}}{1-\hat{c}}\frac{\hat{Q}_S}{\hat{\boldsymbol{a}}^\top \widehat{\boldsymbol{\Omega}}^{-1}\hat{\boldsymbol{a}}}\Big)^{-1}\overset{p}{\rightarrow} \mu_{0}^{-1}\sigma_\epsilon^2,$$
 where $\hat{\boldsymbol{b}} = (1,-\hat{\beta}^{\text{WIT}})$ and $\hat{Q}_S = \frac{\hat{\boldsymbol{b}}^\top \boldsymbol{T}\hat{\boldsymbol{b}}}{\hat{\boldsymbol{b}}^\top \widehat{\boldsymbol{\Omega}}\hat{\boldsymbol{b}}}$.
\end{corollary}
}
Corollary \ref{Corollary 2} addresses a scenario where the violation of validity, $\alpha^*_j$, is held constant and thus is of the same order when $\tau_j=0$, or exceeds its strength parameter, $\gamma^*_j$, for $j \in \mathcal{V}^*_c$. This scenario frequently arises when conducting robust estimations, particularly in circumstances that inadvertently incorporate weakly relevant but strongly invalid IVs, or when there is no prior information on the candidate IVs set. Dealing with these situations effectively is a crucial component of robust statistical analysis.

Conversely, as noted by an anonymous reviewer, experienced researchers often select IVs that are argued to be both relevant and valid. A violation of IV validity could happen among relatively strong IVs. This corresponds to a scenario where the ratio of $\alpha^*_j/\gamma^*_j$ is small. In the following proposition, we demonstrate that the WIT estimator can effectively manage cases where $\alpha^*_j/\gamma^*_j = o(1)$.

\begin{proposition}\label{Proposition lowdim}
{\bf (Finite $p$ with $\alpha^*_j/\gamma^*_j = o(1)$)} Suppose Assumptions 1$^\prime$, 2, 3 and 6 hold, and Assumption 4 holds for each DGP in $\mathcal{Q}$. Each $\gamma^*_j$ is at least a weak IV such that $\gamma^*_j = O(n^{-\tau_j})$ and $0\leq\tau_j\leq 1/2$ for $j = 1,2,\ldots,p$ and satisfy $\gamma_l^*>\underset{j\in \mathcal{V}^{c*}}{\operatorname{max}}\,O(n^{-1/2} \gamma_j^* /\alpha_j^*)$ for valid $l \in \mathcal{V}^*$. Additionally, we suppose ${min}_{\operatorname{j\in \mathcal{V}^{c*}}}\, |{\alpha}^*_j| >O(n^{-1/2})$,  $\alpha^*_j/\gamma^*_j = o(1)$. Under these conditions, conclusions (a) to (d) in Corollary \ref{Corollary 2} remain valid.
\end{proposition}
Proposition \ref{Proposition lowdim} gives the desired asymptotic results for the WIT estimator for DGPs in the $\alpha^*_j/\gamma^*_j=o(1)$ scenario, which aligns with the aforementioned case that some researchers may pick relatively strong IVs in the design phase, and true $\alpha^*_j$ to satisfy the separation condition in Assumption 5. It only requires that $|\gamma_l^*|>\underset{j\in \mathcal{V}^{c*}}{\operatorname{inf}}O(n^{-1/2} \gamma_j^* /\alpha_j^*)$, which means the valid IVs' strength signals are stronger than the rate of $n^{-1/2}$ over $\alpha^*_j/\gamma_j^*$. It can be easily implied from this condition that a smaller value in ratio of $\alpha^*_j$ to $\gamma^*_j$ helps to reduce the needed strength signal $\gamma_l^*$ of the valid IVs. 
 Therefore, $\gamma^*_l = o(\gamma^*_j)$ for $l \in \mathcal{V}^*$ and $j \in \mathcal{V}^{c*}$ is a sufficient condition to verify it.  For the second part in Assumption 5: $|\tilde{\alpha}_j^c|>\kappa^{c}(n)$ for $j \in \{j:\alpha^*_j/\gamma^*_j = \tilde{c}\neq c\}$ will hold automatically under $\alpha^*_j/\gamma^*_j = o(1)$. Table \ref{tab:my_label} subsequently illustrates a simplified representative example featuring only 4 IVs, two of which are valid, and demonstrates the rates of $\kappa(n)$ and $\kappa^{c}(n)$ in true and transformed DGPs.

As a result, Corollary \ref{Corollary 2} and Proposition \ref{Proposition lowdim} provide easily interpretable examples of common DGPs across two distinct scenarios. These examples affirm the validity of Assumption 5, facilitating a more straightforward understanding for practitioners.

\begin{table}\caption{\label{tab:my_label} Illustration of separation condition in case of finite $p$ with $\alpha^*_j/\gamma^*_j = o(1)$}
\centering
\begin{adjustbox}{width = 12cm}
    \begin{threeparttable}
    \begin{tabular}{c|cccc|c|c}
    \hline\hline
         & \multicolumn{4}{c}{Coefficients}&  \multicolumn{2}{c}{Thresholds}\\\hline
        $\boldsymbol{\alpha}^*$ & 0& 0&$\alpha^*_3$& $\alpha^*_4$&$ \kappa(n)$ &$n^{-1/2}$\\
        
        $\boldsymbol{\gamma}^*$&$\gamma^*_1$ &$\gamma^*_2$& $\gamma^*_3$&$\gamma^*_4$&-&- \\\hline
        $\tilde{\boldsymbol{\alpha}}^{c_3}$&$\gamma^*_1\cdot c_3$ & $\gamma^*_2 \cdot c_3$& 0 & $\alpha^*_4-\gamma^*_4 \cdot c_3$& $\kappa^{c_3}(n)$& $ n^{-1/2}$\\
        $\tilde{\boldsymbol{\alpha}}^{c_4}$&$\gamma^*_1 \cdot c_4$& $\gamma^*_2 \cdot c_4$ & $\alpha^*_3$ - $\gamma^*_3\cdot  c_4$ & 0& $\kappa^{c_4}(n) $ & $  n^{-1/2}$\\
        \hline
    \end{tabular}
    	    	\begin{tablenotes}
\item[]\scriptsize Note: This example illustrates a scenario with four IVs, where the first two are valid, $c_3 = \alpha^*_3/\gamma^*_3$, $c_4 = \alpha^*_4/\gamma^*_4$ and $c_1 \neq c_2$. The first two rows represent the true DGPs. The last two rows, $\tilde{\boldsymbol{\alpha}}^c$ and ${\kappa}^c(n)$, correspond to the transformed DGPs according to equation \eqref{GAP} and the threshold defined in the separation condition: Assumption 5. The corresponding $\kappa^{c}(n)$ has been simplified through the proof presented in Proposition \ref{Proposition lowdim}.
\end{tablenotes}
\end{threeparttable}
\end{adjustbox}
\end{table}

\subsubsection{High Dimensional (Many IVs) Case}
Next, we consider a common scenario \citep{Andrews2019,chao2005consistent} with a large number (can grow with $n$) of individually weak instruments, where $\gamma^*_j = O(\sqrt{\log p/n^{1-\delta}})$ for $j = 1,2,\ldots,p$, and some small enough $\delta>0$ such that $\log n/n^{-\delta}<\infty$ to avoid an exploded variance in the outcome variable $Y_i$. Regarding invalid IVs, we assume the same rate of $|\boldsymbol{\alpha}_{\mathcal{V}^{c*}}^*|_{\text{min}} = O(\sqrt{\log p/n^{1-\delta}})$. 

Consequently, the separation condition $\kappa(n)$ is reduced to:
\begin{equation*}
    \kappa(n) \asymp \sqrt{\frac{\log p_{\mathcal{V}^*}}{n}} + \frac{p_{\mathcal{V}^*}}{n}\cdot\frac{\|\Tilde{\Tilde{\boldsymbol{Q}}}_n\boldsymbol{\gamma}_{\mathcal{V}^*}^*\|_{\infty}}{\boldsymbol{\gamma}_{\mathcal{V}^*}^{*\top}\Tilde{\Tilde{\boldsymbol{Q}}}_n\boldsymbol{\gamma}_{\mathcal{V}^*}^*} + {|\operatorname{Bias}(\hat{\beta}_{\text{or}}^{\text{TSLS}})| \|\bar{\boldsymbol{\gamma}}^*_{\mathcal{V}^{c*}}\|_\infty} = O\Big(\sqrt{\frac{\log p_{\mathcal{V}^*}}{n}}\Big)
\end{equation*}
as per Proposition \eqref{Proposition 2}. This satisfies the first part of Assumption 5, that is, $|\boldsymbol{\alpha}_{\mathcal{V}^{c*}}^*|_{\text{min}}>\kappa(n)$. The second part of Assumption 5, $|\tilde{\alpha}_j^c|>\kappa^{c}(n)$ for $j \in \{j:\alpha^*_j/\gamma^*_j = \tilde{c}\neq c\}$, can be confirmed according to \eqref{GAP},
\begin{equation}
     |\tilde{{\alpha}}^c_j| = |\alpha^*_j-c\gamma_j^*| = |\tilde{c}-c| \cdot |\gamma^*_j|= O\Big(\sqrt{\log p/n^{1-\delta}}\Big),
\end{equation}
where $|\tilde{c}-c|$ has the same order of $\alpha^*_j/\gamma^*_j \asymp  O(1)$ for $j \in \mathcal{V}_c^* $. The threshold of $\kappa^c(n)$ in \eqref{rate of kappa^c} is now specified as 
\begin{equation}
\begin{aligned}
    \kappa^c(n) &\asymp (1+c)\Big\{{\sqrt{\frac{\log |\mathcal{I}_c|}{n}}} + {\frac{|\mathcal{I}_c|}{n}\cdot\frac{\|\Tilde{\Tilde{\boldsymbol{Q}}}^c_n\boldsymbol{\gamma}_{\mathcal{I}_c}^*\|_{\infty}}{\boldsymbol{\gamma}_{\mathcal{I}_c}^{*\top}\Tilde{\Tilde{\boldsymbol{Q}}}^c_n\boldsymbol{\gamma}_{\mathcal{I}_c}^*}}\Big\} + {|\operatorname{Bias}(\hat{\beta}_{\text{or}}^{c~\text{TSLS}})|\|\bar{\boldsymbol{\gamma}}^*_{\mathcal{I}^c_c}\|_\infty}\\
&\asymp\sqrt{\frac{\log |\mathcal{I}_c|}{n}}.
    \end{aligned}
\end{equation}
Therefore, the second part of Assumption 5 holds as well.

We summarize the above discussions as the following proposition.
\begin{proposition}\label{high_dim theorem}
    Suppose Assumptions 1-4 and 6 hold. If we assume the uniform rate of $\gamma_j^* = \alpha_l^* = O(\sqrt{\log p/n^{1-\delta}})$ for $j = 1,2,\ldots,p$ , $l\in \mathcal{V}_c^*$ and small enough $\delta>0$ such that $\log n/n^{-\delta}<\infty$, then Assumption 5 is satisfied and the results of Theorems 3 and 4 hold consequently.
\end{proposition}
}

\subsection{Computational Implementation of WIT Estimator}\label{sec:3.3}
Through Proposition \ref{Proposition 1}, we know that MCP belongs to the surrogate sparsest penalty under Assumptions 1-6, and it ensures the global solution in (\ref{Problem3}) matches the sparsest rule. However, each element in $\mathcal{Q}$ is the local solution of (\ref{Problem3}) and we can only obtain one local solution from one initial value practically. A multiple starting points strategy is needed to achieve the global solution. Enumerating the whole $\boldsymbol{\alpha}^*\in \mathbb{R}^p$ is impossible. Therefore, it is important to develop an efficient algorithm for the MCP penalty.

In light of practical use, we adopt the iterative local adaptive majorize-minimization
(I-LAMM) algorithm \citep{fan2018lamm}, which satisfies (\ref{B_0}) as shown in \cite{feng2019sorted} Section 2.1, with different initial values to achieve the local solution of $\widehat{\boldsymbol{\alpha}}$ in (\ref{Problem3}). See more technical details  and derivation in \ref{sec:A3}.

Motivated by the individual IV estimator \citep{windmeijer2019confidence} such that $\hat{\beta}_j = \hat{\Gamma}_j/\hat{\gamma}_j\overset{p}{\rightarrow}\beta^*+\alpha^*_j/\gamma^*_j$, where $\hat{\boldsymbol{\Gamma}} = (\boldsymbol{Z}^\top \boldsymbol{Z})^{-1}\boldsymbol{Z}^\top \boldsymbol{Y}$ and $\hat{\boldsymbol{\gamma}} = (\boldsymbol{Z}^\top \boldsymbol{Z})^{-1}\boldsymbol{Z}^\top \boldsymbol{D}$, we can construct
\begin{equation}
    \check{\boldsymbol{\alpha}}(\check{\beta}) = \hat{\boldsymbol{\Gamma}}-\check{\beta}\hat{\gamma} = \hat{\boldsymbol{\Gamma}}-\beta^*\hat{\boldsymbol{\gamma}}+(\beta^*-\check{\beta})\hat{\boldsymbol{\gamma}},
\end{equation}
with any initial value $\check{\beta}$. Thus $\| \check{\boldsymbol{\alpha}}(\check{\beta})-\boldsymbol{\alpha}^*\|_1 \overset{p}{\rightarrow}|\beta^*-\check{\beta}| \cdot\|\hat{\boldsymbol{\gamma}}\|_1$. It is valid to replace $(\beta^*,\boldsymbol{\alpha}^*)$ by $(\tilde{\beta}^c,\tilde{\boldsymbol{\alpha}}^c)$, we know that $\| \check{\boldsymbol{\alpha}}(\check{\beta})-\tilde{\boldsymbol{\alpha}}^c\|_1$ is asymptotically controlled by $|\tilde{\beta}^c-\check{\beta}| \cdot\|\hat{\boldsymbol{\gamma}}\|_1$. Thus, varying $\check{\beta}\in\mathbb{R}^1$ is equivalent to varying  $\check{\boldsymbol{\alpha}}(\check{\beta})\in {\mathbb{R}}^{p}$ of a close DGP $\mathcal{Q}$.

Here we provide a simple heuristic procedure to efficiently choose a proper $\check{\beta}$. Let $\hat{\beta}_{[j]}$ be the sorted $\hat{\beta}_j$. A fuzzy MCP regression is conducted to clustering $\hat{\beta}_j$:
\begin{equation}
    \bar{\boldsymbol{\beta}} = \underset{\grave{\boldsymbol{\beta}}}{\operatorname{arg min}} \sum_{j=1}^{p}\|\grave{\beta_j}-\hat{\beta}_{[j]}\|^2_2+\sum_{j=2}^{p}p_{\bar{\lambda}}^{\text{MCP}}(|\grave{\beta}_{j}-\grave{\beta}_{j-1}|),\label{fuzzy}
\end{equation}
where $\bar{\boldsymbol{\beta}}\in \mathbb{R}^p$, $\bar{\lambda}$ is a prespecified penalty level and $\grave{\boldsymbol{\beta}}$ is initialized by $\grave{{\beta}}_j = \hat{\beta}_{[j]}$. Therefore, $\bar{\boldsymbol{\beta}}$ consists of less than $p$ distinct values. Thus we can choose the initial $\widehat{\boldsymbol{\alpha}}$ in \eqref{B_0} such that
\begin{equation}
    \widehat{\boldsymbol{\alpha}}^{(0)}(\check{\beta}) = \check{\boldsymbol{\alpha}}(\check{\beta}),\label{alpha_0}
\end{equation}
where $\check{\beta}$'s are chosen in $\bar{\boldsymbol{\beta}}$ with the priority given to the largest (remaining) cluster and $\widehat{{\alpha}}_j^{(0)}(\check{\beta}) = 0$ for $j$ to be the unsorted index of $\hat{\beta}_j$ such that $\check{\beta} = \bar{{\beta}}_j$.

Following \eqref{B_0}, which provides a theoretical guideline for the tuning parameter, we look for a data-driven tuning procedure that has good performance in practice. From numerical studies, results are found to not be sensitive to the choice of $\rho$, which is fixed to $2$ for most applications. However, $\lambda$ is required to be tuned. Cross-validation is implemented in sisVIVE, but is known to be time-consuming and to select too many IVs as invalid. \cite{windmeijer2019use,windmeijer2019confidence} proposed to use the Sargan test under low dimensions to choose the tuning parameter consistently and obtain a good finite sample performance. However, the Sargan test is designed for fixed $p$ and cannot handle many IVs. Therefore, we propose the modified Cragg-Donald (MCD, \citealp{kolesar2018minimum}) test-based tuning procedure, which extends the Sargan test to allow high-dimensional covariates and IVs.

Specifically, consider a local solution $\widehat{\boldsymbol{\alpha}}$ in \eqref{B_0}, Denote ${{p_{\widehat{\mathcal{V}}}}} = |\{j:\widehat{{\alpha}}_j = 0\}|$ and ${{p_{\widehat{\mathcal{V}}^{c}}}} = |\{j:\widehat{{\alpha}}_j \neq 0\}|$.  Let $m_{\text{min}}$ be the minimum eigenvalue of $\boldsymbol{S}^{-1}\boldsymbol{T}$, where $\boldsymbol{S}$ and $\boldsymbol{T}$ are defined as a function of $\widehat{\boldsymbol{\alpha}}$ in \eqref{Two Stat}. Then the MCD test is given by $nm_{\text{min}}$.  According to \cite{kolesar2018minimum} Proposition 4, the MCD test with asymptotic size $\varrho_n$ would reject the null of ${\boldsymbol{\alpha}}_{\hat{\mathcal{V}}}= \boldsymbol{0}$, when
\begin{equation}
    nm_{\text{min}}>\chi^2_{{{p_{\widehat{\mathcal{V}}^{c}}}} -1}\left\{\Phi\left(\sqrt{\frac{n-{{p_{\widehat{\mathcal{V}}^{c}}}} }{n-{{p_{\widehat{\mathcal{V}}^{c}}}} -{{p_{\widehat{\mathcal{V}}}}} }}\Phi^{-1}(\varrho_n)\right)\right\},\label{MCD test}
\end{equation}
where $\Phi(\cdot)$ denotes the CDF of the standard normal distribution and $\chi^2_{{{p_{\widehat{\mathcal{V}}^{c}}}} -1}(\varrho_n)$ is  $1-\varrho_n$ quantile of $\chi^2_{{{p_{\widehat{\mathcal{V}}^{c}}}} -1}$ distribution. This property holds regardless of whether ${p_{\mathcal{V}^{c*}}}$ is fixed or grows with $n$. For the sake of the model selection consistency, the size of the MCD test needs to be $o(1)$. Following \cite{Belloni2012sparse,windmeijer2019use,windmeijer2019confidence}, we adopt a scaled rate $\varrho_n = 0.5/\log n$ that works well in simulation studies. Thus, with  $\widehat{\boldsymbol{\alpha}}^{(0)}$ in \eqref{alpha_0} and a sequence of $\boldsymbol{\lambda}^{\text{seq}} =\boldsymbol{C} \sqrt{\log p/n}$ where $\boldsymbol{C} = 0.1\times(1,\ldots,20)^\top$, we propose to use the MCD test to select the proper $\lambda\in \boldsymbol{\lambda}^{\text{seq}}$ that is not rejected by \eqref{MCD test} with size $\varrho_n$ and  largest ${{p_{\widehat{\mathcal{V}}}}} $.

To sum up, we provide Algorithm \ref{Algorithm 1} to demonstrate the detailed implementation.

\begin{algorithm}[t!]
\footnotesize
\caption{\label{Algorithm 1}WIT estimator with MCD test tuning strategy}
\hspace*{0.01in} {\bf Input:}
$\boldsymbol{Y},{\boldsymbol{Z}},\boldsymbol{D},\boldsymbol{\lambda}^{\text{seq}},\varrho_n, \text{ and } J$
\begin{algorithmic}[1]
\State  Calculate $\bar{\boldsymbol{\beta}}$ in \eqref{fuzzy}, initialize $\widehat{\boldsymbol{\alpha}}^{\text{MCP}} = \boldsymbol{1}$ and $\boldsymbol{I} = \boldsymbol{1}$
    \For { $\widehat{\boldsymbol{\alpha}}^{(0)} = \boldsymbol{0},\widehat{\boldsymbol{\alpha}}^{(0)}(\check{\beta}_j)$ \eqref{alpha_0}} \Comment{ $\check{\beta}_j \in \bar{\boldsymbol{\beta}}$ in priority of the largest cluster, for $j= 1,\ldots,J$}
    \For {$\lambda$ in $\boldsymbol{\lambda}^{\text{seq}}$}
    \State Calculate $\widehat{\boldsymbol{\alpha}}$ through I-LAMM in Algorithm \ref{algorithm 2}
    \If{$\widehat{\boldsymbol{\alpha}}$ is not rejected by MCD test \eqref{MCD test} with size $\varrho_n$ }
    \State $\boldsymbol{I}[l] = 0$ for $l\in \{l: \widehat{\alpha}_l = 0\}$
    \If{$|\{j:\widehat{{\alpha}}_j = 0\}|>|\{j:\widehat{{\alpha}}^{\text{MCP}}_j = 0\}|$}
    \State $\widehat{\boldsymbol{\alpha}}^{\text{MCP}} =\widehat{\boldsymbol{\alpha}}$
    \EndIf
    \EndIf
    \EndFor
    \If{$\|\boldsymbol{I}\|_1\leq |\{j:\{j:\widehat{{\alpha}}^{\text{MCP}}_j = 0\}|$} \Comment{Impossible for the existence of more sparse $\widehat{\boldsymbol{\alpha}}$}
    \State Break
    \EndIf
    \EndFor
\end{algorithmic}
\hspace*{0.02in} {\bf Output:}
$\widehat{\boldsymbol{\alpha}}^{MCP}$
\end{algorithm}

\section{Numerical Simulations}\label{sec:sim}
In this section, we conduct numerical studies to evaluate the finite sample performance of the proposed WIT estimator. In the design of the simulation experiments, we consider scenarios corresponding to different empirically relevant problems.

We consider the same model in Section \ref{sec:model},
$$
    \boldsymbol{Y} = \boldsymbol{D}\beta^*+\boldsymbol{Z}\boldsymbol{\alpha}^*+\boldsymbol{\epsilon}, \quad
    \boldsymbol{D} = \boldsymbol{Z}\boldsymbol{\gamma}^*+\boldsymbol{\eta}.
$$
Throughout all settings, we fix true treatment effect $\beta^* = 1$. $\boldsymbol{Z}$ is the $n\times p$ potential IV matrix and $\boldsymbol{Z_{i.}}\overset{i.i.d.}{\sim}\mathcal{N}(\boldsymbol{0},\boldsymbol{\Sigma}^{\boldsymbol{Z}})$, where $\boldsymbol{\Sigma}^{\boldsymbol{Z}}_{jj} = 0.3$ and $\boldsymbol{\Sigma}^{\boldsymbol{Z}}_{jk} = 0.3|j-k|^{0.8}$ for $i = 1,\ldots,n$ and $k,j = 1,\ldots,p$.  Denote $\boldsymbol{\epsilon}= (\epsilon_1,\epsilon_2,\ldots,\epsilon_n)^\top$ and $\boldsymbol{\eta} = (\eta_1,\eta_2,\ldots,\eta_n)^\top$ and generate $(\epsilon_i,\eta_i)^\top \overset{\operatorname{i.i.d.}}{\sim}\mathcal{N}\left(\boldsymbol{0},\left(\begin{smallmatrix}
\sigma_\epsilon^2 &\sigma_{\epsilon,\eta} \\
\sigma_{\epsilon,\eta}  & \sigma_\eta^2
\end{smallmatrix}\right)\right)$. We let $\sigma_\epsilon = 0.5$ and $\operatorname{\text{corr}}(\epsilon_i,\eta_i) = 0.6$ in all settings but vary $\sigma_\eta^2$ to get different concentration parameters concerning strong or weak IVs cases.

We compare the WIT estimator with other popular estimators in the literature. Specifically, sisVIVE is computed by R package \texttt{sisVIVE}; Post-ALasso \citep{windmeijer2019use}, TSHT and CIIV are implemented using codes on Github \citep{guo2018confidence,windmeijer2019confidence}. TSLS, LIML, oracle-TSLS and oracle-LIML (the truly valid set $\mathcal{V}^*$ is known a priori) are also included.  Regarding our proposed WIT estimator, the MCD tuning strategy is implemented to determine $\lambda$, and we fix $\rho=2$. In the I-LAMM algorithm, we take $\delta_c = 10^{-3}$ and $\delta_t = 10^{-5}$ as the tolerance levels. We report results based on 500 simulations.

We measure the performance of
all estimators in terms of median absolute deviation (MAD), standard deviation (STD), and coverage probability (CP) based on $95\%$ confidence intervals. Moreover, we provide measurements on the estimation of $\boldsymbol{\alpha}^*$ and IV model selection. Specifically, we measure the performance of invalid IVs selection by false positive rate (FPR) and false negative rate (FNR). To be concrete, denote the number of incorrect selections of valid and invalid IVs as FP and FN, respectively, and the number of correct selections of valid and invalid as TP and TN, respectively. Thus, FPR $=\text{FP} /(\text{FP}$ $+\text{TN})$ and $\text{FNR}=\text{FN} /(\text{FN}+\text{TP})$.

\subsection{Case 1: Low dimension}\label{sec:simu1}
We first consider the low dimension scenario:
\begin{enumerate}[itemsep=2pt,topsep=0pt,parsep=0pt]
    \item[] Case 1(I): $\boldsymbol{\gamma}^* = (\boldsymbol{0.5}_{4},\boldsymbol{0.6}_{6})^{\top}$ and $\boldsymbol{\alpha}^* = (\boldsymbol{0}_{5},\boldsymbol{0.4}_{3},\boldsymbol{0.8}_{2})^{\top}$.
    \item[] Case 1(II): $\boldsymbol{\gamma}^* = (\boldsymbol{0.04}_3,\boldsymbol{0.5}_2,0.2,\boldsymbol{0.1}_4)^\top$ and $\boldsymbol{\alpha}^* = (\boldsymbol{0}_5,1,\boldsymbol{0.7}_4)^\top$.
\end{enumerate}
In the above two cases, we maintain $\sigma_\eta = 1$ and vary the sample size $n = 200$ to $500$. Case 1(I) refers to all strong IVs case, but SAIS \eqref{SAIS} condition holds. Case 1(II) refers to Example \ref{Example 1} with mixed strong and weak IVs.

\begin{table}
\caption{\label{tab:simu1}Simulation results in low dimension}
\centering
\begin{tabular}{lccccccccccc}
\hline\hline
\multicolumn{1}{l}{\multirow{2}{*}{Case}} & \multirow{2}{*}{Approaches} & \multicolumn{4}{c}{$n = 200$}                         && \multicolumn{4}{c}{ $n=500$}                        \\
\multicolumn{1}{l}{}                    &                     & MAD                        &    CP    &FPR&FNR       && MAD             &CP           & FPR  &FNR              \\\cline{3-6} \cline{8-11}
\multirow{9}{*}{{1(I)}}& TSLS& 0.532& 0  &   -  &   - &&0.530 & 0&     -  &   -\\
& LIML &0.952& 0.004  &-&- && 0.978& 0 &-&- \\
& oracle-TSLS& 0.038& 0.938&-&-&&0.023& 0.958&     -  &   -\\
& oracle-LIML&0.038& 0.938  &   -  &   -&&0.023& 0.958&     -  &   -\\
&TSHT&0.060& 0.828& 0.158& 0.018&&0.023& 0.950& 0& 0.002\\
&CIIV&0.045& 0.810& 0.072& 0.009&&0.023 & 0.946& 0.004& 0.003\\
&sisVIVE&0.539&    -& 0.428& 0.957&&0.589&    -& 0.479& 1\\
&Post-Alasso&0.532 &0& 1& 0&&0.530& 0& 0.996& 0\\
&WIT&0.046& 0.818 &0.068& 0.065&&0.024& 0.948 & 0.004& 0.020\\\cline{3-6} \cline{8-11}

\multirow{9}{*}{{1(II)}}& TSLS& 1.098& 0&     -&  -&&1.111&  0&     - &    -\\
&LIML&7.437& 0.292&-&- && 7.798& 0.030&-&-\\
& oracle-TSLS &0.072& 0.938 &   -  &   -&&0.046& 0.948 &    -  &   -\\
& oracle-LIML&0.072& 0.950&     -  &   -&&0.046 & 0.958 &    -  &   -\\
&TSHT& 0.110& 0.914& 0.122& 0.585 &&0.742& 0.598& 0.423& 0.712\\
&CIIV&0.099& 0.724& 0.088 & 0.642 &&4.321& 0.334 & 0.360& 0.824\\
&sisVIVE& 0.259&   -& 0.018& 0.234&&0.154&    - &0& 0.168\\
&Post-Alasso& 1.831& 0.016& 0.429& 0.258&&3.533& 0& 0.560& 0.387\\
&WIT&0.079& 0.914& 0.016& 0.034&& 0.049& 0.920 & 0.016& 0.016\\\hline\hline
\end{tabular}
\end{table}

Table \ref{tab:simu1} presents the detailed results. In Case 1(I), high FPR or FNP indicates that sisVIVE and Post-Alasso mistarget $\boldsymbol{\alpha}^*$ because of lack of majority rule and SAIS holds. Their performances do not improve much with sample size $n$. Due to finite strong IVs, WIT performs similarly to TSHT, CIIV, and oracle-LIML. In low dimension settings, the oracle-TSLS is very close to oracle-LIML. Case 1(II) shows how weak IVs break the strong IVs-based plurality rule. As shown in Example \ref{Example 1}, TSHT and CIIV worsen in terms of all measures though the sample size increase from $n=200$ to $500$. Post-Alasso also fails due to the failure of the majority rule. As the analysis in Example \ref{Example 2}, such a mixture of weak IVs does not impede penalized methods. The WIT estimator outperforms the other estimators when $n=200$ and approaches oracle-LIML when $n=500$. Interestingly, the comparably low FPR and MAD imply sisVIVE correctly target true $\boldsymbol{\alpha}^*$ since the additional requirement of matching objective function \eqref{additional con} happens to hold in this example. However, its FNR and MAD are worse than WIT due to the conservative cross-validation tuning strategy and the non-ignorable bias of Lasso, respectively.

Further, for closer comparison, we present a replication of the simulation design considered in \cite{windmeijer2019confidence} and its weak IVs variant:
\begin{enumerate}[itemsep=2pt,topsep=0pt,parsep=0pt]
    \item[] Case 1(III) : $\boldsymbol{\gamma}^* = (\boldsymbol{0.4}_{21})^{\top}$ and $\boldsymbol{\alpha}^* = (\boldsymbol{0}_{9},\boldsymbol{0.4}_{6},\boldsymbol{0.2}_{6})^{\top}$.
    \item[] Case 1(IV) : $\boldsymbol{\gamma}^* = (\boldsymbol{0.15}_{21})^{\top}$ and $\boldsymbol{\alpha}^* = (\boldsymbol{0}_{9},\boldsymbol{0.4}_{6},\boldsymbol{0.2}_{6})^{\top}$.
\end{enumerate}
We now vary sample size $n = 500$ to $1000$ and fix $\sigma_\eta = 1$ to strictly follow their design. Between them, Case 1(III) corresponds to the exact setting, while Case 1(IV) scales down the magnitude of $\boldsymbol{\gamma}^*$ to introduce small coefficients in the first-stage.

\begin{table}
\caption{\label{tab:simu1(2)}Simulation results in low dimension: A replication of experiment \citep{windmeijer2019confidence}}
\centering
\begin{tabular}{lccccccccccc}
\hline\hline
\multicolumn{1}{l}{\multirow{2}{*}{Case}} & \multirow{2}{*}{Approaches} & \multicolumn{4}{c}{$n = 500$}                         && \multicolumn{4}{c}{ $n=1000$}                        \\
\multicolumn{1}{l}{}                    &                     & MAD                        &    CP    &FPR&FNR       && MAD             &CP           & FPR  &FNR              \\\cline{3-6} \cline{8-11}

\multirow{9}{*}{{1(III)}}& TSLS& 0.436& 0&        -&         - &&0.435& 0&        -   &       -\\
& LIML & 0.729&   0& -&-&&0.739&0&-&-\\
& oracle-TSLS& 0.021& 0.936& -&-&&0.014&  0.942&-&-\\
& oracle-LIML&0.021& 0.932  &      -   &      -&&0.014& 0.944 &        -        &  -\\
&TSHT&0.142& 0.404& 0.398& 0.150&&0.016& 0.924& 0.023& 0.004\\
&CIIV&0.037& 0.710& 0.125& 0.032&&0.017& 0.894& 0.031& 0.002\\
&sisVIVE&0.445&   -& 0.463& 0.972&&0.465&    -& 0.482& 0.999\\
&Post-Alasso&0.436& 0& 1& 0&&0.435& 0& 0.999& 0\\
&WIT&0.036& 0.708& 0.121& 0.099&&0.016& 0.910& 0.020& 0.027\\\cline{3-6} \cline{8-11}

\multirow{9}{*}{{1(IV)}}& TSLS& 1.124& 0&        -   &      - &&1.144& 0&        -  &       -\\
& LIML & 1.952&0&-&-&&1.976&0&-&-\\
& oracle-TSLS& 0.060& 0.936&-&-&&0.044&0.942&-&-\\
& oracle-LIML&0.056& 0.948  &      -  &       -&&0.042& 0.962 &      -     &    -\\
&TSHT&0.532& 0.058 &0.342& 0.457&&0.155& 0.660& 0.310& 0.208\\
&CIIV&1.213& 0.224 & 0.337& 0.670&&0.100& 0.574 &0.300& 0.426\\
&sisVIVE&1.101&    -& 0.392& 0.936&&1.175&    - & 0.428& 0.996\\
&Post-Alasso&1.112& 0& 0.945& 0.010&&1.029& 0& 0.652& 0.205\\
&WIT&0.102& 0.634& 0.198& 0.220&&0.047& 0.898 & 0.051& 0.064\\\hline\hline
\end{tabular}
\end{table}

Table \ref{tab:simu1(2)} shows the results. In Case 1(III), CIIV outperforms TSHT because CIIV can utilize available information better \citep[Section 7]{windmeijer2019confidence}. The WIT estimator performs similar to CIIV and approaches oracle-LIML. sisVIVE and Post-Alasso fail again due to a lack of majority rule. In Case 1(IV), scaling down the first-stage coefficients causes some problems for CIIV and TSHT, since the first-stage selection thresholding $\sigma_\eta\sqrt{2.01\log p/n} = 0.111<0.15$, which might break the plurality rule numerically. TSHT and CIIV perform poorly when $n = 500$ and improve when $n =1000$ when the issue of violating the plurality rule is mitigated. Among penalized methods, sisVIVE and Post-Alasso mistarget and perform like TSLS because an additional requirement for sisVIVE \eqref{additional con} and majority rule fail simultaneously. Distinguished from them, the WIT estimator outperforms with acceptable MAD when $n=500$. The FPR and FNR improve when the sample size increases. Fig. \ref{fig:simulation} presents all replications in Case 1(IV) when $n = 1000$. It shows that WIT is nearly oracle besides the mild number of incorrect selections. By contrast, CIIV and TSHT fail in selections more frequently.

\begin{figure}
    \centering
    \includegraphics[width = 14.5cm]{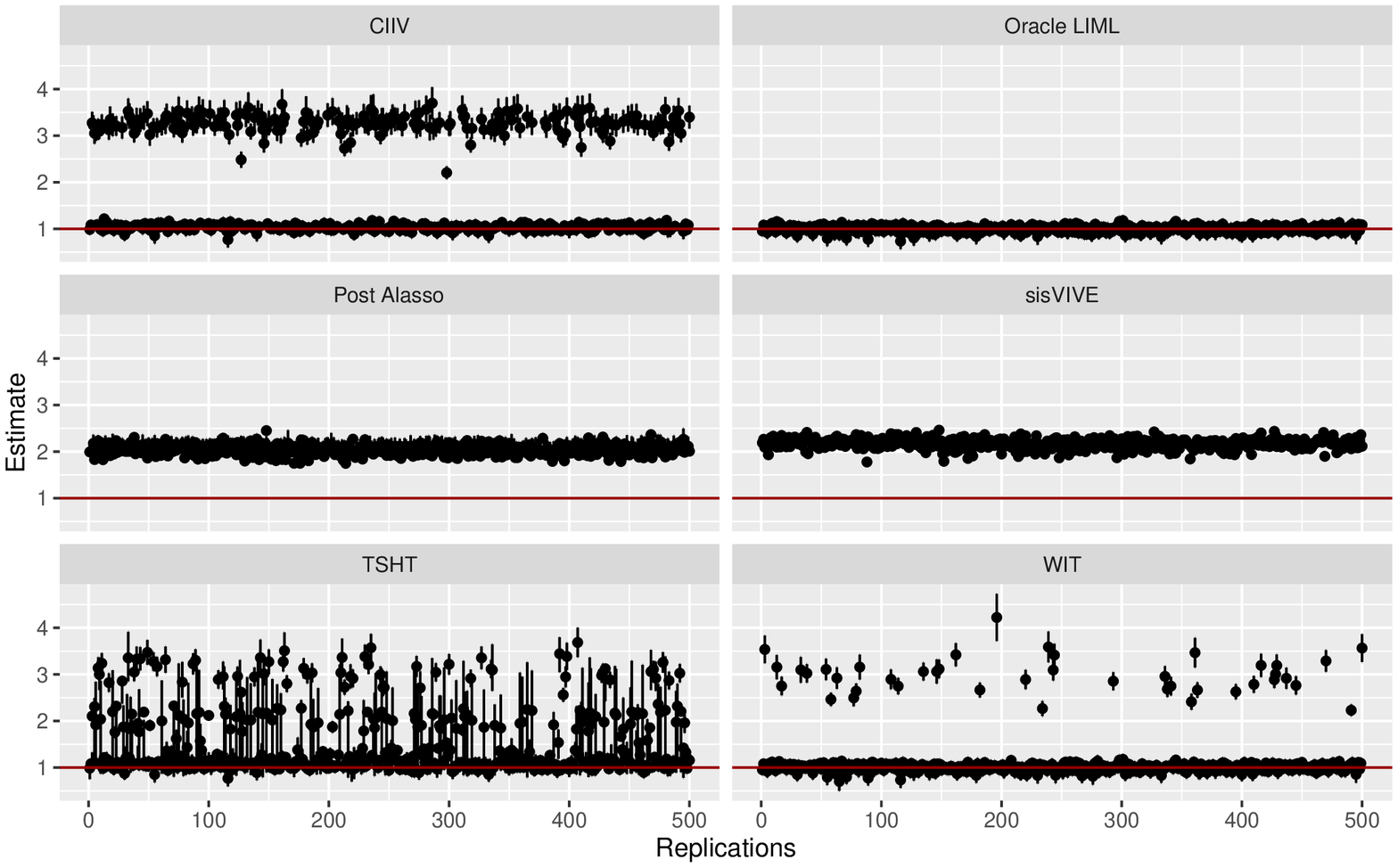}
    \caption{Scatter plot of estimations of $\beta^*$ with confidence intervals of Case 1(IV) for $n = 1000$. The red solid lines show the true value of $\beta^* =1$.}
    \label{fig:simulation}
\end{figure}

\subsection{Case 2: High dimension (many IVs)}
To assess performance in many IVs, we consider the following examples:\\
{  Case 2(I) ($p = 0.5 n$): \\$\boldsymbol{\gamma}^* = (\boldsymbol{2\sqrt{\log p/n^{1-\delta}}})_{p}^{\top}$;
$\boldsymbol{\alpha}^* = (\boldsymbol{0}_{0.6p},\boldsymbol{2\sqrt{\log {p}/n^{1-\delta}}}_{0.4p})^{\top}$.\\
Case 2(II) ($p = 0.6n$):\\ $\boldsymbol{\gamma}^* = (\boldsymbol{2\sqrt{\log{p}/n^{1-\delta}}})_{p}^{\top}$;  $\boldsymbol{\alpha}^* = (\boldsymbol{0}_{0.5p},\boldsymbol{-2\sqrt{\log{p}/n^{1-\delta}}}_{0.2p},\boldsymbol{4\sqrt{\log{p}/n^{1-\delta}}}_{0.3p}$).\\
Among Case 2(I) and (II), $\delta = 0.01$ is taken.} The numbers of valid and invalid IVs are growing with the sample size. To verify the theoretical result, we maintain the ratio of concentration parameter to sample size $n$ at a low constant level, i.e.\ $\mu_n/n = 0.5$, by adjusting $\sigma^2_\eta$. {  We vary the sample size $n$ from $500$ to $1000$, and let the first-stage coefficients in $\boldsymbol{\gamma}$ and elements in $\boldsymbol{\alpha}_{\mathcal{V}^c}$ go to $0$ simultaneously, which aligns with the case in Proposition \ref{high_dim theorem}.} Due to the computational burden in CIIV for many IVs case, we omit its results.

\begin{table}
\caption{\label{tab:simu2}Simulation results in high dimension (many IVs)}
\centering
\begin{tabular}{lccccccccccc}
\hline\hline
\multicolumn{1}{l}{\multirow{2}{*}{Case}} & \multirow{2}{*}{Approaches} & \multicolumn{4}{c}{$n = 500$}                         && \multicolumn{4}{c}{ $n=1000$}                        \\
\multicolumn{1}{l}{}                    &                     & MAD                        &    CP    &FPR&FNR       && MAD             &CP           & FPR  &FNR              \\\cline{3-6} \cline{8-11}

\multirow{8}{*}{{2(I)}}& TSLS&  0.394& 0&     -&           - &&0.385& 0&          -&          -\\
& LIML&1.355& 0&-&-&& 1.299& 0&-&- \\
& oracle-TSLS&0.040& 0.625&-&-&& 0.032& 0.422&-&-\\
& oracle-LIML&0.006& 0.937 &     -&          -&&0.003& 0.953&           -&          -\\
&TSHT&0.262& 0 &0.141& 0.743  &&0.273& 0 &0.162& 0.752\\
&sisVIVE&0.133&    -& 0& 0.143 &&0.125&    -& 0& 0.040\\
&Post-Alasso&0.394& 0& 1& 0&&0.403& 0& 0.083& 0\\
&WIT&0.004& 0.936& 0& 0.007 &&0.003& 0.945& 0& 0.002 \\\cline{3-6} \cline{8-11}

\multirow{8}{*}{{2(II)}}& TSLS&  0.421& 0 &         -&        -&&0.389 & 0&           -&       -\\
& LIML &10.538& 0.307 &-&-&&10.425& 0.269&-&-\\
&oracle-TSLS& 0.054& 0.307&-&-&&0.012& 0.192&-&-\\
& oracle-LIML&0.004& 0.923    &     -&      -&&0.003& 0.934 &          -&        -\\
&TSHT&0.104& 0.384& 0.027& 0.828 &&0.048 &0.461& 0.030& 0.841\\
&sisVIVE&0.051&    -& 0& 0.054&& 0.049 &    -& 0& 0.031\\
&Post-Alasso&0.420& 0& 1& 0&&0.389& 0& 1& 0\\
&WIT&0.006& 0.769& 0.003& 0.017 &&0.005& 0.882 &0.002& 0.018\\\hline\hline
\end{tabular}
\end{table}

{  Table \ref{tab:simu2} provides the detailed estimation results. Case 2(I) satisfies the majority rule, and only two groups are present. All weak IVs narrow available choices in TSHT. When $n =1000$, a low FPR but high FNR indicates that TSHT only selected a limited number of valid IVs. Both low FPR and FNR imply sisVIVE targets the correct solutions. Meanwhile, a high MAD shows that the TSLS-based method is biased in many IVs. The majority rule holds to ensure Post-Alasso is consistent. However, when $n=500$, Post-Alasso is severely biased with relatively high FPR. It may result from the sensitivity problem of the initial estimator of Post-Alasso with weak IVs. The WIT estimator behaves similarly to oracle-LIML and achieves the best performance in every measure, even when $n=500$. On the contrary, oracle-TSLS has a larger bias than LIML in many IVs cases, and the coverage rate is poor.

Case 2(II) has more invalid IVs with the sparsest rule. In terms of valid IV selection, TSHT has subpar performance in this case since the strong IVs-based plurality rule is unlikely to hold. sisVIVE correctly identifies solutions by chance in this example. However, its estimate suffers from the bias of Lasso and TSLS-based bias in many IVs. Without the majority rule, Post-Alasso targets the wrong solution as the initial median estimator is inconsistent. On the other hand, WIT works when the majority rule does not hold and circumvents the bias from TSLS. Thus, WIT outperforms the other methods in many IVs cases and approaches oracle-LIML when the sample size increases. The numerical results demonstrate the capacity that WIT can work in many IV cases under the conditions that $\alpha_j^*$ and $\gamma_j^*$ share the same rate going to zero. These simulation studies strongly support the conclusion in Proposition \ref{high_dim theorem}}.

\section{The Effect of BMI on Diastolic Blood Pressure}\label{sec:real}
This section illustrates the usefulness of the proposed WIT estimator for the method of Mendelian Randomization. We implement the WIT estimator to obtain an estimate of the causal effect of BMI on diastolic blood pressure and compare it to the other estimators, OLS, sisVIVE, Post-Alasso, TSHT and CIIV. This comparative analysis is designed to exemplify the efficiency and robustness of the WIT estimator in various research scenarios.

{
\subsection{Data Description}

In Mendelian Randomization genetic markers, called single nucleotide polymorphisms (SNPs), function as IVs for the identification and estimation of causal effects of modifiable phenotypes on outcomes \citep{von2016genetic}. This research design utilizes the random distribution of alleles at conception \citep{locke2015genetic}.

Nonetheless, SNPs can be invalid IVs. This is primarily due to pleiotropy, which refers to the potential of genetic variants to be associated with multiple phenotypes, but can also be due to influences such as linkage disequilibrium and population stratification. Unfortunately, these violations are typically unidentified prior to their selection. Furthermore, the correlation between SNPs and treatments is often weak, and so we are dealing with the problem of potentially many weak instruments.

In our study, we analyzed data from 105,276 individuals from the UK Biobank to investigate the effect of body mass index (BMI) on diastolic blood pressure (DBP). As suggested by \cite{windmeijer2019use,windmeijer2019confidence}, we used 96 SNPs as potential IVs for BMI. To account for skewness, we applied log-transformations to both BMI and DBP. We consider the same model specification detailed in Section 8 of \cite{windmeijer2019confidence}.  The model also included age, age squared, and sex as explanatory variables, along with 15 principal components of the genetic relatedness matrix. 

\subsection{Result and Analysis}
Table \ref{Table real} provides the estimation results of the effect of $\log$(BMI) on $\log$(DBP). 

\begin{table}\caption{\label{Table real}Empirical results, the effect  of $\log$(BMI) on $\log$(DBP)}
\centering
\begin{adjustbox}{max width = 15cm}
    \begin{threeparttable}
    \begin{tabular}{ccccccc}
    \hline\hline
         & $\hat{\beta}$ &  SE$(\hat{\beta})$& 95\% CI &  \# Valid IVs $\widehat{\mathcal{V}}$& \# Relevant IVs selected & p-value Sargan Test\\\cline{2-7}
        OLS & 0.206 & 0.002 & (0.202, 0.210)&-&-&-\\
        
        TSLS&0.087 &  0.016 &(0.055, 0.119)&96&-&2.05e-19\\
        TSHT$^*$&0.087 &  0.016 &(0.055, 0.119)&96&-&2.05e-19\\
        TSHT&0.098 &  0.016 &(0.066, 0.130)& 61&62&5.29e-14\\
        CIIV$^*$&0.140 &  0.019 &(0.103, 0.177)&83&-&0.011\\
        CIIV&0.174 &  0.020 &(0.135, 0.213) &49&62&0.014\\
        sisVIVE& 0.111 & n.a. &n.a.&76&-&0.064\\
        Post-ALasso&0.163 & 0.018 & (0.128, 0.198)&85&-&0.013\\
        WIT&0.123 &  0.020 &(0.083,0.163)& 81&-&0.140\\
        \hline
    \end{tabular}
    	    	\begin{tablenotes}
\item[]\scriptsize Note: Sample size $n = 105,276$; $ p = 96$ potential SNPs (IVs). Two-stage methods were used for both TSHT and CIIV: these methods first select strong IVs and then pick valid IVs from this subset. TSHT$^*$ and CIIV$^*$ represent the methods without employing a first-stage thresholding process. The sisVIVE does not report standard error or confidence interval. The symbol ``-'' denotes that no first-stage selection was performed, and all the original $96$ potential IVs were directly used.
\end{tablenotes}
\end{threeparttable}
\end{adjustbox}
\end{table}

The OLS estimate of 0.206 is potentially severely biased due to endogeneity issues, such as reverse causality or latent confounders.

Regarding the IV methods, we first denote TSHT$^*$ and CIIV$^*$ to indicate the estimation results without first-stage thresholding. For TSLS and TSHT$^*$, both without first-stage thresholding, they yield identical estimates of 0.087. However, the near-zero Sargan test p-value strongly rejects the model, implying TSHT's inability to detect potentially invalid IVs among weak ones. Despite the slight improvement of TSHT in discerning invalid IVs through first-stage thresholding, it still leaves some possibly invalid instruments in the model, as indicated by the slightly increased but still very small Sargan test p-values.

Compared to TSHT, CIIV identified thirteen strong and invalid IVs, leading to an estimate 0f 0.174 and a Sargan p-value of 0.011. CIIV$^*$ detected the same (as CIIV) strong invalid IVs with a similar Sargan p-value. CIIV$^*$ adjusted the estimate to 0.140.

The sisVIVE approach yielded an estimate of 0.111 which is lower than that obtained with the CIIV method. First, we observe that sisVIVE picks 20 invalid IVs (the highest count among comparison methods). Second, a Sargan p-value of 0.064 indicates that it overly penalizes invalid IVs while missing some true targets. In contrast, the Post-Alasso estimate is equal to 0.163, with the minimum number of 11 identified invalid IVs, and a Sargan p-value of 0.013.

The WIT estimator produced an estimate of 0.123, markedly lower than that of CIIV (and a little lower than CIIV$^*$), accompanied by the highest Sargan test p-value of 0.140. The invalid IVs detected by WIT encompassed all the relevant and invalid IVs identified by CIIV without first-stage thresholding. Moreover, compared to CIIV$^*$, WIT effectively identified two additional individually weak and invalid IVs. It thus significantly improved the Sargan p-value and estimation while capturing all valid yet weak IVs' information. In hindsight, sisVIVE penalized too many weak and valid IVs, leading to a loss in efficiency. Notably, the minimum validity violation $|\hat{\boldsymbol{\alpha}}^{MCP}|$ is 0.0014, aligned with the average magnitude of first-stage coefficients. In summary, the WIT method minimizes the risk of including invalid IVs while fully utilizing all valid IVs.

To summarize, the WIT estimator serves as a powerful tool for estimating treatment effects in biomedical research using SNPs as potential IVs. Its robustness to invalid and weak IVs makes it highly suitable for Mendelian Randomization applications, where there are potentially many weak IVs with uncertain validity.  

}

\section{Conclusion}\label{sec:con}
We extended the study of IV models with unknown invalid IVs to allow for many weak IVs. We provided a complete framework to investigate the identification issue of such models. Sticking to the sparsest rule, we proposed the surrogate sparsest penalty that fits the identification condition. We proposed a novel WIT estimator that addresses the issues that can lead to poor performance of sisVIVE and Post-Alasso, and can outperform the plurality rule-based TSHT and CIIV. Simulations and real data analysis support the theoretical findings and the advantages of the proposed method over existing approaches.

\section*{Acknowledgements}

\bibliographystyle{rss}

\bibliography{manuscript}

\newpage
\appendix
\renewcommand{\theequation}{\Alph{section}\arabic{equation}} 
\renewcommand{\thesection}{Appendix \Alph{section}}
\renewcommand{\thesubsection}{Appendix \Alph{section}\arabic{subsection}}
\renewcommand{\thetable}{A\arabic{table}}
\renewcommand{\thelemma}{A\arabic{lemma}}
\setcounter{page}{1}
\setcounter{footnote}{0}
\setcounter{section}{0}%
\setcounter{subsection}{0}%
\setcounter{equation}{0} %
\setcounter{table}{0}%
\setcounter{lemma}{0}

\begin{center}
	
	{\bf \large Web-based supporting materials for ``On the instrumental variable estimation with many weak and  invalid instruments''}\\
	\hspace{.2cm}\\
	Yiqi Lin\textsuperscript{a}, Frank Windmeijer\textsuperscript{b}, Xinyuan Song\textsuperscript{a}, Qingliang Fan\textsuperscript{c}\\
	\hspace{.2cm}\\
	
	\textsuperscript{a}Department of Statistics, The Chinese University of Hong Kong, Hong Kong.\\
	\textsuperscript{b}Department of Statistics, University of Oxford, Oxford, U.K.\\
	\textsuperscript{c}Department of Economics, The Chinese University of Hong Kong, Hong Kong.\\
	
\end{center}

This supplementary material mainly includes the following parts: \ref{A} provides additional details of the main paper. 	\ref{B} contains all technical proofs. Throughout the online supplementary material, we allow constant $C$ to be a generic positive constant that may differ in different cases.	
\section{Additional demonstrations}\label{A}

\subsection{Discussion of $\boldsymbol{Z}$ and its non-linear transformation}\label{append:Z}
In Section \ref{sec:model_1}, it is mentioned that $\boldsymbol{Z}_{i.}$ can be non-linear transformations of original variables such as polynomials and B-splines to form a high-dimensional model and provide flexible model fitting. In general, Eq. \eqref{Structure} can be reformulated as
\begin{eqnarray}
Y_{i} &=&D_{i} \beta+w_1(\boldsymbol{Z}_{i.})+\epsilon_{i} \approx  D_{i} \beta+\boldsymbol{\phi}(\boldsymbol{Z}_{i .})^{\top} \boldsymbol{\alpha}+\epsilon_{i} ,\label{basis expansion} \\
D_i &=& w_2(\boldsymbol{Z}_{i.})+ \eta_i \approx\boldsymbol{\phi}(\boldsymbol{Z}_{i.})^{\top}\boldsymbol{\gamma}+\eta_i, \label{first-stage transform}
\end{eqnarray}
where $w_1(\boldsymbol{Z}_{i.})\neq 0$ and $w_2(\boldsymbol{Z}_{i.})\neq0$ are different unknown functions that can be approximated by the same basis expansion through $\boldsymbol{\phi}(\cdot)$. With the same argument in \eqref{Structure}, the $\boldsymbol{\alpha}$ in \eqref{basis expansion} should consist of zero and non-zero terms to ensure the existence of an identifiable sub-model.

A subtle yet important point here is that we should use the same basis $\boldsymbol{\phi}(\cdot)$ to approximate different $w_1$ and $w_2$. Otherwise, using different $\boldsymbol{\phi}_1$ and $\boldsymbol{\phi}_2$ in \eqref{basis expansion} and \eqref{first-stage transform}, respectively, requires  strong prior information to justify why the components in $\{\boldsymbol{\phi}_2(\boldsymbol{Z}_{i.})\} \backslash \{\boldsymbol{\phi}_1(\boldsymbol{Z}_{i.})\}$ can be treated as valid IVs. In the following we show that this could be a stringent assumption for model identification. For simplicity of illustrating this point, we consider $\{\boldsymbol{\phi}_1(\boldsymbol{Z}_{i.})\}\subseteq \{\boldsymbol{\phi}_2(\boldsymbol{Z}_{i.})\}$ firstly. It can be extended to general $\boldsymbol{\phi}_1 \neq \boldsymbol{\phi}_2$ since coefficients of $\{\boldsymbol{\phi}_1(\boldsymbol{Z}_{i.})\} \backslash \{\boldsymbol{\phi}_2(\boldsymbol{Z}_{i.})\}$ can be estimated directly by their moment conditions. Nonetheless, in the simplified case, it is easy to see that
\begin{eqnarray*}
Y_{i} &=&  D_{i} \beta^*+\boldsymbol{\phi}_1(\boldsymbol{Z}_{i .})^{\top} \boldsymbol{\alpha}^*+\epsilon_{i}\notag\\
&=& D_{i} \beta^*+\boldsymbol{\phi}_1(\boldsymbol{Z}_{i .})^{\top} \boldsymbol{\alpha}^*+(\{\boldsymbol{\phi}_2(\boldsymbol{Z}_{i.})\} \backslash \{\boldsymbol{\phi}_1(\boldsymbol{Z}_{i.})\})^{\top} \boldsymbol{0}+\epsilon_{i} \notag \\
&=&D_{i} \beta^* + \boldsymbol{\phi}_2(\boldsymbol{Z}_{i.})^{\top} (\boldsymbol{\alpha}^{*\top},\boldsymbol{0}^\top)^\top +\epsilon_{i},
\\
D_i &=& \boldsymbol{\phi}_2(\boldsymbol{Z}_{i.})^{\top}\boldsymbol{\gamma}^*+\eta_i,
\end{eqnarray*}
where we denote the $\boldsymbol{\alpha}^{\text{new}*} = (\boldsymbol{\alpha}^{*\top},\boldsymbol{0}^\top)^\top$.
Therefore, by Theorem \ref{Theorem 1}, we immediately know that all other possible $\tilde{\boldsymbol{\alpha}}^{\text{new}}$ in the remaining valid DGPs must have non-zero coefficients in $\{\boldsymbol{\phi}_2(\boldsymbol{Z}_{i.})\} \backslash \{\boldsymbol{\phi}_1(\boldsymbol{Z}_{i.})\}$ but are impossible to be selected because zero coefficients for $\{(\boldsymbol{\phi}_2(\boldsymbol{Z}_{i.})\} \backslash \{\boldsymbol{\phi}_1(\boldsymbol{Z}_{i.})\}$ are set by default  in \eqref{basis expansion}. The identification condition essentially assumes the fixed functional form of $\boldsymbol{\phi}_1$ in \eqref{basis expansion}.

\subsection{SAIS condition and intuition of why MCP can circumvent it}\label{sec:A2}

Correct selection of valid IVs is a more subtle and important issue in IV content. Recall \cite{windmeijer2019use} indicated that failure of consistent variable selection of sisVIVE is assured if the SAIS condition holds. The SAIS condition was first proposed in \cite{windmeijer2019use} derived from Irrepresentable Condition (IRC) directly. 
IRC is known as (almost) necessary and sufficient condition for variable selection consistency of Lasso \citep{zhao2006model} for $n^{-1}\tilde{\boldsymbol{Z}}^\prime \tilde{\boldsymbol{Z}}$, i.e.,
 \begin{equation}
    \max _{j \in \mathcal{V}^{*}}\left\|\left(\tilde{\boldsymbol{Z}}_{\mathcal{V}^{c*}}^{\top} \tilde{\boldsymbol{Z}}_{\mathcal{V}^{c*}}\right)^{-1} \tilde{\boldsymbol{Z}}_{\mathcal{V}^{c*}}^\top \tilde{\boldsymbol{Z}}_{j}\right\|_{1} \leq \xi <1,~\text{for some}\,\, \xi \in [0,1)\label{irrepre}.
 \end{equation}

 In the standard Lasso regression problem, the IRC only relates to the design matrix and holds for many standard designs (see corollaries in \citealp{zhao2006model}). However, in the context of IV model, the IRC on $ \tilde{\boldsymbol{Z}}$ (instead of $\boldsymbol{Z}$) involves the first-stage signal estimate $\widehat{\boldsymbol{\gamma}}$, which further complicates the verifiability of the SAIS condition.
Typically, two-stage IV regression modeling exacerbates the difficulty of detecting valid IVs through penalized methods than support recovery problem in a simple linear model. Moreover, among the penalty functions, Lasso penalty aggravates the problem if the first-stage coefficients related SAIS condition hold.

The MCP penalty inherits a much weaker condition for oracle property than IRC that Lasso required \citep{zhao2006model}. Theorem 6 proposed in \citep{zhang2012general} has generalized the IRC in Lasso to a concave penalty in a linear regression problem. We briefly state the key result by defining two quantities:
 \begin{eqnarray*}
     \theta_{\text{select}} = \operatorname{inf}\Big\{\theta:\Big\|\Big(\frac{\boldsymbol{\tilde{Z}}_{\mathcal{V}^{c*}}^\top\boldsymbol{\tilde{Z}}_{\mathcal{V}^{c*}}}{n}\Big)^{-1}p_\lambda^{\prime}(\boldsymbol{\varphi}_{\mathcal{V}^{c*}}+\widehat{\boldsymbol{\alpha}}_{\mathcal{V}^{c*}}^{\text{or}})\Big\|_\infty \leq \theta \lambda,\forall\left\|\boldsymbol{\varphi}_{\mathcal{V}^{c*}}\right\|_{\infty} \leq \theta \lambda\Big\},\\
 \kappa_{\text{select}} = \operatorname{sup}  \Big\{ \|\boldsymbol{\tilde{Z}}_{\mathcal{V}^*}^\top\boldsymbol{\tilde{Z}}_{\mathcal{V}^{c*}}(\boldsymbol{\tilde{Z}}_{\mathcal{V}^{c*}}^\top\boldsymbol{\tilde{Z}}_{\mathcal{V}^{c*}})^{-1}p_\lambda^{\prime}(\boldsymbol{\varphi}_{\mathcal{V}^{c*}}+\widehat{\boldsymbol{\alpha}}_{\mathcal{V}^{c*}}^\text{or})\|_\infty/\lambda: \left\|\boldsymbol{\varphi}_{\mathcal{V}^{c*}}\right\|_{\infty} \leq  \theta_{\text{select}}\lambda\Big\},
 \end{eqnarray*}
where $\boldsymbol{\varphi}_{\mathcal{V}^{c*}}$ is a $|\mathcal{V}^{c*}|$-vector, and let the $\widehat{\boldsymbol{\alpha}}^{\text{or}}$ to be the oracle estimate, i.e., $\widehat{\boldsymbol{\alpha}}_{\mathcal{V}^{c*}}^{\text{or}} = (\boldsymbol{\tilde{Z}}_{\mathcal{V}^{c*}}^\top\boldsymbol{\tilde{Z}}_{\mathcal{V}^{c*}})^{-1}\boldsymbol{\tilde{Z}}_{\mathcal{V}^{c*}}\boldsymbol{Y}$ and $\widehat{\boldsymbol{\alpha}}_{\mathcal{V}^*}^{\text{or}} = \boldsymbol{0}$. The extended IRC for concave penalty required $\kappa_{\text{select}}<1$
 as the most crucial one to achieve selection consistency. When replacing MCP penalty $p_\lambda^{\text{MCP}}(\boldsymbol{\alpha})$ with Lasso penalty $\lambda\|\boldsymbol{\alpha}\|_1$ whose coordinate sub-derivative lies in $[-\lambda,\lambda]$, extended condition will be reduced to  $\kappa_{\text{select}}(\boldsymbol{\alpha}) =  \|\boldsymbol{\tilde{Z}}_{\mathcal{V}^{*}}^\top\boldsymbol{\tilde{Z}}_{\mathcal{V}^{c*}}(\boldsymbol{\tilde{Z}}_{\mathcal{V}^{c*}}^\top\boldsymbol{\tilde{Z}}_{\mathcal{V}^{c*}})^{-1}\|_\infty <1$ as identical to IRC of Lasso \eqref{irrepre}. By the feature of nearly unbiasedness characteristic of MCP, $p_n^\prime(t) = 0 \,\,\forall |t|>\lambda\rho$. Once the mild condition $\min_{j\in\mathcal{V}^{c*}} |\widehat{{\alpha}}_j^{\text{or}}| > \lambda\rho$ holds, it implies $\theta_{\text{select}} = 0$ with $\boldsymbol{\varphi}_{\mathcal{V}^{c*}} = \boldsymbol{0}$ and $\kappa_{\text{select}} = 0$ sequentially. Thus the extended IRC $\kappa_{\text{select}}<1$ holds automatically for MCP but does not for Lasso. Consequently, this property is desirable that MCP achieves exact support recovery regardless of the constraint of SAIS condition.

\subsection{I-LAMM Algorithm for MCP penalty}\label{sec:A3}
Theoretically, the proposed WIT estimator enjoys better performance under weaker regulation conditions for low- and high-dimension cases with weak IVs. \cite{fan2018lamm} proposed the iterative local adaptive majorize-minimization (I-LAMM) algorithm for non-convex regularized model. I-LAMM first contracts the initial values in the neighborhood of the optimum solutions to serve as a better sparse coarse initial $\widehat{\boldsymbol{\alpha}}^{(1)}$, then tightens it to the solution under precision tolerance. The computation is in polynomial time.

To be concrete, the I-LAMM algorithm combines the adaptive Local Linear Approximation (LLA) method and proximal gradient method or iterative shrinkage-thresholding (ISTA) algorithm \citep{daubechies2004iterative}. We adopt the I-LAMM to solve a sequence of optimization problems through LLA,
\begin{eqnarray}
\widehat{\boldsymbol{\alpha}}^{(t)} = \underset{\boldsymbol{\alpha}}{\operatorname{arg min}}\,\, \frac{1}{2n} \|\boldsymbol{Y}-\boldsymbol{\widetilde{Z}}\boldsymbol{\alpha}\|_2^2+\sum_{j=1}^p \left[p_\lambda^\prime(|\widehat{{\alpha}}_j^{(t-1)}|)|\alpha_j|\right].\label{iteration}
\end{eqnarray}
For $t = 1$ refers to contraction stage to obtain the better initial estimators which fall in the contraction region. And $t = 2,3,\ldots$, is the tightening stage, until it converges. {It is worth noting \cite{zou2008one} used one-step LLA iteration to achieve the oracle estimator based on OLS estimate as initial. Nevertheless, in our case, the OLS estimate is not achievable under the perfect multicollinearity of design matrix $\boldsymbol{\widetilde{Z}}$. \cite{windmeijer2019use} used the adaptive Lasso with embedded information in the median estimator as the one-step LLA to achieve the oracle
property.}  Compared with the one-step method, iterations of $\boldsymbol{\lambda}^{(t-1)}$ (defined as below) in I-LAMM circumvent the stringent condition (namely, the majority rule) to obtain the root $n$ initial estimator as the proper adaptive weight for ALasso.

For each iteration (\ref{iteration}), the ISTA method is implemented to achieve the closed-form updating formula, which is the reduced model in \cite{fan2018lamm}'s derivation. For a given iteration $t$, let $k = 0,1,2,\ldots$ denotes the iteration in the proximal gradient updating, thus,
\begin{eqnarray}
    \hat{\boldsymbol{\alpha}}^{(t,k)} &=& \underset{\boldsymbol{\alpha}}{\operatorname{arg min}}\,\Bigg\{ \frac{1}{2n} \left\|\boldsymbol{Y}-\boldsymbol{\widetilde{Z}}\widehat{\boldsymbol{\alpha}}^{(t,k-1)}\right\|_2^2-\frac{1}{n}\Big[\widetilde{\boldsymbol{Z}}(\boldsymbol{\alpha}-\widehat{\boldsymbol{\alpha}}^{(t,k-1)})\Big]^\top(\boldsymbol{Y}-\boldsymbol{\widetilde{Z}}\widehat{\boldsymbol{\alpha}}^{(t,k-1)})\notag\\
    &&\quad\quad\quad\quad+\frac{\phi}{2}\|\boldsymbol{\alpha}-\widehat{\boldsymbol{\alpha}}^{(t,k-1)}\|_2^2+
     \sum_{j=1}^p \left[p_\lambda^\prime(|\widehat{{\alpha}}_j^{(t-1)}|)|\alpha_j|\right]\Bigg\}\notag\\
       &=&  \underset{\boldsymbol{\alpha}}{\operatorname{arg min}} \,\Bigg\{ \frac{\phi}{2}\left\|\boldsymbol{\alpha}-\Big(\widehat{\boldsymbol{\alpha}}^{(t,k-1)}+\frac{1}{\phi n}\boldsymbol{\widetilde{Z}}^\top(\boldsymbol{Y}-\boldsymbol{\widetilde{Z}}\widehat{\boldsymbol{\alpha}}^{(t,k-1)})\Big)\right\|_2^2+ \sum_{j=1}^p \left[p_\lambda^\prime(|\widehat{{\alpha}}_j^{(t-1)}|)|\alpha_j|\right]\Bigg\}\notag\\
       &=& S\Big(\widehat{\boldsymbol{\alpha}}^{(t,k-1)}+\frac{1}{\phi n}\boldsymbol{\widetilde{Z}}^\top(\boldsymbol{Y}-\boldsymbol{\widetilde{Z}}\widehat{\boldsymbol{\alpha}}^{(t,k-1)}),\frac{1}{\phi}\boldsymbol{\lambda}^{(t-1)})\Big),
\end{eqnarray}
where $\phi$ should be no smaller than  the largest  eigenvalue of $\widetilde{\boldsymbol{Z}}^\top\widetilde{\boldsymbol{Z}}$, or simply put, to ensure the majorization and $S(\boldsymbol{x},\boldsymbol{a})$ denotes the component-wise soft-threshoding operator, i.e., $S(\boldsymbol{x},\boldsymbol{a})_j = \operatorname{sgn}(x_j)(|x_j|-a_j)_+$, with $\boldsymbol{\lambda}^{(t-1)} = \Big(p_n^\prime(|\widehat{{\alpha}}_1^{(t-1)}|),\ldots,p_n^\prime(|\widehat{{\alpha}}_p^{(t-1)}|)\Big)^\top$.

\begin{algorithm}[hbt]
\footnotesize
\caption{{\label{algorithm 2}}I-LAMM algorithm for $\widehat{\boldsymbol{\alpha}}$ with MCP penalty}
\hspace*{0.01in} {\bf Input:}
$\boldsymbol{Y},\widetilde{\boldsymbol{Z}},\widehat{\boldsymbol{\alpha}}^{(0)},\lambda,\phi = \lambda_{\text{max}}(\widetilde{\boldsymbol{Z}}^{\top}\widetilde{\boldsymbol{Z}}),\delta_c = 10^{-3},\delta_t = 10^{-5}$
\begin{algorithmic}[1]

   \For{$t = 1,2,\ldots$}
       \State $\boldsymbol{\lambda}^{(t-1)} = (p_\lambda^\prime(|\widehat{{\alpha}}_1^{(t-1)}|),\ldots,p_\lambda^\prime(|\widehat{{\alpha}}_p^{(t-1)}|))^\top$ \Comment{$p_n^{\top}$ is the derivative of MCP penalty}
       \State $\widehat{\boldsymbol{\alpha}}^{(t,0)} = \widehat{\boldsymbol{\alpha}}^{(t-1)}$

            \For{$k= 1,2,\ldots$}
                \State $\widehat{\boldsymbol{\alpha}}^{(t,k)} = S(\widehat{\boldsymbol{\alpha}}^{(t,k-1)}+\frac{1}{n\phi}\boldsymbol{\widetilde{Z}}^\top(\boldsymbol{Y}-\boldsymbol{\widetilde{Z}}\widehat{\boldsymbol{\alpha}}^{(t,k-1)}),\frac{1}{\phi}\boldsymbol{\lambda}^{(t-1)})) $ \Comment{LAMM updating}
                \State $\omega_{\boldsymbol{\lambda}^{(t-1)}}(\widehat{\boldsymbol{\alpha}}^{(t,k)}) =\min _{\boldsymbol{\xi} \in \partial|\widehat{\boldsymbol{\alpha}}^{(t,k)}|}\left\{\|-\boldsymbol{\widetilde{Z}}^\top(\boldsymbol{Y}-\boldsymbol{\widetilde{Z}}\widehat{\boldsymbol{\alpha}}^{(t,k)})+\boldsymbol{\lambda}^{(t-1)} \odot \boldsymbol{\xi}\|_{\infty}\right\}$
                \If{$\omega_{\boldsymbol{\lambda}^{(t-1)}}(\widehat{\boldsymbol{\alpha}}^{(t,k)})\leq \mathbf{1}(t=1)\delta_c+\mathbf{1}(t\neq 1)\delta_t$}
                \State $\widehat{\boldsymbol{\alpha}}^{(t)} = \widehat{\boldsymbol{\alpha}}^{(t,k)}$
                \State \textbf{break}
                \EndIf
            \EndFor
        \If{$\|\widehat{\boldsymbol{\alpha}}^{(t)}-\widehat{\boldsymbol{\alpha}}^{(t-1)}\|_{\infty} \leq \delta_t$}
        \State $\widehat{\boldsymbol{\alpha}} = \widehat{\boldsymbol{\alpha}}^{(t)}$
        \State \textbf{break}
        \EndIf
    \EndFor
\end{algorithmic}
\hspace*{0.02in} {\bf Output:}
$\widehat{\boldsymbol{\alpha}}$
\end{algorithm}

We adopt the ﬁrst order optimality condition as a stopping criterion in the sub-problem. Let
$$
\omega_{\boldsymbol{\lambda}^{(t-1)}}(\boldsymbol{\alpha})=\min _{\boldsymbol{\xi} \in \partial|\boldsymbol{\alpha}|}\left\{\|-\boldsymbol{\widetilde{Z}}^\top(\boldsymbol{Y}-\boldsymbol{\widetilde{Z}}{\boldsymbol{\alpha}})+\boldsymbol{\lambda}^{(t-1)} \odot \boldsymbol{\xi}\|_{\infty}\right\}
$$
as a natural measure of suboptimality of $\boldsymbol{\alpha}$, where $\odot$ is the Hadamard product. Once $\omega_{\hat{\boldsymbol{\lambda}}^{(t-1)}}(\boldsymbol{\alpha}^{(t,k)})\leq \delta$, where $\delta$ is a pre-determined tolerance level and is assigned $\delta_c$ and $\delta_t$ for contraction and tightening stages, respectively, we stop the inner iteration and take $\widehat{\boldsymbol{\alpha}}^{(t)} = \widehat{\boldsymbol{\alpha}}^{(t,k)}$ as the $\delta$-optimal solution in the sub-problem. Notably, it is an early stopped variant of ISTA method in each sub-problem to obtain $\widehat{\boldsymbol{\alpha}}^{(t)}$ from $\widehat{\boldsymbol{\alpha}}^{(t-1)}$. 
The following Algorithm \ref{algorithm 2} demonstrates the details of I-LAMM algorithm for WIT.

{  \subsection{Additional simulations for comparison}\label{sec: Add simu}
In this subsection, we further compare the performance of our proposed WIT estimator using the MCP penalty against the SCAD penalty \citep{fan2001variable}, under the same simulation settings. The setup of MCP penalty is demonstrated in the main text. For SCAD penalty, its derivative is defined as
\begin{equation*}
    p^\prime(\alpha) = \lambda \left[\mathbb{I}(\alpha\leq \lambda )+\frac{(a \lambda - \alpha)_+}{(a-1)\lambda}\mathbb{I}(\alpha>\lambda)\right],
\end{equation*}
where $a$ controls the convexity and is not sensitive. As suggested by \cite{fan2001variable}, we take the default $a = 3.7$, while $\lambda$ is to be tuned the same way as MCP. Their performances are shown in the table \ref{tab:simuA}.

\begin{table}
\caption{\label{tab:simuA} Simulation results comparisons of WIT using MCP and SCAD}
\centering
\begin{tabular}{lccccccccccc}
\hline\hline
\multicolumn{1}{l}{\multirow{2}{*}{Case}} & \multirow{2}{*}{Approaches} & \multicolumn{4}{c}{$n = 200$}                         && \multicolumn{4}{c}{ $n=500$}                        \\
\multicolumn{1}{l}{}                    &                     & MAD                        &    CP    &FPR&FNR       && MAD             &CP           & FPR  &FNR             \\ \cline{3-6} \cline{8-11}

\multirow{2}{*}{{1(I)}}& WIT(MCP)&  0.046& 0.818&     0.068&           0.065 &&0.024& 0.948&          0.004&          0.020\\
& WIT(SCAD)&0.050& 0.800&0.082&0.074&& 0.028& 0.926&0.008&0.022 \\
\cline{3-6} \cline{8-11}
\multirow{2}{*}{{1(II)}}& WIT(MCP)&  0.079& 0.914&     0.016&           0.034 &&0.049& 0.920&          0.016&          0.016\\
& WIT(SCAD)&0.077& 0.885 &0.014&0.030&& 0.052& 0.916 &0.014&0.018 \\
\cline{3-6} 
\hline\hline
\multicolumn{1}{l}{\multirow{2}{*}{Case}} & \multirow{2}{*}{Approaches} & \multicolumn{4}{c}{$n = 500$}                         && \multicolumn{4}{c}{ $n=1000$}                        \\ \multicolumn{1}{l}{}                    &                     & MAD                        &    CP    &FPR&FNR       && MAD             &CP           & FPR  &FNR             \\ \cline{3-6} \cline{8-11}
\multirow{2}{*}{{1(III)}}& WIT(MCP)&  0.036& 0.708&     0.121&           0.099 &&0.016& 0.910&          0.020&          0.027\\
& WIT(SCAD)&0.035& 0.715&0.114&0.092&& 0.018& 0.894&0.022&0.032 \\
\cline{3-6} \cline{8-11}
\multirow{2}{*}{{1(IV)}}& WIT(MCP)&  0.102& 0.634&     0.198&           0.220 &&0.047& 0.898&          0.051&          0.064\\
& WIT(SCAD)& 0.114&0.0550&0.254&0.226&& 0.047& 0.876&0.049&0.068 \\\cline{3-6} \cline{8-11}
\multirow{2}{*}{{2(I)}}& WIT(MCP)&  0.004& 0.936&     0&           0.007 &&0.003& 0.945&          0&          0.002\\
& WIT(SCAD)&0.008& 0.915&0.003&0.016&& 0.005& 0.958&0.001&0.001 \\
\cline{3-6} \cline{8-11}
\multirow{2}{*}{{2(II)}}& WIT(MCP)&  0.006& 0.769&     0.003&           0.017 &&0.005& 0.882&          0.002&          0.018\\
& WIT(SCAD)&0.007& 0.782&0.004&0.008&& 0.004& 0.911&0&0.013 \\
\hline\hline

\end{tabular}
\end{table}

In most cases, WIT with MCP is slightly better than the counterpart with SCAD penalty. Therefore, we suggest that when using WIT, the default choice of MCP is sufficient for most cases. The SCAD penalty can be used as an alternative, but there is no significant advantage in terms of performance.
}
\section{Proofs} \label{B}
In this section we provide the proofs of the theoretical results in the main text. To proceed, some lemmas from previous literature are needed. We restate some of those results for the convenience of the readers.
\subsection{Ancillary Lemmas}
\begin{lemma}\label{Lemma A1}
(Lemmas 1 and 2 in \citealp{bekker1994alternative} and Lemma A.1 in \citealp{kolesar2015identification}). Consider the quadratic form $Q=(M+U)^{\top} C(M+$ $U)$, where $M \in \mathbb{R}^{n \times S}, C \in \mathbb{R}^{n \times n}$ are non-stochastic, $C$ is symmetric and idempotent with rank $J_{n}$ which may depend on $n$, and $U=\left(u_{1}, \ldots, u_{n}\right)^{\top}$, with $u_{i} \sim_{i.i.d.}[0, \Omega]$. Let $a \in \mathbb{R}^{S}$ be a non-stochastic vector. Then,\\
(a) If $u_{i}$ has finite fourth moment:
$$
\begin{aligned}
\mathbb{E}[Q \mid C]=& M^{\top} C M+J_{N} \Omega \\
\operatorname{var}(Q a \mid C)=& a^{\top} \Omega a M^{\top} C M+a^{\top} M^{\top} C M a \Omega+\Omega a a^{\top} M^{\top} C M+M C M a a^{\top} \Omega +J_{N}\left(a^{\top} \Omega a \Omega+\Omega a a^{\top} \Omega\right) \\
&+d_{C}^{\top} d_{C}\left[\mathbb{E}\left(a^{\top} u\right)^{2} u u^{\top}-a^{\top} \Omega a a^{\top} \Omega-a^{\top} \Omega a \Omega\right]+2 d_{C}^{\top} C M a \mathbb{E}\left[\left(a^{\top} u\right) u u^{\top}\right] \\
&+M^{\top} C d_{C} \mathbb{E}\left[\left(a^{\top} u\right)^{2} u^{\top}\right]+\mathbb{E}\left[\left(a^{\top} u\right)^{2} u\right] d_{C}^{\top} C M
\end{aligned}
$$
where $d_{C}=\operatorname{diag}(C)$. If the distribution of $u_{i}$ is normal, the last two lines of the variance component equal zero.\\
(b) Suppose that the distribution of $u_{i}$ is normal, and as $n \rightarrow \infty,$
$$
M^{\top} C M / N \rightarrow Q_{C M}
, J_{n} / n \rightarrow \alpha_{r}
$$
where the elements $c_{i s}$ of $C$ may depend on $N$. Then,
$$
\sqrt{n}(Q a / n-\mathbb{E} Q a / n) \stackrel{d}{\rightarrow} \mathcal{N}(0, V),
$$
where $
V=a^{\top} \Omega a Q_{C M}+a^{\top} Q_{C M} a \Omega+\Omega a a^{\top} Q_{C M}+Q_{C M} a a^{\top} \Omega+\alpha_{r}\left(a^{\top} \Omega a \Omega+\Omega a a^{\top} \Omega\right).$
\end{lemma}

\subsection{Proof of Theorem \ref{Theorem 1}}
\begin{proof}
Firstly, we prove that procedure (\ref{procedure}) can generate the $G$ different groups of $\mathcal{P}_c$ satisfying the requirements. With direct calculation, $\mathcal{P}_c = \{\tilde{\beta}^{c},\tilde{\boldsymbol{\alpha}}^{c},\tilde{\boldsymbol{\epsilon}}^{c}\} = \{\beta^*+c,\boldsymbol{\alpha}^*-c\boldsymbol{\gamma}^*,\boldsymbol{\epsilon}-c\boldsymbol{\eta}\}$ and $E(\tilde{\boldsymbol{\epsilon}}^c) = E(\boldsymbol{\epsilon})-cE(\boldsymbol{\eta} ) = \boldsymbol{0}$ for $c = {\alpha^*_j}/{\gamma^*_j}, j\in \mathcal{I}_c$. Therefore, $\tilde{\boldsymbol{\alpha}}^c_{\mathcal{I}_c} = \boldsymbol{\alpha}_{\mathcal{I}_c}^*-c\boldsymbol{\gamma}^*_{\mathcal{I}_c} =  \boldsymbol{\alpha}_{\mathcal{I}_c}^*- \boldsymbol{\alpha}_{\mathcal{I}_c}^* = \boldsymbol{0}$ and $\tilde{\alpha}^c_j = \alpha^*_j-c\gamma^*_j \neq 0$ for $j \notin \mathcal{I}_c$. Going through all possible $c = \{\alpha_j^*/\gamma^*_j:j \notin \mathcal{V}^*\}$, we conclude that procedure (\ref{procedure}) has generated $G$ groups additional DGP$_c$. The exhaustive and mutually exclusive property of $\mathcal{I}_c$ and $\mathcal{V}^*$ is also guaranteed by its construction.

Second, we show the proof by contradiction that there is no other possible DGP with the same sparse structure and zero mean structural error.

Assume there is an additional DGP $\{\breve{\beta},\breve{\boldsymbol{\alpha}},\breve{\epsilon}\}$ differentiating with $\mathcal{P}_c$ and $\mathcal{P}_0$ but still has $E(\breve{\epsilon}) = \boldsymbol{0}$ and zero(s) component in $\breve{\alpha}$, i.e., $\breve{\mathcal{I}} = \{j:\breve{\boldsymbol{\alpha}_j} = 0\}\neq\varnothing$. By the property of exhaustive and mutually exclusiveness of $\mathcal{V}^*$ and $\{\mathcal{I}_c\}_{c\neq 0}$, WLOG, we assume $\mathcal{I}_g \cap \breve{\mathcal{I}}\neq \varnothing$ for some $g$ and DGP$_g =  \{\tilde{\beta}^{g},\tilde{\boldsymbol{\alpha}}^{g},\tilde{\boldsymbol{\epsilon}}^{g}\}$. Since $E(\breve{\boldsymbol{\epsilon}}) = E(\tilde{\boldsymbol{\epsilon}}^g) = \boldsymbol{0}$, it suffices to show the contradiction in terms of moment condition (\ref{moment}). Hence, $\{\tilde{\beta}^{g},\tilde{\boldsymbol{\alpha}}^{g}\}$ and $\{\breve{\beta},\breve{\boldsymbol{\alpha}}\}$ are solutions of $$\boldsymbol{\Gamma}^* = \boldsymbol{\alpha}+ \beta\boldsymbol{\gamma}^*.$$
For $j \in \mathcal{I}_g\cap \breve{\mathcal{I}}$, $\Gamma^*_j = \tilde{\beta}^g\gamma^*_j = \breve{\beta}\gamma_j^*$ derives $\tilde{\beta}^g = \breve{\beta}$. In turn, $\forall j \notin \mathcal{I}_g\cap \breve{\mathcal{I}}$,
$$\Gamma_j^*-\breve{\beta}\gamma_j^* = \tilde{{\alpha}}_j = \breve{\alpha}_j.$$ Thus, $\{\tilde{\beta}^{g},\tilde{\boldsymbol{\alpha}}^{g}\}$ and $\{\breve{\beta},\breve{\boldsymbol{\alpha}}\}$ are equivalent. So as $\tilde{\boldsymbol{\epsilon}}^g$ and $\breve{\boldsymbol{\epsilon}}$ since $$\tilde{\boldsymbol{\epsilon}}^c = \boldsymbol{Y}-\boldsymbol{D}\tilde{\beta}^g-\boldsymbol{Z}\tilde{\boldsymbol{\alpha}}^g = \boldsymbol{Y}-\boldsymbol{D}\breve{\beta}-\boldsymbol{Z}\breve{\boldsymbol{\alpha}} = \breve{\boldsymbol{\epsilon}}.$$

Hence DGP $\{\breve{\beta}, \breve{\boldsymbol{\alpha}}, \breve{\epsilon}\}$ is equivalent to DGP$_g$ and forms a contradiction. It concludes that procedure (\ref{procedure}) can produce all possible DGPs.\qed
\end{proof}

\subsection{Proof of Proposition \ref{Proposition 1}}
\begin{proof}
\begin{equation}
\begin{aligned}
     &\boldsymbol{\alpha}^* = \underset{{\mathcal{P}=\{\beta,\boldsymbol{\boldsymbol{\alpha}},\boldsymbol{\epsilon}\}\in \mathcal{Q}}}{\operatorname{arg min}} p_\lambda^{\text{pen}}(\boldsymbol{\alpha}),\\
    \iff& \sum_{j=1}^p p_\lambda^{\text{pen}}({\alpha_j^*})< \sum_{j=1}^p p_\lambda^{\text{pen}}({\tilde{\alpha}_j^c})\\
    \iff& \sum_{j\in \mathcal{V}^{c*}} p_\lambda^{\text{pen}}({\alpha_j^*})< \sum_{j\in {\mathcal{I}}_c^{c}} p_\lambda^{\text{pen}}({\tilde{\alpha}_j^c}),\label{iff_1}
\end{aligned}
\end{equation}
where $\mathcal{I}_c = \{j: \alpha^*_j/\gamma^*_j = c,c\neq 0\}$ and $\mathcal{I}_c^c$ is the complement of $\mathcal{I}_c$. By the Assumption 6:
$\boldsymbol{\alpha}^* = {\operatorname{arg min}}_{{\mathcal{P}=\{\beta,\boldsymbol{\boldsymbol{\alpha}},\boldsymbol{\epsilon}\}\in \mathcal{Q}}} \|\boldsymbol{\alpha}\|_0$, we have $|\mathcal{V}^{c*}|<|\mathcal{I}_c^c|$.

Define  $\pi(\mathcal{I}_c^c,\mathcal{V}^{c*})$
 as $|\mathcal{V}^{c*}|$-combination of $\mathcal{I}_c^c$. Now we rewrite (\ref{iff_1}) as
 \begin{equation}
     \sum_{j\in \mathcal{V}^{c*}} p_\lambda^{\text{pen}}({\alpha_j^*})-\sum_{k \in \pi(\mathcal{I}_c^c,\mathcal{V}^{c*})}p_\lambda^{\text{pen}}(\tilde{\alpha}^c_k)< \sum_{l \in\mathcal{I}_c^c/ \pi(\mathcal{I}_c^c,\mathcal{V}^{c*})}p_\lambda^{\text{pen}}(\tilde{\alpha}^c_l),\label{iff_2}
 \end{equation}
 and it should hold for any $|\mathcal{V}^{c*}|$-combination and $\mathcal{P}_0$.
 Note that the RHS in (\ref{iff_2}) are non-negative by definition of the 	penalty function and $|\mathcal{I}_c^c/ \pi(\mathcal{I}_c^c,\mathcal{V}^{c*})|>0$.

 For arbitrary $\pi(\mathcal{I}_c^c,\mathcal{V}^{c*})$ and $\mathcal{P}_0$, we consider the worst case where
 \begin{equation*}
    |\tilde{\alpha}_l^c| <|\tilde{\alpha}_k^c|\,\, \text{and}\,\, |\tilde{\alpha}^c_k|<|\alpha^*_j|,
 \end{equation*}
 for $\forall k \in \pi(\mathcal{I}_c^c,\mathcal{V}^{c*}),l \in\mathcal{I}_c^c/ \pi(\mathcal{I}_c^c,\mathcal{V}^{c*})$ and for each pair $(j,k)\in( \mathcal{V}^{c*}, \pi(\mathcal{I}_c^c,\mathcal{V}^{c*}))$. Using Taylor expansion and mean value theorem, we obtain
 \begin{equation*}
 \begin{aligned}
   &\Big\{ \sum_{(j,k)\in( \mathcal{V}^{c*}, \pi(\mathcal{I}_c^c,\mathcal{V}^{c*}))} |\alpha^*_j|-|\tilde{\alpha}^c_k|\Big\} \underset{(j,k)\in( \mathcal{V}^{c*}, \pi(\mathcal{I}_c^c,\mathcal{V}^{c*}))}{\operatorname{min}}[{p_\lambda^{\text{pen}}}^\prime({\bar{\xi}_{(j,k)}})] \\
   <& \sum_{j\in \mathcal{V}^{c*}} p_\lambda^{\text{pen}}({\alpha_j^*})-\sum_{k \in \pi(\mathcal{I}_c^c,\mathcal{V}^{c*})}p_\lambda^{\text{pen}}(\tilde{\alpha}^c_k)\\
     <& \sum_{l \in\mathcal{I}_c^c/ \pi(\mathcal{I}_c^c,\mathcal{V}^{c*})}p_\lambda^{\text{pen}}(\tilde{\alpha}^c_l)\\
     <& \|\tilde{\boldsymbol{\alpha}}^c_{\mathcal{I}_c^c/ \pi(\mathcal{I}_c^c,\mathcal{V}^{c*})}\|_1  \underset{l \in \mathcal{I}_c^c/ \pi(\mathcal{I}_c^c,\mathcal{V}^{c*})}{\operatorname{max}}[{p_\lambda^{\text{pen}}}^\prime({\bar{\xi}_l})] ,
\end{aligned}
 \end{equation*}
 where $\bar{\xi}_{(j,k)}\in (|\tilde{\alpha}^c_k|,|\alpha^*_j|)$ and $\bar{\xi}_l\in (0,|\tilde{\alpha}^c_l|)$. Thus, it leads to
 \begin{equation}
     \frac{ \underset{(j,k)\in( \mathcal{V}^{c*}, \pi(\mathcal{I}_c^c,\mathcal{V}^{c*}))}{\operatorname{min}}[{p_\lambda^{\text{pen}}}^\prime({\bar{\xi}_{(j,k)}})]}{ \underset{l \in \mathcal{I}_c^c/ \pi(\mathcal{I}_c^c,\mathcal{V}^{c*})}{\operatorname{max}}[{p_\lambda^{\text{pen}}}^\prime({\bar{\xi}_l})] }<\frac{\|\tilde{\boldsymbol{\alpha}}^c_{\mathcal{I}_c^c/ \pi(\mathcal{I}_c^c,\mathcal{V}^{c*})}\|_1 }{\Big\{ \sum_{(j,k)\in( \mathcal{V}^{c*}, \pi(\mathcal{I}_c^c,\mathcal{V}^{c*}))} |\alpha^*_j|-|\tilde{\alpha}^c_k|\Big\} }.\label{iff_3}
 \end{equation}
Because ${p_\lambda^{\text{pen}}}^\prime $ is free of $\mathcal{P}_0$,  the RHS in (\ref{iff_3}) could be much smaller than 1 in some extreme case leads to ${p_\lambda^{\text{pen}}}^\prime $ being a bounded monotone decreasing function.  That is to say, $p_\lambda^{\text{pen}}$ should be a concave penalty.

Consider another case where we keep $\boldsymbol{\alpha}^*$ fixed, but vary $\gamma^*$ to make $\tilde{\boldsymbol{\alpha}}^{c}$ to have the same order with $\kappa(n)$ defined in Assumption 5.
Again by (\ref{iff_3}),
\begin{equation*}
\begin{aligned}
     { \underset{(j,k)\in( \mathcal{V}^{c*}, \pi(\mathcal{I}_c^c,\mathcal{V}^{c*}))}{\operatorname{min}}[{p_\lambda^{\text{pen}}}^\prime({\bar{\xi}_{(j,k)}})]}<&\frac{\|\tilde{\boldsymbol{\alpha}}^c_{\mathcal{I}_c^c/ \pi(\mathcal{I}_c^c,\mathcal{V}^{c*})}\|_1 }{\Big\{ \sum_{(j,k)\in( \mathcal{V}^{c*}, \pi(\mathcal{I}_c^c,\mathcal{V}^{c*}))} |\alpha^*_j|-|\tilde{\alpha}^c_k|\Big\} }\cdot { \underset{l \in \mathcal{I}_c^c/ \pi(\mathcal{I}_c^c,\mathcal{V}^{c*})}{\operatorname{max}}[{p_\lambda^{\text{pen}}}^\prime({\bar{\xi}_l})] }\\
     <&\frac{\lambda\|\tilde{\boldsymbol{\alpha}}^c_{\mathcal{I}_c^c/ \pi(\mathcal{I}_c^c,\mathcal{V}^{c*})}\|_1 }{\Big\{ \sum_{(j,k)\in( \mathcal{V}^{c*}, \pi(\mathcal{I}_c^c,\mathcal{V}^{c*}))} |\alpha^*_j|-|\tilde{\alpha}^c_k|\Big\} }\\
     \asymp&\frac{\lambda \kappa(n) (|\mathcal{I}_c^c|-|\mathcal{V}^{c*}|)}{\Big\{ \sum_{(j,k)\in( \mathcal{V}^{c*}, \pi(\mathcal{I}_c^c,\mathcal{V}^{c*}))} |\alpha^*_j|-|\tilde{\alpha}^c_k|\Big\} }
     \asymp C\lambda\kappa(n).
\end{aligned}
 \end{equation*}
Thus, it concludes that ${p_\lambda^{\text{pen}}}^{\prime}(t) = O(\lambda \kappa(n))$ for any $t>\kappa(n)$.\qed
\end{proof}
\subsection{Proof of Lemma \ref{Lemma 1}}
\begin{proof}
$\tilde{\boldsymbol{Z}}_j = M_{{\widehat{\boldsymbol{D}}}}\boldsymbol{Z}_j$, $\boldsymbol{Z}_j$ is the $j$-th instrument. The target is to analyze that the rate of $\tilde{\boldsymbol{Z}}_j^\top \boldsymbol{\epsilon}/n$ is not inflated. To proceed, $|\tilde{\boldsymbol{Z}}_j^\top \boldsymbol{\epsilon}/n| \leq |\boldsymbol{Z}_j^\top \boldsymbol{\epsilon} /n |+|\boldsymbol{Z}_j^\top P_{\widehat{\boldsymbol{D}}}\boldsymbol{\epsilon}/n|$. The first term is known to be $O_p(n^{-1/2})$. For the second term, we have,
\begin{equation}
        |\boldsymbol{Z}_j^\top P_{\widehat{\boldsymbol{D}}}\boldsymbol{\epsilon}/n|
         = \underbrace{\frac{|\boldsymbol{Z}_j^\top\widehat{\boldsymbol{D}}|}{n}}_{(I)}\cdot  \underbrace{\frac{|\boldsymbol{\epsilon}^\top\widehat{\boldsymbol{D}}|}{n}}_{(II)}\cdot \left(\underbrace{\frac{|\widehat{\boldsymbol{D}}^\top \widehat{\boldsymbol{D}}|}{n}}_{(III)}\right)^{-1}.
\end{equation}
Notably, $\widehat{\boldsymbol{D}} = \boldsymbol{Z}\boldsymbol{\gamma}^* + P_{\boldsymbol{Z}}\boldsymbol{\eta}$ is a projected endogenous variable. It consists of the random and non-random parts. Then we analyze the size of the above three terms sequentially.
For $(I)$, we have
\begin{equation}
    \frac{\boldsymbol{Z}_j^\top\widehat{\boldsymbol{D}}}{n} = \frac{\boldsymbol{Z}_j^{\top} (\boldsymbol{Z}\boldsymbol{\gamma}^*+\boldsymbol{\eta})}{n} =(\boldsymbol{Q}_{nj})^\top \boldsymbol{\gamma}^* +\frac{\boldsymbol{Z}_j^\top\boldsymbol{\epsilon}}{n}=\boldsymbol{Q}_{nj}^\top \boldsymbol{\gamma}^* +O_p(n^{-1/2}).
\end{equation}

With regard to $(II)$,
\begin{equation}
    \frac{\boldsymbol{\epsilon}^\top \widehat{\boldsymbol{D}}}{n} = \frac{\boldsymbol{\epsilon}^\top P_{\boldsymbol{Z}} (\boldsymbol{Z}\boldsymbol{\gamma}^*+\boldsymbol{\eta})}{n} = \frac{\boldsymbol{\epsilon}^\top  \boldsymbol{Z}\boldsymbol{\gamma}^*}{n}+\frac{\boldsymbol{\epsilon} P_{\boldsymbol{Z}}\boldsymbol{\eta}}{n}.
\end{equation}
Within this decomposition, we first have $E({\boldsymbol{\epsilon}^\top  \boldsymbol{Z}\boldsymbol{\gamma}^*}/{n}) = E(E(\boldsymbol{\epsilon}^\top |\boldsymbol{Z})\boldsymbol{Z}\boldsymbol{\gamma}^* ) = 0$ and $$\operatorname{var}({\boldsymbol{\epsilon}^\top  \boldsymbol{Z}\boldsymbol{\gamma}^*}/{n}) = E(\operatorname{var}({\boldsymbol{\epsilon}^\top  \boldsymbol{Z}\boldsymbol{\gamma}^*}/{n}|\boldsymbol{Z})) = E(\sigma_{\epsilon}^2\boldsymbol{\gamma}^{*\top}\boldsymbol{Z}^\top \boldsymbol{Z}\boldsymbol{\gamma}^*/n^2) = O(\boldsymbol{\gamma}^{*\top}\boldsymbol{Q}_n\boldsymbol{\gamma}^*/n).$$
Thus we obtain ${\boldsymbol{\epsilon}^\top  \boldsymbol{Z}\boldsymbol{\gamma}^*}/{n} = O_P(\sqrt{\boldsymbol{\gamma}^{*\top}\boldsymbol{Q}_n\boldsymbol{\gamma}^*/n})$.
Regarding the second term, we have $E(\boldsymbol{\epsilon}^\top P_{\boldsymbol{Z}} \boldsymbol{\eta}/n) = E(\operatorname{tr}(P_{\boldsymbol{Z}} \boldsymbol{\eta}\boldsymbol{\epsilon}^\top)/n) = \operatorname{tr}(E(P_{\boldsymbol{Z}} E(\boldsymbol{\eta}\boldsymbol{\epsilon}^\top|\boldsymbol{Z}))/n) = \sigma_{\epsilon,\eta}^2 p/n$ and
\begin{equation*}
\begin{aligned}
        \operatorname{var}(\boldsymbol{\epsilon}^\top P_{\boldsymbol{Z}} \boldsymbol{\eta}/n) = E(\operatorname{var}(\boldsymbol{\epsilon}^\top P_{\boldsymbol{Z}} \boldsymbol{\eta}/n|\boldsymbol{Z})) &= E(2\sigma_{\epsilon}^2 \sigma_{\eta}^2[\operatorname{tr}(P_{\boldsymbol{Z}})/n^2]) +  E([n^{-2}\sum_{i}(P_{\boldsymbol{Z}})^2_{ii}])([\sigma_{\epsilon,\eta}^2]^2-2\sigma_\epsilon^2\sigma_\eta^2)\\
        &\leq O(p/n^2)+O(p/n^2) = O(p/n^2),
\end{aligned}
\end{equation*}
where the inequality holds since $\operatorname{tr}(P_{\boldsymbol{Z}}) = p$, $(P_{\boldsymbol{Z}})_{ii} \in (0,1)$ and $\sum_{i}(P_{\boldsymbol{Z}})^2_{ii}\leq \sum_{i}(P_{\boldsymbol{Z}})_{ii} = p$. Thus, we conclude the size of $(II)$ is $O_P(\sqrt{\boldsymbol{\gamma}^{*\top}\boldsymbol{Q}\boldsymbol{\gamma}^*/n}\lor \sqrt{p/n^2})+ \sigma_{\epsilon,\eta}^2 p/n $.

Now we turn to $(III)$, with similar argument,
\begin{equation*}
\begin{aligned}
        \frac{\widehat{\boldsymbol{D}}^\top \widehat{\boldsymbol{D}}}{n} = \frac{\boldsymbol{D}^\top P_{\boldsymbol{Z}}\boldsymbol{D}}{n} &= \frac{\boldsymbol{\gamma}^{*\top}\boldsymbol{Z}^\top\boldsymbol{Z}\boldsymbol{\gamma}^*}{n}+2\frac{\boldsymbol{\eta}^\top \boldsymbol{Z}\boldsymbol{\gamma}^*}{n}+\frac{\boldsymbol{\eta}^\top P_{\boldsymbol{Z}}\boldsymbol{\eta}}{n}\\
        &= \boldsymbol{\gamma}^{*\top}\boldsymbol{Q}_n\boldsymbol{\gamma}^*+o_p(1)+O_P(\sqrt{\boldsymbol{\gamma}^{*\top}\boldsymbol{Q}_n\boldsymbol{\gamma}^*/n} \lor \sqrt{p/n^2})+\sigma_{\eta}^2 p/n\\
        & = \boldsymbol{\gamma}^{*\top}\boldsymbol{Q}_n\boldsymbol{\gamma}^*+\sigma_{\eta}^2 p/n +O_P(\sqrt{\boldsymbol{\gamma}^{*\top}\boldsymbol{Q}_n\boldsymbol{\gamma}^*/n }\lor \sqrt{p/n^2}).
\end{aligned}
\end{equation*}

Therefore, together with above three terms, we are able to derive the size of $|\boldsymbol{Z}_j^\top P_{\widehat{\boldsymbol{D}}}\boldsymbol{\epsilon}/n| $,
\begin{equation}
\begin{aligned}
        \boldsymbol{Z}_j^\top P_{\widehat{\boldsymbol{D}}}\boldsymbol{\epsilon}/n  &= \frac{\left[\boldsymbol{Q}_{nj}^\top \boldsymbol{\gamma}^* +O_p(n^{-1/2})\right]\cdot\left[\sigma_{\epsilon,\eta}^2 p/n+ O_P(\sqrt{\boldsymbol{\gamma}^{*\top}\boldsymbol{Q}_n\boldsymbol{\gamma}^*/n}\lor \sqrt{p/n^2})
     \right]}{\boldsymbol{\gamma}^{*\top}\boldsymbol{Q}_n\boldsymbol{\gamma}^*+\sigma_{\eta}^2 p/n +O_P(\sqrt{\boldsymbol{\gamma}^{*\top}\boldsymbol{Q}_n\boldsymbol{\gamma}^*/n }\lor \sqrt{p/n^2})}\\
     &=\sigma_{\epsilon,\eta}^2 p/n\cdot \frac{\boldsymbol{Q}_{nj}^\top \boldsymbol{\gamma}^*}{\boldsymbol{\gamma}^{*\top}\boldsymbol{Q}_n\boldsymbol{\gamma}^*+\sigma_{\eta}^2 p/n} +O_p(n^{-1/2})
\end{aligned}
\end{equation}\qed

\end{proof}

\subsection{Proof of Lemma \ref{Lemma 2}}
\begin{proof}
We have the transformed $\tilde{\boldsymbol{Z}} = M_{\widehat{\boldsymbol{D}}}\boldsymbol{Z}$, where $\widehat{\boldsymbol{D}} = P_{\boldsymbol{Z}} \boldsymbol{D} = \boldsymbol{Z}\widehat{\boldsymbol{\gamma}}$. Thus,
\begin{equation}
\begin{aligned}
         \tilde{\boldsymbol{Z}} &= \boldsymbol{Z}-\boldsymbol{Z}\widehat{\boldsymbol{\gamma}}\left(\widehat{\boldsymbol{\gamma}}^\top\boldsymbol{Z}^\top \boldsymbol{Z}\widehat{\boldsymbol{\gamma}}\right)^{-1}\widehat{\boldsymbol{\gamma}}^\top\boldsymbol{Z}^\top\boldsymbol{Z},\\
        \boldsymbol{C}_n =  \frac{\tilde{\boldsymbol{Z}}^\top \tilde{\boldsymbol{Z}}}{n} &=\left(\frac{\boldsymbol{Z}^\top\boldsymbol{Z}}{n}\right)- \left(\frac{\boldsymbol{Z}^\top\boldsymbol{Z}}{n}\right)\widehat{\boldsymbol{\gamma}}\left(\widehat{\boldsymbol{\gamma}}^\top\left(\frac{\boldsymbol{Z}^\top \boldsymbol{Z}}{n}\right)\widehat{\boldsymbol{\gamma}}\right)^{-1}\widehat{\boldsymbol{\gamma}}^\top\left(\frac{\boldsymbol{Z}^\top\boldsymbol{Z}}{n}\right).
\end{aligned}
\end{equation}
 Denote $\boldsymbol{Q}_n = {\boldsymbol{Z}^\top\boldsymbol{Z}}/{n}$, consider the square of restricted eigenvalue of  $\tilde{\boldsymbol{Z}}$, $K_{\mathscr{C}}^2$, we have
\begin{equation}
\begin{aligned}
     K_{\mathscr{C}}^2(\mathcal{V}^*,\xi) =& \,\, \underset{\boldsymbol{u}}{\operatorname{inf}}\{\|\boldsymbol{u}^\top (n^{-1}\tilde{\boldsymbol{Z}}^\top\tilde{\boldsymbol{Z}})\boldsymbol{u}/\|\boldsymbol{u}\|_2^2; \boldsymbol{u}\in \mathscr{C}(\mathcal{V}^* ; \xi) \}\\
     =& \underset{ \boldsymbol{u}\in \mathscr{C}(\mathcal{V}^* ; \xi) }{\operatorname{inf}} \,\, \{\|\boldsymbol{u}^\top \boldsymbol{C}_n\boldsymbol{u}/\|\boldsymbol{u}\|_2^2\}\\
     =& \underset{ \boldsymbol{u}\in \mathscr{C}(\mathcal{V}^* ; \xi) }{\operatorname{inf}} \,\, \frac{\boldsymbol{u}^\top ( \boldsymbol{Q}_n-\boldsymbol{Q}_n\widehat{\boldsymbol{\gamma}}(\widehat{\boldsymbol{\gamma}}^{\top}\boldsymbol{Q}_n\widehat{\boldsymbol{\gamma}})^{-1}\widehat{\boldsymbol{\gamma}}^{\top}\boldsymbol{Q}_n)\boldsymbol{u}}{\|\boldsymbol{u}\|_2^2}\\
         =& \underset{ \boldsymbol{u}\in \mathscr{C}(\mathcal{V}^* ; \xi) }{\operatorname{inf}} \,\,
         \frac{\boldsymbol{u}^\top\boldsymbol{Q}_n\boldsymbol{u}\widehat{\boldsymbol{\gamma}}^{\top}\boldsymbol{Q}_n\widehat{\boldsymbol{\gamma}}^{\top}-(\widehat{\boldsymbol{\gamma}}^{\top}\boldsymbol{Q}_n\boldsymbol{u})^2}{\|\boldsymbol{u}\|_2^2\widehat{\boldsymbol{\gamma}}^{\top}\boldsymbol{Q}_n\widehat{\boldsymbol{\gamma}}}
     \label{RE_Appendix 1}
\end{aligned}
\end{equation}

Notice the denominator $\boldsymbol{u}^\top\boldsymbol{Q}_n\boldsymbol{u}\widehat{\boldsymbol{\gamma}}^{\top}\boldsymbol{Q}_n\widehat{\boldsymbol{\gamma}}^{\top}-(\widehat{\boldsymbol{\gamma}}^{\top}\boldsymbol{Q}_n\boldsymbol{u})^2\geq 0$ by Cauchy–Schwarz inequality and the equality holds if and only if $\boldsymbol{u} = k \widehat{\boldsymbol{\gamma}}$ for any $k \neq 0$. Furthermore, the exact difference arises from this equation can be determined by the Lagrange's identity. One can always take $\xi \in (0, \|\widehat{\boldsymbol{\gamma}}_{\mathcal{V}^{*}}\|_1/\|\widehat{\boldsymbol{\gamma}}_{\mathcal{V}^{c*}}\|_1)$ such that the cone $\mathscr{C}(\mathcal{V}^* ; \xi)=\left\{\boldsymbol{u}:\left\|\boldsymbol{u}_{\mathcal{V}^{*}}\right\|_{1} \leq \xi\left\|\boldsymbol{u}_{\mathcal{V}^{c*}}\right\|_{1}\right\}$ excludes the membership of $k\widehat{\boldsymbol{\gamma}}$ because the cone $\mathscr{C}(\mathcal{V}^* ; \xi)$ is invariant of scale. Therefore, the denominator  $\boldsymbol{u}^\top\boldsymbol{Q}_n\boldsymbol{u}\widehat{\boldsymbol{\gamma}}^{\top}\boldsymbol{Q}_n\widehat{\boldsymbol{\gamma}}^{\top}-(\widehat{\boldsymbol{\gamma}}^{\top}\boldsymbol{Q}_n\boldsymbol{u})^2> 0$ and $  K_{\mathscr{C}}^2(\mathcal{V}^*,\xi) > 0$ holds strictly.\qed

\end{proof}

\subsection{Proof of Lemma \ref{Lemma 3}}
\begin{proof}
Let $\boldsymbol{R^*} = \widetilde{\boldsymbol{Z}}^{\top}(\boldsymbol{Y}-\widetilde{\boldsymbol{Z}}{\boldsymbol{\alpha}}^*)/n$ and $\tilde{\boldsymbol{Z}} = M_{\widehat{\boldsymbol{D}}}\boldsymbol{Z}$,$\boldsymbol{D} = \boldsymbol{Z}\boldsymbol{\gamma}^*+\boldsymbol{\eta}$, we have
\begin{equation}
    \begin{aligned}
        \boldsymbol{R}^* &= \widetilde{\boldsymbol{Z}}^{\top}(\boldsymbol{Y}-\widetilde{\boldsymbol{Z}}{\boldsymbol{\alpha}}^*)/n
         = \boldsymbol{Z}^\top M_{\widehat{\boldsymbol{D}}} (\boldsymbol{D}\beta^*+\boldsymbol{\epsilon})/n = \boldsymbol{Z}^\top \Bigg(\boldsymbol{I}-\frac{\widehat{\boldsymbol{D}}\widehat{\boldsymbol{D}}^\top}{\widehat{\boldsymbol{D}}^\top\widehat{\boldsymbol{D}}}\Bigg)(\boldsymbol{D}\beta^*+\boldsymbol{\epsilon})/n\\
         &=\beta^*\frac{\boldsymbol{Z}^\top\boldsymbol{D}}{n}+\frac{\boldsymbol{\boldsymbol{Z}^\top\boldsymbol{\epsilon}}}{n}-\frac{\frac{\boldsymbol{Z}^\top\widehat{\boldsymbol{D}}}{n}\cdot\frac{\widehat{\boldsymbol{D}}^\top(\boldsymbol{D}\beta^*+\boldsymbol{\epsilon})}{n}}{\frac{\widehat{\boldsymbol{D}}^\top\widehat{\boldsymbol{D}}}{n}}\\
         &=\beta^*\frac{\boldsymbol{Z}^\top\boldsymbol{D}}{n}+\frac{\boldsymbol{\boldsymbol{Z}^\top\boldsymbol{\epsilon}}}{n}-\beta^*\frac{\boldsymbol{Z}^\top\widehat{\boldsymbol{D}}}{n}-\frac{\frac{\boldsymbol{Z}^\top\widehat{\boldsymbol{D}}}{n}\cdot\frac{\widehat{\boldsymbol{D}}^\top\boldsymbol{\epsilon}}{n}}{\frac{\widehat{\boldsymbol{D}}^\top\widehat{\boldsymbol{D}}}{n}}\\
         &=\frac{\boldsymbol{\boldsymbol{Z}^\top\boldsymbol{\epsilon}}}{n}-\frac{\frac{\boldsymbol{Z}^\top\widehat{\boldsymbol{D}}}{n}\cdot\frac{\widehat{\boldsymbol{D}}^\top\boldsymbol{\epsilon}}{n}}{\frac{\widehat{\boldsymbol{D}}^\top\widehat{\boldsymbol{D}}}{n}}\label{R^*}.
    \end{aligned}
\end{equation}
By direct algebra, using $\widehat{\boldsymbol{D}} = \boldsymbol{Z}\boldsymbol{\gamma}^*+P_{\boldsymbol{Z}}\boldsymbol{\eta}$ and  $\widehat{\boldsymbol{D}}^\top\widehat{\boldsymbol{D}}/n = \boldsymbol{\gamma}^{*\top}\boldsymbol{Q}_n\boldsymbol{\gamma}^*+2\boldsymbol{\eta}^\top\boldsymbol{Z}\boldsymbol{\gamma}^*/n+\boldsymbol{\eta}^\top P_{\boldsymbol{Z}}\boldsymbol{\eta}/n$,
we obtain
\begin{equation*}
\begin{aligned}
     \boldsymbol{R}^* = \frac{\boldsymbol{\boldsymbol{Z}^\top\boldsymbol{\epsilon}}}{n}-\frac{\frac{\boldsymbol{Z}^\top\widehat{\boldsymbol{D}}}{n}\cdot\frac{\widehat{\boldsymbol{D}}^\top\boldsymbol{\epsilon}}{n}}{\frac{\widehat{\boldsymbol{D}}^\top\widehat{\boldsymbol{D}}}{n}}=\frac{\boldsymbol{Z}^\top \boldsymbol{\epsilon}}{n}-\frac{\Big(\boldsymbol{Q}_n\boldsymbol{\gamma}^*+\boldsymbol{Z}^\top\boldsymbol{\eta}/n\Big)\Big(\boldsymbol{\epsilon}^\top \boldsymbol{Z}\boldsymbol{\gamma}^*/n+\boldsymbol{\epsilon}^\top P_{\boldsymbol{Z}}\boldsymbol{\eta}/n\Big)}{\boldsymbol{\gamma}^{*\top}\boldsymbol{Q}_n\boldsymbol{\gamma}^*+2\boldsymbol{\eta}^\top\boldsymbol{Z}\boldsymbol{\gamma}^*/n+\boldsymbol{\eta}^\top P_{\boldsymbol{Z}}\boldsymbol{\eta}/n},
\end{aligned}
\end{equation*}
in which the expression is free of $\beta^*$. Thus, by triangle inequality, we attain the bound
     \begin{equation}
         \|\boldsymbol{R}\|_\infty \leq \left\|\frac{\boldsymbol{Z}^\top \boldsymbol{\epsilon}}{n}\right\|_\infty+ \Bigg\|\frac{\Big(\boldsymbol{Q}_n\boldsymbol{\gamma}^*+\boldsymbol{Z}^\top\boldsymbol{\eta}/n\Big)\Big(\boldsymbol{\epsilon}^\top \boldsymbol{Z}\boldsymbol{\gamma}^*/n+\boldsymbol{\epsilon}^\top P_{\boldsymbol{Z}}\boldsymbol{\eta}/n\Big)}{\boldsymbol{\gamma}^{*\top}\boldsymbol{Q}_n\boldsymbol{\gamma}^*+2\boldsymbol{\eta}^\top\boldsymbol{Z}\boldsymbol{\gamma}^*/n+\boldsymbol{\eta}^\top P_{\boldsymbol{Z}}\boldsymbol{\eta}/n}\Bigg\|_\infty\label{Inf_norm_residual}.
     \end{equation}
Under standard argument through concentration inequality,
\begin{equation*}
    Pr\Bigg(\left\|\frac{\boldsymbol{Z}^\top \boldsymbol{\epsilon}}{n}\right\|_\infty\geq t\Bigg) = Pr\Big(\underset{ 1\leq j \leq p}{\operatorname{max}}|\boldsymbol{Z}_j^\top\boldsymbol{\epsilon}|\geq nt\Big)\leq \sum_{1\leq j \leq p}Pr(|\boldsymbol{Z}_j^\top \boldsymbol{\epsilon}|\geq nt)\leq 2p\exp(-\frac{nt^2}{2\sigma^2_\epsilon}).
\end{equation*}
Let $t = \sigma_\epsilon\sqrt{\frac{2}{n}\log (2p)}$, we obtain $\left\|{\boldsymbol{Z}^\top \boldsymbol{\epsilon}}/{n}\right\|_\infty  = O_p\Big(\sigma_\epsilon\sqrt{\frac{2}{n}\log (2p)}\Big) $.

For the second term,
\begin{equation*}
    \Bigg\|\frac{\Big(\boldsymbol{Q}_n\boldsymbol{\gamma}^*+\boldsymbol{Z}^\top\boldsymbol{\eta}/n\Big)\Big(\boldsymbol{\epsilon}^\top \boldsymbol{Z}\boldsymbol{\gamma}^*/n+\boldsymbol{\epsilon}^\top P_{\boldsymbol{Z}}\boldsymbol{\eta}/n\Big)}{\boldsymbol{\gamma}^{*\top}\boldsymbol{Q}_n\boldsymbol{\gamma}^*+2\boldsymbol{\eta}^\top\boldsymbol{Z}\boldsymbol{\gamma}^*/n+\boldsymbol{\eta}^\top P_{\boldsymbol{Z}}\boldsymbol{\eta}/n}\Bigg\|_\infty  = \frac{\Big|\Big(\boldsymbol{\epsilon}^\top \boldsymbol{Z}\boldsymbol{\gamma}^*/n+\boldsymbol{\epsilon}^\top P_{\boldsymbol{Z}}\boldsymbol{\eta}/n\Big)\Big|\Big\|\boldsymbol{Q}_n\boldsymbol{\gamma}^*+\boldsymbol{Z}^\top\boldsymbol{\eta}/n\Big\|_\infty}{\Big|\boldsymbol{\gamma}^{*\top}\boldsymbol{Q}_n\boldsymbol{\gamma}^*+2\boldsymbol{\eta}^\top\boldsymbol{Z}\boldsymbol{\gamma}^*/n+\boldsymbol{\eta}^\top P_{\boldsymbol{Z}}\boldsymbol{\eta}/n\Big|}
\end{equation*}
Similarly, we attain $\big\|\boldsymbol{Q}_n\boldsymbol{\gamma}^*+\boldsymbol{Z}^\top\boldsymbol{\eta}/n\big\|_\infty\leq \|\boldsymbol{Q}_n\boldsymbol{\gamma}^*\|_{\infty} + O_p\Big(\sigma_\eta\sqrt{\frac{2}{n}\log (2p)}\Big)$.

Also using the ancillary lemmas \ref{Lemma A1} and similar analyses in proof of Lemma \ref{Lemma 1}, we know the second term in (\ref{Inf_norm_residual}) is upper-bounded by:
\begin{equation}
\begin{aligned}
     &\frac{\Big(\|\boldsymbol{Q}_n\boldsymbol{\gamma}^*\|_{\infty} + O_p\Big(\sigma_\eta\sqrt{\frac{2}{n}\log (2p)}\Big)\Big)\Big( \sigma_{\epsilon,\eta} p/n +O_P(\sqrt{\boldsymbol{\gamma}^{*\top}\boldsymbol{Q}_n\boldsymbol{\gamma}^*/n}\lor \sqrt{p/n^2}) \Big)}{\boldsymbol{\gamma}^{*\top}\boldsymbol{Q}_n\boldsymbol{\gamma}^*+\sigma_{\eta}^2 p/n +O_P(\sqrt{\boldsymbol{\gamma}^{*\top}\boldsymbol{Q}_n\boldsymbol{\gamma}^*/n }\lor \sqrt{p/n^2}) }\\
    =&\sigma_{\epsilon,\eta} p/n\frac{\|\boldsymbol{Q}_n\boldsymbol{\gamma}^*\|_{\infty}}{\boldsymbol{\gamma}^{*\top}\boldsymbol{Q}_n\boldsymbol{\gamma}^*+\sigma_{\eta}^2 p/n}+O_p\Big(\sigma_\eta\sqrt{\frac{2}{n}\log (2p)}\Big).
\end{aligned}
\end{equation}
Hence, together with above results, we obtain the final upper bound of $\|\boldsymbol{R}^*\|_{\infty}$:
\begin{equation}
\begin{aligned}
    \|\boldsymbol{R}^*\|_{\infty} &\leq \sigma_{\epsilon,\eta} p/n\frac{\|\boldsymbol{Q}_n\boldsymbol{\gamma}^*\|_{\infty}}{\boldsymbol{\gamma}^{*\top}\boldsymbol{Q}_n\boldsymbol{\gamma}^*+\sigma_{\eta}^2 p/n}+O_p\Big((\sigma_\eta+\sigma_\epsilon)\sqrt{\frac{2}{n}\log (2p)}\Big)\\
    &= O_p\Big(\frac{p}{n}\cdot\frac{\|\boldsymbol{Q}_n\boldsymbol{\gamma}^*\|_{\infty}}{\boldsymbol{\gamma}^{*\top}\boldsymbol{Q}_n\boldsymbol{\gamma}^*}+\sqrt{\frac{\log p}{n}}\Big),\label{B10}
\end{aligned}
\end{equation}
where the second line holds provided $\boldsymbol{\gamma}^{*\top}\boldsymbol{Q}_n\boldsymbol{\gamma}^*$ dominates $\sigma_{\epsilon,\eta}p/n$.

Similarly, $\boldsymbol{R}^{\text{or}} = \widetilde{\boldsymbol{Z}}^{\top}(\boldsymbol{Y}-\widetilde{\boldsymbol{Z}}{\widehat{\boldsymbol{\alpha}}}^{\text{or}})/n =\tilde{\boldsymbol{Z}}^\top \big[\boldsymbol{Y}-\tilde{\boldsymbol{Z}}_{\mathcal{V}^{c*}}(\tilde{\boldsymbol{Z}}^\top_{\mathcal{V}^{c*}} \tilde{\boldsymbol{Z}}_{\mathcal{V}^{c*}})^{-1}\tilde{\boldsymbol{Z}}^\top_{\mathcal{V}^{c*}} \boldsymbol{Y}\big]/n$. Therefore, for the sake of controlling the supremum norm of $\boldsymbol{R}^{\text{or}}$, we only need to consider the valid IV $\mathcal{V}^{*}$ since $\boldsymbol{R}^{\text{or}}_{j} = 0$ for $j \in \mathcal{V}^{c*}$. Recall $\hat{\beta}_{\text{or}}^{\text{TSLS}} = [\boldsymbol{D}^\top(P_{\boldsymbol{Z}}-P_{\boldsymbol{Z}_{\mathcal{V}^{c*}}})\boldsymbol{D}]^{-1}[\boldsymbol{D}^\top(P_{\boldsymbol{Z}}-P_{\boldsymbol{Z}_{\mathcal{V}^{c*}}})\boldsymbol{Y}]$, we have
\begin{equation}
\begin{aligned}
        \widehat{\boldsymbol{\alpha}}_{\mathcal{V}^{c*}}^{\text{or}} &= (\boldsymbol{Z}_{\mathcal{V}^{c*}}^\top \boldsymbol{Z}_{\mathcal{V}^{c*}})^{-1}[\boldsymbol{Z}_{\mathcal{V}^{c*}}^\top(\boldsymbol{Y}-\widehat{\boldsymbol{D}}\hat{\beta}_{\text{or}}^{\text{TSLS}})]\\
        & = (\boldsymbol{Z}_{\mathcal{V}^{c*}}^\top \boldsymbol{Z}_{\mathcal{V}^{c*}})^{-1}\boldsymbol{Z}_{\mathcal{V}^{c*}}^\top\boldsymbol{Y} - (\boldsymbol{Z}_{\mathcal{V}^{c*}}^\top \boldsymbol{Z}_{\mathcal{V}^{c*}})^{-1}\boldsymbol{Z}_{\mathcal{V}^{c*}}^\top \boldsymbol{D} [\boldsymbol{D}^\top(P_{\boldsymbol{Z}}-P_{\boldsymbol{Z}_{\mathcal{V}^{c*}}})\boldsymbol{D}]^{-1}[\boldsymbol{D}^\top(P_{\boldsymbol{Z}}-P_{\boldsymbol{Z}_{\mathcal{V}^{c*}}})\boldsymbol{Y}]\\
        &= \boldsymbol{\alpha}^*_{\mathcal{V}^{c*}}+ (\boldsymbol{Z}_{\mathcal{V}^{c*}}^\top \boldsymbol{Z}_{\mathcal{V}^{c*}})^{-1}\boldsymbol{Z}_{\mathcal{V}^{c*}}^\top\Big[ \boldsymbol{I}- \boldsymbol{D} [\boldsymbol{D}^\top(P_{\boldsymbol{Z}}-P_{\boldsymbol{Z}_{\mathcal{V}^{c*}}})\boldsymbol{D}]^{-1}\boldsymbol{D}^\top(P_{\boldsymbol{Z}}-P_{\boldsymbol{Z}_{\mathcal{V}^{c*}}})\Big]\boldsymbol{\epsilon}.\label{TSLS_bias}
\end{aligned}
\end{equation}
Thus, for any $j\in \mathcal{V}^{*}$, we obtain
\begin{equation}
    \begin{aligned}
          \boldsymbol{R}^{\text{or}}_{\mathcal{V}^{*}} &= \tilde{\boldsymbol{Z}}_{\mathcal{V}^*}^\top \Big\{\boldsymbol{Y}-\tilde{\boldsymbol{Z}}(\widehat{\boldsymbol{\alpha}}^{\text{or}}-\boldsymbol{\alpha}^*+\boldsymbol{\alpha}^*)\Big\}/n \\
         &= \tilde{\boldsymbol{Z}}_{\mathcal{V}^*}^\top (\boldsymbol{Y}-\tilde{\boldsymbol{Z}}{\boldsymbol{\alpha}}^{*})/n+\tilde{\boldsymbol{Z}}_{\mathcal{V}^*}^\top\tilde{\boldsymbol{Z}}_{\mathcal{V}^{c*}}(\boldsymbol{\alpha}_{\mathcal{V}^{c*}}^*-\widehat{\boldsymbol{\alpha}}_{\mathcal{V}^{c*}}^{\text{or}})/n\\
         &= \boldsymbol{R}^{*}_{\mathcal{V}^*} + \boldsymbol{Z}_{\mathcal{V}^*}^\top M_{\widehat{\boldsymbol{D}}}P_{\boldsymbol{Z}_{\mathcal{V}^{c*}}}\Big(\boldsymbol{D}\cdot\operatorname{Bias}(\hat{\beta}_{\text{or}}^{\text{TSLS}})-\boldsymbol{\epsilon}\Big)/n, \label{R_or}
    \end{aligned}
\end{equation}
where $\operatorname{Bias}(\hat{\beta}_{\text{or}}^{\text{TSLS}}) =  \frac{\boldsymbol{D}^\top(P_{\boldsymbol{Z}}-P_{\boldsymbol{Z}_{\mathcal{V}^{c*}}})\boldsymbol{\epsilon}}{\boldsymbol{D}^\top(P_{\boldsymbol{Z}}-P_{\boldsymbol{Z}_{\mathcal{V}^{c*}}})\boldsymbol{D}} $.

To explore the second term, we denote $\bar{\boldsymbol{D}} = P_{\boldsymbol{Z}_{\mathcal{V}^{c*}}}\boldsymbol{D}$, $\bar{\boldsymbol{\epsilon}} = P_{\boldsymbol{Z}_{\mathcal{V}^{c*}}}\boldsymbol{\epsilon}$.  Due to the blockwise formula for the projection matrix, we have
\begin{equation*}
    P_{\boldsymbol{Z}}-P_{\boldsymbol{Z}_{\mathcal{V}^{c*}}} = P_{M_{\boldsymbol{Z}_{\mathcal{V}^{c*}}}\boldsymbol{Z}_{\mathcal{V}^{*}}},
\end{equation*}
which is also a projection matrix.
Thus, we denote $\TTD = P_{M_{\boldsymbol{Z}_{\mathcal{V}^{c*}}}\boldsymbol{Z}_{\mathcal{V}^{*}}}\boldsymbol{D}$, $\TTE = P_{M_{\boldsymbol{Z}_{\mathcal{V}^{c*}}}\boldsymbol{Z}_{\mathcal{V}^{*}}}\boldsymbol{\epsilon}$ and $\operatorname{Bias}(\hat{\beta}_{\text{or}}^{\text{TSLS}}) =  \frac{\boldsymbol{D}^\top(P_{\boldsymbol{Z}}-P_{\boldsymbol{Z}_{\mathcal{V}^{c*}}})\boldsymbol{\epsilon}}{\boldsymbol{D}^\top(P_{\boldsymbol{Z}}-P_{\boldsymbol{Z}_{\mathcal{V}^{c*}}})\boldsymbol{D}}  = \frac{\TTD^\top \TTE}{\TTD^\top\TTD}$. Thus we have
\begin{equation}
    \widehat{\boldsymbol{D}}-\bar{\boldsymbol{D}} = (P_{\boldsymbol{Z}}-P_{\boldsymbol{Z}_{\mathcal{V}^{c*}}})\boldsymbol{D} =  \TTD.\label{D_iden}
\end{equation}
Therefore, the second term inside the RHS of \eqref{R_or} can be reformulated as:
\begin{equation}
\begin{aligned}
      &\boldsymbol{Z}_{\mathcal{V}^*}^\top M_{\widehat{\boldsymbol{D}}}P_{\boldsymbol{Z}_{\mathcal{V}^{c*}}}\Big(\boldsymbol{D}\cdot\operatorname{Bias}(\hat{\beta}_{\text{or}}^{\text{TSLS}})-\boldsymbol{\epsilon}\Big)/n \\
      =&\boldsymbol{Z}_{\mathcal{V}^*}^\top \Bigg(\boldsymbol{I}-\frac{\widehat{\boldsymbol{D}}\widehat{\boldsymbol{D}}^\top}{\widehat{\boldsymbol{D}}^\top\widehat{\boldsymbol{D}}}\Bigg)(\bar{\boldsymbol{D}}\operatorname{Bias}(\hat{\beta}_{\text{or}}^{\text{TSLS}})-\bar{\boldsymbol{\epsilon}})/n\\
      =&\frac{(I)+(II)+(III)+(IV)}{\widehat{\boldsymbol{D}}^\top\widehat{\boldsymbol{D}}/n}, \label{zong}
\end{aligned}
\end{equation}
where
\begin{equation*}
 \begin{aligned}
     (I) =& \operatorname{Bias}(\hat{\beta}_{\text{or}}^{\text{TSLS}})\cdot\frac{\boldsymbol{Z}_{\mathcal{V}^*}^\top\widehat{\boldsymbol{D}}^\top}{n}\cdot\frac{\widehat{\boldsymbol{D}}^\top \bar{\boldsymbol{D}}}{n} = \frac{\TTD^\top \TTE}{\TTD^\top\TTD}\cdot \frac{\widehat{\boldsymbol{D}}^\top\widehat{\boldsymbol{D}}}{n}\cdot\frac{\boldsymbol{Z}_{\mathcal{V}^*}^\top\bar{\boldsymbol{D}}}{n},\\
     (II) =& -\frac{\boldsymbol{Z}_{\mathcal{V}^*}^\top\widehat{\boldsymbol{D}}^\top}{n}\cdot \frac{\widehat{\boldsymbol{D}}\bar{\boldsymbol{\epsilon}}}{n} = -\frac{\widehat{\boldsymbol{D}}^\top\widehat{\boldsymbol{D}}}{n}\cdot\frac{\boldsymbol{Z}_{\mathcal{V}^*}^\top\bar{\boldsymbol{\epsilon}}}{n},\\
     (III) =&-\operatorname{Bias}(\hat{\beta}_{\text{or}}^{\text{TSLS}})\cdot\frac{\boldsymbol{Z}_{\mathcal{V}^*}^\top \widehat{\boldsymbol{D}}}{n}\cdot\frac{\widehat{\boldsymbol{D}}^\top \bar{\boldsymbol{D}}}{n} =-\frac{\TTD^\top \TTE}{\TTD^\top\TTD}\cdot \frac{\boldsymbol{Z}_{\mathcal{V}^*}^\top \widehat{\boldsymbol{D}}}{n}\cdot\frac{\widehat{\boldsymbol{D}}^\top \bar{\boldsymbol{D}}}{n},\\
     (IV) =& \frac{\boldsymbol{Z}_{\mathcal{V}^*}^\top \widehat{\boldsymbol{D}}}{n}\cdot\frac{\widehat{\boldsymbol{D}}^\top \bar{\boldsymbol{\epsilon}}}{n} .
\end{aligned}
\end{equation*}

Now, using identity \eqref{D_iden}, we replace $\frac{\widehat{\boldsymbol{D}}^\top \bar{\boldsymbol{D}}}{n}$ in (III) with $\frac{\widehat{\boldsymbol{D}}^\top (\widehat{\boldsymbol{D}}-\TTD)}{n}$. Therefore we obtain
\begin{equation}
\begin{aligned}
    (I)+(III) =& \frac{\TTD^\top \TTE}{\TTD^\top\TTD}\cdot \frac{\widehat{\boldsymbol{D}}^\top\widehat{\boldsymbol{D}}}{n}\cdot\frac{\boldsymbol{Z}_{\mathcal{V}^*}^\top(\bar{\boldsymbol{D}}-\widehat{\boldsymbol{D}})}{n}+ \frac{\TTD^\top \TTE}{\TTD^\top\TTD}\cdot \frac{\widehat{\boldsymbol{D}}^\top\TTD}{n}\cdot\frac{\boldsymbol{Z}_{\mathcal{V}^*}^\top\widehat{\boldsymbol{D}}}{n}\\
    =&\frac{\TTD^\top \TTE}{\TTD^\top\TTD}\cdot \frac{\widehat{\boldsymbol{D}}^\top\widehat{\boldsymbol{D}}}{n}\cdot\frac{-\boldsymbol{Z}_{\mathcal{V}^*}^\top\TTD}{n}+ \frac{\TTD^\top \TTE}{\TTD^\top\TTD}\cdot \frac{\widehat{\boldsymbol{D}}^\top\TTD}{n}\cdot\frac{\boldsymbol{Z}_{\mathcal{V}^*}^\top\widehat{\boldsymbol{D}}}{n}.
    \end{aligned}
\end{equation}
Hence, \eqref{zong} becomes
\begin{equation}
    \begin{aligned}
        & \frac{(I)+(II)+(III)+(IV)}{\widehat{\boldsymbol{D}}^\top\widehat{\boldsymbol{D}}/n}\\
         =&\frac{-\frac{\widehat{\boldsymbol{D}}^\top\widehat{\boldsymbol{D}}}{n} \Big(\frac{\boldsymbol{Z}_{\mathcal{V}^*}^\top\TTD}{n}\cdot\frac{\TTD^\top \TTE}{\TTD^\top\TTD} +\frac{\boldsymbol{Z}_{\mathcal{V}^*}^\top\bar{\boldsymbol{\epsilon}}}{n}\Big) + \frac{\boldsymbol{Z}_{\mathcal{V}^*}^\top \widehat{\boldsymbol{D}}}{n}\Big(\frac{\widehat{\boldsymbol{D}}^\top \bar{\boldsymbol{\epsilon}}}{n} +\frac{\widehat{\boldsymbol{D}}^\top\TTD}{n}\cdot
        \frac{\TTD^\top \TTE}{\TTD^\top\TTD} \Big) }{\frac{\widehat{\boldsymbol{D}}^\top\widehat{\boldsymbol{D}}}{n}}\\
        =&-\frac{\boldsymbol{Z}_{\mathcal{V}^*}^\top}{n}(P_{\TTD}\TTE+\bar{\boldsymbol{\epsilon}})+\frac{\frac{\boldsymbol{Z}_{\mathcal{V}^*}^\top \widehat{\boldsymbol{D}}}{n}\cdot \frac{\widehat{\boldsymbol{D}}^\top}{n}(P_{\TTD}\TTE+\bar{\boldsymbol{\epsilon}}) }{\frac{\widehat{\boldsymbol{D}}^\top\widehat{\boldsymbol{D}}}{n}},\label{B26}
    \end{aligned}
\end{equation}
where $P_{\TTD} =\TTD\TTD^\top /\TTD^\top \TTD$ is a projection matrix of $\TTD$. Notice
\begin{equation}
\begin{aligned}
         P_{\TTD}\TTE+\bar{\boldsymbol{\epsilon}} =  \{P_{\TTD}(P_{\boldsymbol{Z}} -P_{\boldsymbol{Z}_{\mathcal{V}^{c*}}})+P_{\boldsymbol{Z}_{\mathcal{V}^{c*}}}\}\boldsymbol{\epsilon} = (P_{\TTD}+P_{\boldsymbol{Z}_{\mathcal{V}^{c*}}})\boldsymbol{\epsilon},\\
         \widehat{\boldsymbol{D}}^\top\TTD = \boldsymbol{D}^\top P_{\boldsymbol{Z}}(P_{\boldsymbol{Z}} -P_{\boldsymbol{Z}_{\mathcal{V}^{c*}}})\boldsymbol{D} = \boldsymbol{D}^\top (P_{\boldsymbol{Z}} -P_{\boldsymbol{Z}_{\mathcal{V}^{c*}}})\boldsymbol{D}=\TTD^\top\TTD.
\end{aligned}
\end{equation}
Thus, we are able to further simplify \eqref{B26} as
\begin{equation}
    \begin{aligned}
         &-\frac{\boldsymbol{Z}_{\mathcal{V}^*}^\top}{n}(P_{\TTD}\TTE+\bar{\boldsymbol{\epsilon}})+\frac{\frac{\boldsymbol{Z}_{\mathcal{V}^*}^\top \widehat{\boldsymbol{D}}}{n}\cdot \frac{\widehat{\boldsymbol{D}}^\top}{n}(P_{\TTD}\TTE+\bar{\boldsymbol{\epsilon}}) }{\frac{\widehat{\boldsymbol{D}}^\top\widehat{\boldsymbol{D}}}{n}}\\
         =&\frac{-\boldsymbol{Z}_{\mathcal{V}^*}^\top(P_{\TTD}+P_{\boldsymbol{Z}_{\mathcal{V}^{c*}}})\boldsymbol{\epsilon}}{n}+\frac{\frac{\boldsymbol{Z}_{\mathcal{V}^*}^\top \widehat{\boldsymbol{D}}}{n}\cdot \frac{\widehat{\boldsymbol{D}}^\top}{n}(P_{\TTD}+P_{\boldsymbol{Z}_{\mathcal{V}^{c*}}})\boldsymbol{\epsilon} }{\frac{\widehat{\boldsymbol{D}}^\top\widehat{\boldsymbol{D}}}{n}}\\
         =& \frac{-\frac{\boldsymbol{Z}_{\mathcal{V}^*}^\top \TTD}{n}\cdot \frac{\TTD^\top\boldsymbol{\epsilon}}{n}}{\frac{\TTD^\top\TTD}{n}}+\frac{-\boldsymbol{Z}_{\mathcal{V}^*}^\top P_{\boldsymbol{Z}_{\mathcal{V}^{c*}}}\boldsymbol{\epsilon}}{n} +\frac{\frac{\boldsymbol{Z}_{\mathcal{V}^*}^\top \widehat{\boldsymbol{D}}}{n}}{\frac{\widehat{\boldsymbol{D}}^\top\widehat{\boldsymbol{D}}}{n}}\cdot \Bigg\{\frac{\frac{\widehat{\boldsymbol{D}}^\top \TTD}{n}\cdot \frac{\TTD^\top\boldsymbol{\epsilon}}{n}}{\frac{\TTD^\top\TTD}{n}}+\frac{\widehat{\boldsymbol{D}}^\top P_{\boldsymbol{Z}_{\mathcal{V}^{c*}}}\boldsymbol{\epsilon}}{n}\Bigg\}\\
         =& \frac{-\frac{\boldsymbol{Z}_{\mathcal{V}^*}^\top \TTD}{n}\cdot \frac{\TTD^\top\boldsymbol{\epsilon}}{n}}{\frac{\TTD^\top\TTD}{n}}+\frac{-\boldsymbol{Z}_{\mathcal{V}^*}^\top P_{\boldsymbol{Z}_{\mathcal{V}^{c*}}}\boldsymbol{\epsilon}}{n}+\frac{\frac{\boldsymbol{Z}_{\mathcal{V}^*}^\top \widehat{\boldsymbol{D}}}{n}}{\frac{\widehat{\boldsymbol{D}}^\top\widehat{\boldsymbol{D}}}{n}}\cdot\frac{(\TTD^\top+\widehat{\boldsymbol{D}}^\top P_{\boldsymbol{Z}_{\mathcal{V}^{c*}}})\boldsymbol{\epsilon}}{n}\\
         =&\frac{-\frac{\boldsymbol{Z}_{\mathcal{V}^*}^\top \TTD}{n}\cdot \frac{\TTD^\top\boldsymbol{\epsilon}}{n}}{\frac{\TTD^\top\TTD}{n}}+\frac{-\boldsymbol{Z}_{\mathcal{V}^*}^\top P_{\boldsymbol{Z}_{\mathcal{V}^{c*}}}\boldsymbol{\epsilon}}{n}+\frac{\frac{\boldsymbol{Z}_{\mathcal{V}^*}^\top \widehat{\boldsymbol{D}}}{n} \cdot\frac{\widehat{\boldsymbol{D}}^\top \boldsymbol{\epsilon}}{n}}{\frac{\widehat{\boldsymbol{D}}^\top\widehat{\boldsymbol{D}}}{n} } \label{Fianl_1234}
    \end{aligned}
\end{equation}

Combining \eqref{Fianl_1234}, \eqref{R^*} and \eqref{R_or}, we obtain
\begin{equation}
    R_{\mathcal{V}^*}^{\text{or}} = \frac{\boldsymbol{Z}_{\mathcal{V}^*}^\top \TTE}{n}-\frac{\frac{\boldsymbol{Z}_{\mathcal{V}^*}^\top \TTD}{n}\cdot \frac{\TTD^\top\boldsymbol{\epsilon}}{n}}{\frac{\TTD^\top\TTD}{n}},
\end{equation}
which has a similar structure to \eqref{R^*}. Consequently, with an analogous argument, we derive
\begin{equation}
    \|\boldsymbol{R}^{\text{or}}\|_\infty = \|\boldsymbol{R}_{\mathcal{V}^*}^{\text{or}}\|_\infty=O_p\Big(\frac{p_{\mathcal{V}^*}}{n}\cdot\frac{\|\Tilde{\Tilde{\boldsymbol{Q}}}_n\boldsymbol{\gamma}_{\mathcal{V}^*}^*\|_{\infty}}{\boldsymbol{\gamma}_{\mathcal{V}^*}^{*\top}\Tilde{\Tilde{\boldsymbol{Q}}}_n\boldsymbol{\gamma}_{\mathcal{V}^*}^*}+\sqrt{\frac{\log p_{\mathcal{V}^*}}{n}}\Big),
\end{equation}
where $\Tilde{\Tilde{\boldsymbol{Q}}}_n = \boldsymbol{Z}_{\mathcal{V}^*}^\top (  P_{\boldsymbol{Z}}-P_{\boldsymbol{Z}_{\mathcal{V}^{c*}}})\boldsymbol{Z}_{\mathcal{V}^*}/n$.\qed

\end{proof}

\subsection{Proof of Theorem \ref{Theorem 3}}\label{proof of theorem 3}
The proof of Theorem \ref{Theorem 3} is similar to the technique of \cite{feng2019sorted}. However, we pay more careful attention to the issue of endogeneity of $\tilde{\boldsymbol{Z}}$ and $\xi$, which is discussed in Lemma \ref{Lemma 2}.

To proceed, we first prove a useful inequality, known as basic inequality in sparse regression literature. Denote $\boldsymbol{\alpha}^0$ as $\boldsymbol{\alpha}^*$ or $\widehat{\boldsymbol{\alpha}}^{\text{or}}$, sharing the same support on $\mathcal{V}^{c*}$, and  $\boldsymbol{R}^0$ is $\boldsymbol{R}^*$ or $\boldsymbol{R}^{\text{or}}$  defined in Lemma \ref{Lemma 2} upon the choice of $\boldsymbol{\alpha}^0$.
\begin{lemma}\label{Lemma A2}
Suppose $\widehat{\boldsymbol{\alpha}}$ is a solution of (\ref{form 2}) and denote $\boldsymbol{\Delta} = \widehat{\boldsymbol{\alpha}}-\boldsymbol{\alpha}^0$ and let
\begin{equation}
\boldsymbol{\omega}(\boldsymbol{\alpha}) = \Big[\frac{\partial}{\partial \boldsymbol{t}}p_\lambda^{\text{MCP}} (\boldsymbol{t})|_{\boldsymbol{t} = \boldsymbol{\alpha}} - \tilde{\boldsymbol{Z}}^\top (\boldsymbol{Y}-\tilde{\boldsymbol{Z}}\boldsymbol{\alpha})/n\Big]/\lambda
\end{equation}
to measure the scaled violation of first order condition in (\ref{form 1}), then
\begin{equation}
    \boldsymbol{\Delta}^\top\boldsymbol{C}_n\boldsymbol{\Delta}\leq -\lambda \boldsymbol{\Delta}^\top \boldsymbol{\omega}(\boldsymbol{\alpha}^0)+1/\rho \|\boldsymbol{\Delta}\|_2^2. \label{basic_ine_1}
\end{equation}
Further, choose a proper sub-derivative at the origin: $\frac{\partial}{\partial {t}}p_\lambda^{\text{MCP}} ({t})|_{{t} =\alpha^0_j} =\lambda sgn(\Delta_j), \forall j \in \mathcal{V}^*$,
\begin{equation}
    \boldsymbol{\Delta}^\top\boldsymbol{C}_n\boldsymbol{\Delta}+   (\lambda-\|\boldsymbol{R}_{\mathcal{V}^*}\|_{\infty} )\|\boldsymbol{\Delta}_{\mathcal{V}^{*}}\|_1  \leq   -\lambda \boldsymbol{\Delta}_{\mathcal{V}^{c*}}^\top \boldsymbol{\omega}_{\mathcal{V}^{c*}}(\boldsymbol{\alpha}^0)+1/\rho \|\boldsymbol{\Delta}\|_2^2.
\end{equation}
\end{lemma}
\begin{proof}
Because $\boldsymbol{\omega}(\widehat{\boldsymbol{\alpha}}) = \boldsymbol{0}$ in (\ref{form 1}), we have
$\tilde{\boldsymbol{Z}}^\top \boldsymbol{Y}/n    =  \frac{\partial}{\partial \boldsymbol{t}}p_\lambda^{\text{MCP}} (\boldsymbol{t})|_{\boldsymbol{t} = \widehat{\boldsymbol{\alpha}}}+\tilde{\boldsymbol{Z}}^\top \tilde{\boldsymbol{Z}}\widehat{\boldsymbol{\alpha}}/n
$. Recall $\boldsymbol{C}_n = \tilde{\boldsymbol{Z}}^\top \tilde{\boldsymbol{Z}}/n$. Thus, we replace $\tilde{\boldsymbol{Z}}^\top \boldsymbol{Y}/n$ in $\boldsymbol{\omega}(\boldsymbol{\alpha}^*)$ and obtain
\begin{equation*}
    \boldsymbol{C}_n\boldsymbol{\Delta} = -\lambda \boldsymbol{\omega}(\boldsymbol{\alpha}^0)+ \frac{\partial}{\partial \boldsymbol{t}}p_\lambda^{\text{MCP}} (\boldsymbol{t})|_{\boldsymbol{t} = \boldsymbol{\alpha}^0} -\frac{\partial}{\partial \boldsymbol{t}}p_\lambda^{\text{MCP}} (\boldsymbol{t})|_{\boldsymbol{t} = \widehat{\boldsymbol{\alpha}}}.
\end{equation*}
Multiply $\boldsymbol{\Delta}^\top$ on both sides, we have
\begin{equation}
    \begin{aligned}
        \boldsymbol{\Delta}^\top\boldsymbol{C}_n\boldsymbol{\Delta} &= -\lambda\boldsymbol{\Delta}^\top \boldsymbol{\omega}(\boldsymbol{\alpha}^0)+\boldsymbol{\Delta}^\top\Big( \frac{\partial}{\partial \boldsymbol{t}}p_\lambda^{\text{MCP}} (\boldsymbol{t})|_{\boldsymbol{t} = \boldsymbol{\alpha}^0} -\frac{\partial}{\partial \boldsymbol{t}}p_\lambda^{\text{MCP}} (\boldsymbol{t})|_{\boldsymbol{t} = \widehat{\boldsymbol{\alpha}}}\Big)\\
        & \leq -\lambda\boldsymbol{\Delta}^\top \boldsymbol{\omega}(\boldsymbol{\alpha}^0)+1/\rho \|\boldsymbol{\Delta}\|_2^2,
    \end{aligned}
\end{equation}
where the second line follows the convexity level of MCP up to $1/\rho$ and concludes (\ref{basic_ine_1}).
Moreover, we further examine the terms in $\boldsymbol{\omega}(\boldsymbol{\alpha}^0)$ with respect to $\mathcal{V}^*$. We obtain,
\begin{equation}
\begin{aligned}
    \boldsymbol{\omega}_{\mathcal{V}^*}(\boldsymbol{\alpha}^0) &= \Big[\frac{\partial}{\partial \boldsymbol{t}}p_\lambda^{\text{MCP}} (\boldsymbol{t})|_{\boldsymbol{t} = \boldsymbol{\alpha}^0_{\mathcal{V}^*}} - \tilde{\boldsymbol{Z}}^\top_{\mathcal{V}^*} (\boldsymbol{Y}-\tilde{\boldsymbol{Z}}\boldsymbol{\alpha}^0)/n\Big]/\lambda \\
    & = \Big[ \frac{\partial}{\partial \boldsymbol{t}}p_\lambda^{\text{MCP}} (\boldsymbol{t})|_{\boldsymbol{t} = \boldsymbol{0}_{\mathcal{V}^*}}-\boldsymbol{R}^0_{\mathcal{V}^*}\Big]/\lambda.
\end{aligned}
\end{equation}
Thus, we rewrite (\ref{basic_ine_1}) as
\begin{equation}
\begin{aligned}
     -\lambda \boldsymbol{\Delta}_{\mathcal{V}^{c*}}^\top \boldsymbol{\omega}_{\mathcal{V}^{c*}}(\boldsymbol{\alpha}^0)+1/\rho \|\boldsymbol{\Delta}\|_2^2
      &\geq \boldsymbol{\Delta}^\top\boldsymbol{C}_n\boldsymbol{\Delta} + \lambda \boldsymbol{\Delta}_{\mathcal{V}^{*}}^\top \boldsymbol{\omega}_{\mathcal{V}^{*}}(\boldsymbol{\alpha}^0) \\
      & = \boldsymbol{\Delta}^\top\boldsymbol{C}_n\boldsymbol{\Delta}+  \lambda \boldsymbol{\Delta}_{\mathcal{V}^{*}}^\top \Big[ \frac{\partial}{\partial \boldsymbol{t}}p_\lambda^{\text{MCP}} (\boldsymbol{t})|_{\boldsymbol{t} = \boldsymbol{0}_{\mathcal{V}^*}}-\boldsymbol{R}^0_{\mathcal{V}^*}\Big]/\lambda \\
    & \geq  \boldsymbol{\Delta}^\top\boldsymbol{C}_n\boldsymbol{\Delta}+   \boldsymbol{\Delta}_{\mathcal{V}^{*}}^\top \Big[ \frac{\partial}{\partial \boldsymbol{t}}p_\lambda^{\text{MCP}} (\boldsymbol{t})|_{\boldsymbol{t} = \boldsymbol{0}_{\mathcal{V}^*}}\Big] -\|\boldsymbol{\Delta}_{\mathcal{V}^{*}}\|_1\|\boldsymbol{R}^0_{\mathcal{V}^*}\|_{\infty} \\
    & = \boldsymbol{\Delta}^\top\boldsymbol{C}_n\boldsymbol{\Delta}+   (\lambda-\|\boldsymbol{R}^0_{\mathcal{V}^*}\|_{\infty} )\|\boldsymbol{\Delta}_{\mathcal{V}^{*}}\|_1,
\end{aligned}
\end{equation}
where the last equality holds for a proper sub-derivative at the origin, i.e., for $j \in \mathcal{V}^*$, $\frac{\partial}{\partial {t}}p_\lambda^{\text{MCP}} ({t})|_{{t} =\alpha^0_j} =\lambda sgn(\Delta_j).$\qed
\end{proof}

Under the event $\Omega =\{ \boldsymbol{\omega}_{\mathcal{V}^{c*}}(\widehat{\boldsymbol{\alpha}}^{\text{or}}) = \boldsymbol{0}\}$, we now prove the estimation error $\boldsymbol{\Delta}$ belongs to the cone $\mathscr{C}(\mathcal{V}^* ; \xi)$, where $\xi$ is chosen in Lemma \ref{Lemma 2}. Recall $\mathscr{B}(\lambda,\rho) = \{\widehat{\boldsymbol{\alpha}} \text{ in} \,\,(\ref{form 2}):  \lambda \geq \zeta, \rho > K^{-2}_{\mathscr{C}}(\mathcal{V}^*,\xi)\lor 1 \}$ as a collection of $\widehat{\boldsymbol{\alpha}}$ computed in (\ref{form 2}) through a broad class of MCP, in which $\zeta$ is define in (\ref{margin}).

\begin{lemma}\label{Lemma A3}
Under the event $\Omega =\{ \boldsymbol{\omega}_{\mathcal{V}^{c*}}(\widehat{\boldsymbol{\alpha}}^{\text{or}}) = \boldsymbol{0}\}$, consider $\widehat{\boldsymbol{\alpha}}_{\lambda_1}, \widehat{\boldsymbol{\alpha}}_{\lambda_2} \in \mathscr{B}(\lambda,\rho)$ with different penalty levels $\lambda_1$ and  $\lambda_2$, and denote their estimation errors as $\boldsymbol{\Delta}_1 =\widehat{\boldsymbol{\alpha}}_{\lambda_1}-\widehat{\boldsymbol{\alpha}}^{\text{or}}$ and $\boldsymbol{\Delta}_2 =\widehat{\boldsymbol{\alpha}}_{\lambda_2}-\widehat{\boldsymbol{\alpha}}^{\text{or}}$, respectively. Define $a_1 =1 -  \|\boldsymbol{R^{\text{or}}}_{\mathcal{V}^{*}}\|_\infty/\lambda$, $a_2 = a_1\xi\rho/[2(\xi+1)]$,  $a_{3}=a_{1} \xi /\left(\xi+1+a_{1}\right) $, and $a_0 = a_2 \wedge \{a_2a_3/(1\lor \xi)\}$. Then, once $\boldsymbol{\Delta}_1 \in \mathscr{C}(\mathcal{V}^* ; \xi)$ and $\|\boldsymbol{\Delta}_1-\boldsymbol{\Delta}_2\|_1\leq a_0 \lambda $ hold, we conclude $\boldsymbol{\Delta}_2 \in \mathscr{C}(\mathcal{V}^* ; \xi)$.
\end{lemma}
\begin{proof}
Let $a_1 =1 -  \|\boldsymbol{R^{\text{or}}}_{\mathcal{V}^{*}}\|_\infty/\lambda > 0$. Applying Lemma \ref{Lemma A2} to $\boldsymbol{\Delta_2}$ and under the event $\Omega$, we have
\begin{equation}
\begin{aligned}
      \boldsymbol{\Delta}_2^\top\boldsymbol{C}_n\boldsymbol{\Delta}_2+   (\lambda-\|\boldsymbol{R}^{\text{or}}_{\mathcal{V}^*}\|_{\infty} )\|\boldsymbol{\Delta}_{2\mathcal{V}^{*}}\|_1  \leq   1/\rho \|\boldsymbol{\Delta}_2\|_2^2\\
      \iff \boldsymbol{\Delta}_2^\top\boldsymbol{C}_n\boldsymbol{\Delta}_2+   a_1\lambda\|\boldsymbol{\Delta}_{2\mathcal{V}^{*}}\|_1  \leq   1/\rho \|\boldsymbol{\Delta}_2\|_2^2.
\end{aligned}
\end{equation}
Consider the first case $\|\boldsymbol{\Delta}_1\|_{1} \lor \|{\boldsymbol{\Delta}_1}-\boldsymbol{\Delta}_2\|_{1} \leqslant a_{2} {\lambda}$, where $a_2 = a_1\xi\rho/[2(\xi+1)]$. We have $$\|\boldsymbol{\Delta}_2\|_2^2 \leq \|\boldsymbol{\Delta}_2\|_\infty\cdot \|\boldsymbol{\Delta}_2\|_1 \leq (\|\boldsymbol{\Delta}_2-\boldsymbol{\Delta}_1\|_1+\|\boldsymbol{\Delta}_1\|_1)\cdot \|\boldsymbol{\Delta}_2\|_1 \leq 2a_2\lambda \|\boldsymbol{\Delta}_2\|_1.$$ The above inequalities yield
\begin{equation}
    a_1\lambda\|\boldsymbol{\Delta}_{2\mathcal{V}^{*}}\|_1  \leq   1/\rho \|\boldsymbol{\Delta}_2\|_2^2 \leq \left\{a_{1} \xi /(\xi+1)\right\}\left(\left\|{\boldsymbol{\Delta}_2}_{\mathcal{V}^{c*}}\right\|_{1}+\left\|{\boldsymbol{\Delta}}_{2\mathcal{V}^{*}}\right\|_{1}\right),
\end{equation}
which is equivalent to $\boldsymbol{\Delta}_2\in \mathscr{C}(\mathcal{V}^*,\xi)$ by algebra in the first case.

Consider the second case that $\|\boldsymbol{\Delta}_1\|_1 \geq a_2\lambda$ and $\|\boldsymbol{\Delta}_1-\boldsymbol{\Delta}_2\|_1\leq \lambda a_2a_3/(1 \lor \xi)$, where $a_3 = a_1\xi/(\xi+1+a_1)$.
Similarly, applying Lemma \ref{Lemma A2} to $\boldsymbol{\Delta}_1$ and using $\boldsymbol{\Delta}_{1} \in \mathscr{C}\left(\mathcal{V}^{*} ; \xi\right)$, we obtain $$a_1\lambda \|\boldsymbol{\Delta}_{1\mathcal{V}^{*}}\|_1\leq 0 \quad \Rightarrow \quad \boldsymbol{\Delta}_{1\mathcal{V}^*} = \boldsymbol{0}.$$
The triangle inequalities $$\|\boldsymbol{\Delta}_{1\mathcal{V}^*}-\boldsymbol{\Delta}_{2\mathcal{V}^*}\|_1+\|\boldsymbol{\Delta}_{2\mathcal{V}^*}\|_1 \geq \|\boldsymbol{\Delta}_{1\mathcal{V}^*}\|_1, \quad \|\boldsymbol{\Delta}_{2\mathcal{V}^{c*}}-\boldsymbol{\Delta}_{1\mathcal{V}^{c*}}\|_1+\|\boldsymbol{\Delta}_{1\mathcal{V}^{c*}}\|_1 \geq \|\boldsymbol{\Delta}_{2\mathcal{V}^{c*}}\|_1$$
give rise to
\begin{equation}
    \begin{aligned}
        & \|\boldsymbol{\Delta}_{2\mathcal{V}^*}\|_1 - \xi \|\boldsymbol{\Delta}_{2\mathcal{V}^{c*}}\|_1\\
        \leq&  \|\boldsymbol{\Delta}_{1\mathcal{V}^*}\|_1 - \xi \|\boldsymbol{\Delta}_{1\mathcal{V}^{c*}}\|_1+ (1\lor \xi)\|\boldsymbol{\Delta}_2-\boldsymbol{\Delta}_1\|_1\\
        \leq &\|\boldsymbol{\Delta}_{1\mathcal{V}^*}\|_1 - \xi \|\boldsymbol{\Delta}_{1\mathcal{V}^{c*}}\|_1+ a_3\|\boldsymbol{\Delta}_1\|_1\\
        =&(a_3-\xi)\|\boldsymbol{\Delta}_{1\mathcal{V}^{c*}}\|_1 = \frac{-\xi(\xi+1)}{\xi+1+a_1}\|\boldsymbol{\Delta}_{1\mathcal{V}^{c*}}\|_1\leq 0,
    \end{aligned}
\end{equation}
where the second inequality follows the assumptions of $\|\boldsymbol{\Delta}_1\|_1$ and $\|\boldsymbol{\Delta}_1-\boldsymbol{\Delta}_2\|_1$.

Thus, together with above two cases, it concludes the $\boldsymbol{\Delta}_2 \in \mathscr{C}(\mathcal{V}^*,\xi)$ when $\|\boldsymbol{\Delta}_1-\boldsymbol{\Delta}_2\|_1<a_0\lambda$, where $a_0 = a_2 \wedge \{a_2a_3/(1\lor \xi)\}$  \qed
\end{proof}

Based on the above two lemmas, we are able to derive the theoretical results stated in Theorem \ref{Theorem 3}.

\begin{proof}[of Theorem \ref{Theorem 3}]
Consider the local solution $\widehat{\boldsymbol{\alpha}}$ in $\mathscr{B}_0(\lambda,\rho)$ and denote $\hat{\mathcal{V}} = \{j:\widehat{\alpha}_j = 0\}$ and event $\Phi = \{\hat{\mathcal{V}} = \mathcal{V}^*\}$ of most interest.
Thus,
\begin{equation}
    \Pr(\Phi) = \Pr(\Phi,\Omega)+\Pr(\Phi,\Omega^c)\geq \Pr(\Phi|\Omega)\Pr(\Omega).
\end{equation}
Firstly, conditional on the event $\Omega$, we denote $\boldsymbol{\Delta} = \widehat{\boldsymbol{\alpha}}-\widehat{\boldsymbol{\alpha}}^{\text{or}}$ and immediately have $\boldsymbol{\Delta}\in \mathscr{C}(\mathcal{V}^* ; \xi)$ by Lemma \ref{Lemma A3}. Applying Lemma \ref{Lemma A2} to $\boldsymbol{\Delta}$, we have
\begin{equation}
     \boldsymbol{\Delta}^\top\boldsymbol{C}_n\boldsymbol{\Delta}+   (\lambda-\|\boldsymbol{R}^{\text{or}}_{\mathcal{V}^*}\|_{\infty} )\|\boldsymbol{\Delta}_{\mathcal{V}^{*}}\|_1  \leq   -\lambda \boldsymbol{\Delta}_{\mathcal{V}^{c*}}^\top \boldsymbol{\omega}_{\mathcal{V}^{c*}}(\widehat{\boldsymbol{\alpha}}^{\text{or}})+1/\rho \|\boldsymbol{\Delta}\|_2^2. \label{B19}
\end{equation}
By Cauchy-Schwarz inequality, $$-\lambda \boldsymbol{\Delta}_{\mathcal{V}^{c*}}^\top \boldsymbol{\omega}_{\mathcal{V}^{c*}}(\boldsymbol{\alpha}^{\text{or}}) = \lambda \boldsymbol{\Delta}_{\mathcal{V}^{c*}}^\top[-\boldsymbol{\omega}_{\mathcal{V}^{c*}}(\boldsymbol{\alpha}^{\text{or}})]\leq \lambda\|\boldsymbol{\Delta}_{\mathcal{V}^{c*}}\|_2\|\boldsymbol{\omega}_{\mathcal{V}^{c*}}(\boldsymbol{\alpha}^{\text{or}})\|_2 = 0$$
follows the definition of $\Omega$. Rearranging (\ref{B19}) yields,
\begin{equation}
\begin{aligned}
          0 &\geq \boldsymbol{\Delta}^\top\boldsymbol{Q}_n\boldsymbol{\Delta} -1/\rho \|\boldsymbol{\Delta}\|_2^2 +   (\lambda-\|\boldsymbol{R}^{\text{or}}_{\mathcal{V}^*}\|_{\infty} )\|\boldsymbol{\Delta}_{\mathcal{V}^{*}}\|_1  \\
          &\geq (K^2_{\mathscr{C}}(\mathcal{V}^*,\xi)-1/\rho)\|\boldsymbol{\Delta}\|_2^2+ (\lambda-\|\boldsymbol{R}^{\text{or}}_{\mathcal{V}^*}\|_{\infty} )\|\boldsymbol{\Delta}_{\mathcal{V}^{*}}\|_1 {\geq} 0, \label{B20}
\end{aligned}
\end{equation}
where the second line follows membership of cone $\mathscr{C}(\mathcal{V}^* ; \xi)$ of $\boldsymbol{\Delta}$ and the RE condition of $\tilde{\boldsymbol{Z}}$ in Lemma \ref{Lemma 2}. Inequality (\ref{B20}) forces  $\|\boldsymbol{\Delta}_{\mathcal{V^*}}\|_1 = 0$ and $\|\boldsymbol{\Delta}\|_2^2 = 0$, i.e., ${\mathcal{V}}^* = \hat{\mathcal{V}}$, in probability because $\|\boldsymbol{R}^{\text{or}}_{\mathcal{V}^*}\|_{\infty}< \lambda $ holds with probability approaching $1$ in Lemma \ref{Lemma 3}.

Therefore, it remains to investigate event $\Omega$.
\begin{equation}
\begin{aligned}
     \boldsymbol{\omega}_{\mathcal{V}^{c*}}(\boldsymbol{\alpha}^{\text{or}}) = \Big[\frac{\partial}{\partial \boldsymbol{t}}p_\lambda^{\text{MCP}} (\boldsymbol{t})|_{\boldsymbol{t} = \boldsymbol{\alpha}^{\text{or}}_{\mathcal{V}^{c*}}} - \tilde{\boldsymbol{Z}}^\top_{\mathcal{V}^{c*}} (\boldsymbol{Y}-\tilde{\boldsymbol{Z}}\widehat{\boldsymbol{\alpha}}^{\text{or}})\Big]
      =\frac{\partial}{\partial \boldsymbol{t}}p_\lambda^{\text{MCP}} (\boldsymbol{t})|_{\boldsymbol{t} = \widehat{\boldsymbol{\alpha}}^{\text{or}}_{\mathcal{V}^{c*}}}
\end{aligned}
\end{equation}
follows definition of $\widehat{\boldsymbol{\alpha}}^{\text{or}}$. By the definition of $\Omega$ and characteristics of MCP, $$\Omega = \{ \boldsymbol{\omega}_{\mathcal{V}^{c*}}(\widehat{\boldsymbol{\alpha}}^{\text{or}})  = \frac{\partial}{\partial \boldsymbol{t}}p_\lambda^{\text{MCP}} (\boldsymbol{t})|_{\boldsymbol{t} = \widehat{\boldsymbol{\alpha}}^{\text{or}}_{\mathcal{V}^{c*}}} = \boldsymbol{0}\} =\{ {|\widehat{\boldsymbol{\alpha}}^{\text{or}}_{\mathcal{V}^{c*}}|_{\operatorname{min}} > \lambda/\rho}\}.$$

Rearranging (\ref{TSLS_bias}), which yields
\begin{equation}
 \Big\{|{\boldsymbol{\alpha}}^{*}_{\mathcal{V}^{c*}}|_{\operatorname{min}} > \lambda/\rho+\|(\boldsymbol{Z}_{\mathcal{V}^{c*}}^\top \boldsymbol{Z}_{\mathcal{V}^{c*}})^{-1}\boldsymbol{Z}_{\mathcal{V}^{c*}}^\top\boldsymbol{\epsilon}\|_\infty  +\Big\|(\boldsymbol{Z}_{\mathcal{V}^{c*}}^\top \boldsymbol{Z}_{\mathcal{V}^{c*}})^{-1}\boldsymbol{Z}_{\mathcal{V}^{c*}}^\top\boldsymbol{D} \frac{\boldsymbol{D}^\top(P_{\boldsymbol{Z}}-P_{\boldsymbol{Z}_{\mathcal{V}^{c*}}})\boldsymbol{\epsilon}}{\boldsymbol{D}^\top(P_{\boldsymbol{Z}}-P_{\boldsymbol{Z}_{\mathcal{V}^{c*}}})\boldsymbol{D}}\Big\|_\infty\Big\}   \subseteq \Omega .\label{sub-event Phi}
\end{equation}

Thus, it suffices to examine
\begin{equation*}
    \begin{aligned}
         \Big\|(\boldsymbol{Z}_{\mathcal{V}^{c*}}^\top \boldsymbol{Z}_{\mathcal{V}^{c*}})^{-1}\boldsymbol{Z}_{\mathcal{V}^{c*}}^\top\boldsymbol{D} \frac{\boldsymbol{D}^\top(P_{\boldsymbol{Z}}-P_{\boldsymbol{Z}_{\mathcal{V}^{c*}}})\boldsymbol{\epsilon}}{\boldsymbol{D}^\top(P_{\boldsymbol{Z}}-P_{\boldsymbol{Z}_{\mathcal{V}^{c*}}})\boldsymbol{D}}\Big\|_\infty
        =\Big|\frac{\boldsymbol{D}^\top(P_{\boldsymbol{Z}}-P_{\boldsymbol{Z}_{\mathcal{V}^{c*}}})\boldsymbol{\epsilon}}{\boldsymbol{D}^\top(P_{\boldsymbol{Z}}-P_{\boldsymbol{Z}_{\mathcal{V}^{c*}}})\boldsymbol{D}}\Big|\cdot \|(\boldsymbol{Z}_{\mathcal{V}^{c*}}^\top \boldsymbol{Z}_{\mathcal{V}^{c*}})^{-1}\boldsymbol{Z}_{\mathcal{V}^{c*}}^\top\boldsymbol{D}\|_\infty.
    \end{aligned}
\end{equation*}
The first term in the RHS measures the estimation error of TSLS estimator, i.e., $$\operatorname{Bias}(\hat{\beta}^{TSLS}_{or}) = \frac{\boldsymbol{D}^\top(P_{\boldsymbol{Z}}-P_{\boldsymbol{Z}_{\mathcal{V}^{c*}}})\boldsymbol{\epsilon}}{\boldsymbol{D}^\top(P_{\boldsymbol{Z}}-P_{\boldsymbol{Z}_{\mathcal{V}^{c*}}})\boldsymbol{D}} =\frac{\TTD^\top \TTE }{\TTD^\top \TTD},$$
is the unvanished term under many (weak) IVs setting.
While for the second term,
\begin{equation*}
    \begin{aligned}
        & \|(\boldsymbol{Z}_{\mathcal{V}^{c*}}^\top \boldsymbol{Z}_{\mathcal{V}^{c*}})^{-1}\boldsymbol{Z}_{\mathcal{V}^{c*}}^\top\boldsymbol{D}\|_\infty \\
         =&\|\boldsymbol{\gamma}^*_{\mathcal{V}^{c*}}+(\boldsymbol{Z}_{\mathcal{V}^{c*}}^\top \boldsymbol{Z}_{\mathcal{V}^{c*}})^{-1}\boldsymbol{Z}_{\mathcal{V}^{c*}}^\top\boldsymbol{Z}_{\mathcal{V}^*}\boldsymbol{\gamma}^*_{\mathcal{V}^*}\|_\infty+ \|(\boldsymbol{Z}_{\mathcal{V}^{c*}}^\top \boldsymbol{Z}_{\mathcal{V}^{c*}})^{-1}\boldsymbol{Z}_{\mathcal{V}^{c*}}^\top\boldsymbol{\eta}\|_\infty\\
         =& \|\bar{\boldsymbol{\gamma}}^*_{\mathcal{V}^{c*}}\|_\infty+ O_p\Big(\sigma_\eta\sqrt{\frac{2\log(2{p_{\mathcal{V}^{c*}}})}{n}}\Big),
    \end{aligned}
\end{equation*}
where $\bar{\boldsymbol{\gamma}}^*_{\mathcal{V}^{c*}} = \boldsymbol{\gamma}^*_{\mathcal{V}^{c*}}+(\boldsymbol{Z}_{\mathcal{V}^{c*}}^\top \boldsymbol{Z}_{\mathcal{V}^{c*}})^{-1}\boldsymbol{Z}_{\mathcal{V}^{c*}}^\top\boldsymbol{Z}_{\mathcal{V}^*}\boldsymbol{\gamma}^*_{\mathcal{V}^*}$.

Thereby, (\ref{sub-event Phi}) reduces to
\begin{equation*}
    \Bigg\{|{\boldsymbol{\alpha}}^{*}_{\mathcal{V}^{c*}}|_{\operatorname{min}} > \lambda/\rho+O_p\Big(\sigma_\epsilon\sqrt{\frac{2\log(2{p_{\mathcal{V}^{c*}}})}{n}}\Big) +|\operatorname{Bias}(\hat{\beta}^{TSLS}_{or})| \cdot \Big[\|\bar{\boldsymbol{\gamma}}^*_{\mathcal{V}^{c*}}\|_\infty+ O_p\Big(\sigma_\eta\sqrt{\frac{2\log(2{p_{\mathcal{V}^{c*}}})}{n}}\Big)\Big]\Bigg\}   \subseteq \Omega .
\end{equation*}

Combining with $\lambda > \zeta$, we now specify \eqref{sub-event Phi} as
\begin{equation*}
     \Big\{|{\boldsymbol{\alpha}}^{*}_{\mathcal{V}^{c*}}|_{\operatorname{min}} >C \kappa(n)\Big\}   \subseteq \Omega,
\end{equation*}
where $\kappa(n) = {\sqrt{\frac{\log p_{\mathcal{V}^*}}{n}}} + {\frac{p_{\mathcal{V}^*}}{n}\cdot\frac{\|\Tilde{\Tilde{\boldsymbol{Q}}}_n\boldsymbol{\gamma}_{\mathcal{V}^*}^*\|_{\infty}}{\boldsymbol{\gamma}_{\mathcal{V}^*}^{*\top}\Tilde{\Tilde{\boldsymbol{Q}}}_n\boldsymbol{\gamma}_{\mathcal{V}^*}^*}} + {|\operatorname{Bias}(\hat{\beta}_{\text{or}}^{\text{TSLS}})|\cdot \|\bar{\boldsymbol{\gamma}}^*_{\mathcal{V}^{c*}}\|_\infty}$.
Thus, under the condition that event $\Omega$ holds in the finite sample or in probability, we achieve consistency of the selection of valid IVs.

We then turn to $\mathcal{P}_c$. Recall that $\tilde{\boldsymbol{\epsilon}}^c = \boldsymbol{\epsilon}-c\boldsymbol{\eta}$, so $\sqrt{\operatorname{Var}(\tilde{\boldsymbol{\epsilon}}^c)} = \sqrt{\sigma^2_\epsilon+c^2\sigma^2_\eta-2c\sigma_{\epsilon,\eta}}= 1+c$ and $\operatorname{Cov}(\tilde{\boldsymbol{\epsilon}}^c,\boldsymbol{\eta}) = \sigma_\epsilon^2-c\sigma^2_\eta = 1+c$. After some similar derivations, we obtain
\begin{equation}
    \kappa^c(n) \asymp (1+c){\sqrt{\frac{\log |\mathcal{I}_c|}{n}}} +(1+c) {\frac{|\mathcal{I}_c|}{n}\cdot\frac{\|\Tilde{\Tilde{\boldsymbol{Q}}}^c_n\boldsymbol{\gamma}_{\mathcal{I}_c}^*\|_{\infty}}{\boldsymbol{\gamma}_{\mathcal{I}_c}^{*\top}\Tilde{\Tilde{\boldsymbol{Q}}}^c_n\boldsymbol{\gamma}_{\mathcal{I}_c}^*}} + {|\operatorname{Bias}(\hat{\beta}_{\text{or}}^{c~\text{TSLS}})|\cdot \|\bar{\boldsymbol{\gamma}}^*_{\mathcal{I}^c_c}\|_\infty},
\end{equation}
where $\Tilde{\Tilde{\boldsymbol{Q}}}_n^c$ and $\operatorname{Bias}(\hat{\beta}_{\text{or}}^{c~\text{TSLS}})$ are defined as $\mathcal{P}_c$ version of  $\Tilde{\Tilde{\boldsymbol{Q}}}_n$ and $\operatorname{Bias}(\hat{\beta}_{\text{or}}^{\text{TSLS}})$. With similar argument, we also achieve the the consistency of selection of valid counterparts in $\mathcal{P}_c$ if $|\tilde{\alpha}_{j}^{c}|>\kappa^{c}(n)$ for $j \in\{j: \alpha_{j}^{*} / \gamma_{j}^{*}=\tilde{c} \neq c\}$ holds.

\qed
\end{proof}

\subsection{Proof of Proposition \ref{Proposition 2}}
 \begin{proof}
 \begin{equation*}
     T_2 = \frac{p_{\mathcal{V}^*}}{n}\cdot\frac{\|\Tilde{\Tilde{\boldsymbol{Q}}}_n\boldsymbol{\gamma}_{\mathcal{V}^*}^*\|_{\infty}}{\boldsymbol{\gamma}_{\mathcal{V}^*}^{*\top}\Tilde{\Tilde{\boldsymbol{Q}}}_n\boldsymbol{\gamma}_{\mathcal{V}^*}^*}\leq \frac{p_{\mathcal{V}^*}}{n}\cdot\frac{\|\Tilde{\Tilde{\boldsymbol{Q}}}_n^{1/2}\|_\infty\|\Tilde{\Tilde{\boldsymbol{Q}}}_n^{1/2}\boldsymbol{\gamma}_{\mathcal{V}^*}^*\|_\infty}{\|\Tilde{\Tilde{\boldsymbol{Q}}}_n^{1/2}\boldsymbol{\gamma}_{\mathcal{V}^*}^*\|_2^2}\leq \frac{p_{\mathcal{V}}}{n}\cdot\frac{p_{\mathcal{V}^*}\|\Tilde{\Tilde{\boldsymbol{Q}}}_n^{1/2}\|_\infty\|\Tilde{\Tilde{\boldsymbol{Q}}}_n^{1/2}\boldsymbol{\gamma}_{\mathcal{V}^*}^*\|_\infty}{\|\Tilde{\Tilde{\boldsymbol{Q}}}_n^{1/2}\boldsymbol{\gamma}_{\mathcal{V}^*}^*\|_1^2}\rightarrow 0
 \end{equation*}\qed
 \end{proof}

\subsection{Proof of Proposition \ref{Proposition 3}}
This proof is extended from \cite{bun2011comparison}'s higher order approximation arguments.
\begin{proof}
Recall
\begin{equation}
\operatorname{Bias}(\hat{\beta}^{TSLS}_{or}) = \frac{\boldsymbol{D}^\top(P_{\boldsymbol{Z}}-P_{\boldsymbol{Z}_{\mathcal{V}^{c*}}})\boldsymbol{\epsilon}}{\boldsymbol{D}^\top(P_{\boldsymbol{Z}}-P_{\boldsymbol{Z}_{\mathcal{V}^{c*}}})\boldsymbol{D}} = \frac{\TTD^\top \TTE }{\TTD^\top \TTD} \coloneqq \frac{c}{d},
\end{equation}
we have $\bar{c} \coloneqq E(c) = \sigma^2_{\epsilon,\eta} (p-{p_{\mathcal{V}^{c*}}}) = \sigma^2_{\epsilon,\eta}{p_{\mathcal{V}^*}}$ and $\bar{d} \coloneqq E(d) = \sigma^2_\eta(\mu_n+L)$. That is  free of the number of invalid IVs ${p_{\mathcal{V}^{c*}}}$.
Let $s = \operatorname{max}(\mu_n,{p_{\mathcal{V}^*}})$, we have
\begin{equation}
    \begin{aligned}
        \operatorname{Bias}(\hat{\beta}^{\text{TSLS}}_{\text{or}}) &= \frac{\bar{c}}{\bar{d}}+\frac{c-\bar{c}}{\bar{d}}-\frac{\bar{c}(d-\bar{d})}{\bar{d}^{2}}-\frac{(c-\bar{c})(d-\bar{d})}{\bar{d}^{2}}+\frac{\bar{c}(d-\bar{d})^{2}}{\bar{d}^{3}}+O_{p}\left(s^{-\frac{3}{2}}\right)\\
        E[\operatorname{Bias}(\hat{\beta}^{\text{TSLS}}_{\text{or}})] &= \frac{\sigma_{\epsilon,\eta}}{\sigma_{\eta}^{2}}\left(\frac{{p_{\mathcal{V}^*}}}{(\mu_n+{p_{\mathcal{V}^*}})}-\frac{2 \mu_n^{2}}{(\mu_n+{p_{\mathcal{V}^*}})^{3}}\right)+o\left(s^{-1}\right)
    \end{aligned}
\end{equation}
follows from Section 3 in \cite{bun2011comparison}.\qed
\end{proof}


\subsection{Proof of Theorem \ref{theorem 4}}
\begin{proof}
Under the conditions of Theorem \ref{Theorem 3}, we have $\Pr(\hat{\mathcal{V}} = \mathcal{V}^{*})\overset{p}{\rightarrow} 1$. Thus, $\hat{\beta}^{\text{WIT}}\overset{p}{\rightarrow} \hat{\beta}_{\text{or}}^{\text{liml}}$, where $\hat{\beta}_{\text{or}}^{\text{liml}}$ stands for LIML estimator  with known $\mathcal{V}^*$ a priori.
Thus, (a) follows Corollary 1(iv) in \citep{kolesar2015identification} with $\min \operatorname{eig}\left(\Sigma^{-1} \Lambda\right) = 0$ in their context.
(b) and (c) follow \cite{kolesar2018minimum}, Proposition 1.\qed

\end{proof}


\subsection{Proof of Corollary \ref{Corollary 2}}

\begin{proof}
Notice ${p_{\mathcal{V}^*}}/n<p/n\rightarrow 0$ and $\mu_n/n\overset{p}{\rightarrow}\mu_0$, therefore the threshold $T_2\rightarrow0$ in Theorem \ref{Theorem 3}. Likewise $T_3 \rightarrow0$ follows $\operatorname{Bias}(\hat{\beta}^{\text{TSLS}}_{\text{or}})\overset{p}{\rightarrow}\frac{\sigma_{\epsilon \eta}}{\sigma_{\eta}^{2}}\left(\frac{{p_{\mathcal{V}^*}}}{\left(\mu_{n}+{p_{\mathcal{V}^*}}\right)}-\frac{2 \mu_{n}^{2}}{\left(\mu_{n}+{p_{\mathcal{V}^*}}\right)^{3}}\right)=o(1)$. Thus, $\kappa(n) \asymp n^{-1/2}$ in \eqref{rate of kappa} diminishes to $0$. Thus any fixed $\min _{\mathrm{j} \in \mathcal{V}^{c^{*}}} \alpha_{j}^{*}$ would pass $\kappa(n)$ asymptotically.  Let $c = \alpha_l^*/\gamma^*_l = C_1/n^{-\tau_1}$ for $l\in\mathcal{I}_c$.  Then, for $\kappa^c(n)$, we have: $|E \operatorname{Bias}(\hat{\beta}_{\text{or}}^{c~\text{TSLS}})|\approx \frac{\operatorname{Cov}(\tilde{\boldsymbol{\epsilon}}^c,\boldsymbol{\eta})}{{\operatorname{Var}(\boldsymbol{\boldsymbol{\eta}})}}\cdot C_1^{-2} n^{2\tau_1-1} \asymp n^{3\tau_1-1}$ according to Proposition \ref{Proposition 3} and  $\|\bar{\boldsymbol{\gamma}}^*_{\mathcal{I}^c_c}\|_\infty\leq C$ due to Assumption 4. Then, for the first two terms in $\kappa^c(n)$:
    \begin{equation*}
        (1+c){\sqrt{\frac{\log |\mathcal{I}_c|}{n}}} +(1+c) {\frac{|\mathcal{I}_c|}{n}\cdot\frac{\|\Tilde{\Tilde{\boldsymbol{Q}}}^c_n\boldsymbol{\gamma}_{\mathcal{I}_c}^*\|_{\infty}}{\boldsymbol{\gamma}_{\mathcal{I}_c}^{*\top}\Tilde{\Tilde{\boldsymbol{Q}}}^c_n\boldsymbol{\gamma}_{\mathcal{I}_c}^*}} \asymp n^{\tau_1-1/2}+n^{\tau_1-1}.
    \end{equation*}
    Hence, we conclude $\kappa^c(n) \asymp n^{\operatorname{max}(\tau_1-1/2, 3\tau_1-1)}$.
For $|\tilde{\alpha}^{c}_j| = |\alpha_j^*-c\gamma_j^*|$ and $j \in\left\{j: \alpha_{j}^{*} / \gamma_{j}^{*}=\tilde{c} \neq c\right\}$. We consider all possible cases:
\begin{enumerate}
    \item $\tilde{c} = 0$, i.e.\ $j\in \{\mathcal{I}^c_c: \alpha^*_j = 0\}$: $|\tilde{\alpha}^{c}_j| = |c|\cdot|\gamma_j^*|$, and $\gamma^*_j = n^{-\tau_2}$. Hence $|\tilde{\alpha}^{c}_j| \asymp n^{\tau_1-\tau_2}$.
    \item $\tilde{c}\neq 0$, i.e.\ $j\in \{\mathcal{I}^c_c: \alpha^*_j/\gamma^*_j = \tilde{c}\}$: $|\tilde{\alpha}^{c}_j| = |c-\tilde{c}|\cdot|\gamma_j^*|$ and $\gamma^*_j = n^{-\tau_3}$. Hence, $|\tilde{\alpha}^{c}_j| \asymp |C_1-C_3n^{\tau_1-\tau_3}|$
\end{enumerate}
Thus, $|\tilde{\alpha}_j^c|>\kappa^c(n)$ is equivalent to $2\tau_1+\tau_2<1$ and $2\tau_1+\tau_3<1$. By the symmetry of $\tau_1$ and $\tau_3$ of invalid IVs, we have $2\tau_3+\tau_1<1$. Hence $\tau_1+\tau_3<2/3$.
Therefore, Assumption 5 holds automatically. Then it follows Theorem \ref{theorem 4}.\qed
\end{proof}

\subsection{Proof of Proposition \ref{Proposition lowdim}}

\begin{proof}
Same with Proof of Corollary \ref{Corollary 2}, $T_2$ and $T_3$ are o(1) and $\kappa(n)\asymp n^{-1/2}$.  So assumption ${min}_{\operatorname{j\in \mathcal{V}^{c*}}}\, |{\alpha}^*_j| >O(n^{-1/2})$ leads first part in Assumption 5 holds. 

Sequentially, we check if the second part of assumption 5 holds. 
Consider threshold $\kappa^c(n)$ in \eqref{rate of kappa^c}, where $c$ corresponds to non-zero value of ratio $\alpha^*_j/\gamma^*_j = c$ for any $j \in \mathcal{V}^{c*}$:
$$ \kappa^c(n) \asymp (1+c)\Big\{{\sqrt{\frac{\log |\mathcal{I}_c|}{n}}} + {\frac{|\mathcal{I}_c|}{n}\cdot\frac{\|\Tilde{\Tilde{\boldsymbol{Q}}}^c_n\boldsymbol{\gamma}_{\mathcal{I}_c}^*\|_{\infty}}{\boldsymbol{\gamma}_{\mathcal{I}_c}^{*\top}\Tilde{\Tilde{\boldsymbol{Q}}}^c_n\boldsymbol{\gamma}_{\mathcal{I}_c}^*}}\Big\} + {|\operatorname{Bias}(\hat{\beta}_{\text{or}}^{c~\text{TSLS}})|\|\bar{\boldsymbol{\gamma}}^*_{\mathcal{I}^c_c}\|_\infty}.$$
Because Assumption 4 holds each DGP and finite $p$. We have $|\operatorname{Bias}(\hat{\beta}_{\text{or}}^{c~\text{TSLS}})| = o(1/n)$ \citep{bun2011comparison} and $\frac{|\mathcal{I}_c|}{n} = O(1/n)$. Thus, 
$$\kappa^c(n)\lesssim (1+c) \big\{ \frac{1}{\sqrt{n}}+\frac{1}{n}\cdot1\big\} +\frac{1}{n} \asymp (1\wedge c)\frac{1}{\sqrt{n}} = O(\frac{1}{\sqrt{n}})$$
since $c = \frac{\alpha_j}{\gamma_j^*} = o(1)$.

Then we turn to transformed $\tilde{\boldsymbol{\alpha}}^c$ in \eqref{GAP} for certain non-zero $c = \alpha^*_j/\gamma^*_j$ with some $j = 1,2,\ldots, p $. Then we consider the remaining set $\{j:\alpha_l^*/\gamma_l^* = \tilde{c}\neq c\}$, where $\tilde{c}$ can be zero. Then we obtain the separation condition: 
$$|\tilde{\alpha}_l^c| = |\alpha_l^* - c\gamma_l^*| = |\alpha_l^* - \frac{\alpha_j^*}{\gamma_j^*}\cdot \gamma_l^*| >O(\frac{1}{\sqrt{n}}).$$

Hence, it has two cases: 
\begin{enumerate}
    \item $\tilde{c} = 0 \Rightarrow \alpha_l^* = 0$. It leads to $|\frac{\alpha_j^*}{\gamma_j^*}\cdot \gamma_l^*|>O(\frac{1}{\sqrt{n}})$. Thus, it is equivalent to 
    $$\frac{\alpha^*_j}{\gamma^*_j}>O(\frac{n^{-1/2}}{\gamma_l^*}),$$ which holds due to $\gamma_l^*>\underset{j\in \mathcal{V}^{c*}}{\operatorname{max}}O(n^{-1/2} \gamma_j^* /\alpha_j^*)$ for valid $l \in \mathcal{V}^*$.
    \item $\tilde{c}\neq c$ and $\tilde{c}\neq 0$. It yields 
    $$|\alpha_l^* - \frac{\alpha_j^*}{\gamma_j^*}\cdot \gamma_l^*| \asymp |\alpha_l^*| \wedge (|\frac{\alpha_j^*}{\gamma_j^*}|\cdot |\gamma_l^*|)\lesssim |\alpha_l^*|\wedge |\gamma_l^*| >O(1/\sqrt{n})$$ holds as assumption in True DGP.
\end{enumerate} 

\end{proof}
\end{document}